\documentclass[article, prx, longbibliography, superscriptaddress, twocolumn]{revtex4-2}
\usepackage{wrapfig}
\usepackage{amsmath,amssymb,amstext, amsthm}
\usepackage{bbm}
\usepackage{physics}
\usepackage{float}
\usepackage{braket}
\usepackage{graphicx}
\usepackage{xcolor}
\usepackage{hyperref}
\usepackage{subcaption}
\usepackage[shortlabels]{enumitem}
\def\bra#1{\langle#1|} \def\ket#1{|#1\rangle}
\def\braket#1#2{\langle#1|#2\rangle}

\def\proj#1{\ket{#1}\!\bra{#1}}
\def\id{{\mathbb I}}

\def\q0{\underline{0}}

\def\H{{\cal H}}
\def\Z{\mathbb{Z}}
\def\spec{\mbox{spec}}

\def\C{{\mathbb C}}
\def\id{{\mathbb I}}
\def\E{\vec{E}}

\def\H{{\cal H}}

\def\U{{ U}}

\def\R{\mathbb{R}}

\def\D{{\cal D}}
\def\N{\mathbb{N}}

\def\tr{\mbox{tr}}

\def\sign{\mbox{sign}}

\newtheorem{theorem}{Theorem}
\newtheorem{remark}[theorem]{Remark}

\newtheorem{lemma}[theorem]{Lemma}
\newtheorem{prop}[theorem]{Proposition}

\newtheorem{defin}{Definition}

\begin{document}

\title{Timed demolition measurements}
\author{Konstantinos Manos}
\affiliation{Institute for Quantum Optics and Quantum Information (IQOQI) Vienna,\\
 Austrian Academy of Sciences, Boltzmanngasse 3, Wien 1090, Austria}
\author{Mirjam Weilenmann} 
\affiliation{Inria Saclay, Télécom Paris -- LTCI, Institut Polytechnique de Paris,\\
19 Place Marguerite Perey, 91120 Palaiseau, France}
\affiliation{Department of Applied Physics, University of Geneva, Switzerland}
\author{Miguel Navascu\'es}
\affiliation{Institute for Quantum Optics and Quantum Information (IQOQI) Vienna,\\
 Austrian Academy of Sciences, Boltzmanngasse 3, Wien 1090, Austria}

 \begin{abstract}
Picture an experimental scenario where a closed quantum system, evolving through a time-independent Hamiltonian, is subject to a demolition measurement at a chosen time. The Hamiltonian, the measured observables, the initial state of the physical system and even its Hilbert space dimension are unknown; we nonetheless assume a promise or constraint on the energy distribution of the state. In this context we find that, for many natural energy constraints, the set of feasible time series or datasets can be characterized efficiently. Furthermore, under the assumption of a bounded energy spectrum, we prove that there exist ``self-testing'' datasets, whose approximate realization singles out specific Hamiltonians, states and measurement operators. Investigating to what extent the extrapolation of past measurement data is possible in this framework, we identify energy-constrained physical systems for which a non-trivial prediction at time $\tau$ requires a precision in the measurement data superexponential in $\tau$. We also discover two extrapolation phenomena: ``aha! datasets", which drastically increase the predictability of the future statistics of an unrelated measurement; and ``fog banks": fairly simple datasets that exhibit complete unpredictability at some future time $\tau$, but full predictability at a later time $\tau'>\tau$. Besides their relevance for quantum foundations, our results have applications in semi-device independent quantum communication, the simulation of complex quantum systems, and the design of optimal atomic clocks.

\end{abstract}

 \maketitle
 \newpage

\section{Introduction}
The study of the correlations observed by distant parties conducting experiments in a space-like separated way is a well-established subject in quantum information science (QIS) \cite{ScaraniBellNonlocality2019}. Research on such ``space-like quantum correlations'' is justified on the grounds that the security or soundness of many quantum information processing protocols can be certified device-independently. In this regard, space-like quantum correlations alone suffice to certify the randomness of the outcomes of a quantum measurement \cite{Pironio2009, colbeck2009} and generate shared secret keys between separate parties \cite{DI_QKD,DI_QKD2}. Results on device-independent QIS often rely on partial characterizations of the set of quantum space-like correlations through hierarchies of SDPs \cite{npa, npa2, Berta_2016}, and on the phenomenon of \emph{self-testing}, the property of certain space-like correlations to single out the inner workings of the preparation and measurement devices used to produce them \cite{self_testing_review}.

The majority of the experiments in quantum mechanics, though, do not involve multiple parties working in separate labs. On the contrary, at a high level, a standard quantum experiment can be described as the preparation of a quantum system in a single lab, followed by a demolition measurement. This description encompasses all experiments in high-energy physics, as well as most experiments in quantum optics and solid state physics. In many such experiments, one can further choose \emph{when} to carry out the measurement. This allows one to study the natural dynamics of the underlying physical system, by exploring how the measurement statistics change with time. 

Quantum speed limits relate the energy content of a quantum system with the speed with which it can change its state, see, e.g., \cite{speed_limits1,speed_limits2,speed_limits3,margolus_levitin, energy_max_bound}. Quantum states are, however, not directly observable: what one estimates in the lab are measurement statistics. Notwithstanding, a framework to study how measurement statistics change as a function of time, complementary to the framework of quantum nonlocality, is missing. In such a framework, one could explore time-like analogs of basic problems in quantum nonlocality, such as: how does one characterize measurement statistics in time? Are some time series self-testing, i.e., do they pinpoint the Hilbert space of the probed system, the preparation and measurement apparatuses and its time-independent Hamiltonian?

In this paper, we solve the first problem under different constraints on the energy distribution of the underlying quantum state. We find that, if the spectrum of the system Hamiltonian is finite and known, then the problem of characterizing timed correlations can be cast as a semidefinite program (SDP) \cite{sdp}, i.e., a class of convex optimization problems that can be solved efficiently. Alternative energy constraints, such as having a bounded Hamiltonian $0\leq H\leq E^+$, having a bounded average energy or a bounded energy variance can be characterized through complete hierarchies of SDP relaxations, for which we also derive convergence bounds. As a result, the corresponding sets of correlations can be characterized up to arbitrary precision with reasonable computational resources. 

With regards to the second problem, we identify families of datasets that, under the constraint $0\leq H\leq E^+$, are robustly \emph{self-testing}: in analogy with self-testing distributions in quantum nonlocality \cite{self_testing_review}, their approximate realization in the lab would single out the underlying initial state, Hamiltonian and measurement operators. More specifically, for any $N\in\N$, $N\geq 2$, we find an $N$-point dataset that self-tests an $N$-dimensional quantum system.

We then turn to a task that has no analog in the framework of space-like correlations: the extrapolation of past measurement data. If the characterization of time-like quantum correlations answers the question ``how do things change?'', the extrapolation problem raises the question ``what is going to happen?''. On this matter, we show that the extrapolation problem can be solved under all energy constraints listed above through our characterization of the corresponding sets of correlations. While the behavior of some physical systems can be extrapolated to very long times $\tau$ with just a few datapoints and reasonably high statistical noise, for some others the magnitude of the noise must be below $1/g(\tau)$, with $g$ being a superexponential function, in order to make a non-trivial prediction, even in the limit of infinitely many datapoints.

In addition, we discover two unexpected extrapolation phenomena. 

The first one, which we dub the ``aha! effect", refers to a situation in which the dataset generated by conducting measurement $1$ on a quantum system at specific times does not allow us to say anything at all about the measurement statistics at some future time $\tau$. However, if we also considers the corresponding dataset generated by a second, independent measurement $2$, then it is possible to predict the future outcome distribution of measurement $1$ with certainty. 

The second effect, which we call ``fog bank", describes an extrapolation scenario where the system's behavior at some future time $\tau$ is completely unpredictable, and yet it becomes fully predictable at some later time $\tau'>\tau$. 

We find that both effects can be observed in datasets generated by measuring closed, low-dimensional quantum systems.

It is worth mentioning that, throughout this paper, we assume that the system under observation is closed, namely, that its dynamics are governed by some unknown, time-independent Hamiltonian. All our results are nevertheless robust against weak interactions with an arbitrarily energetic environment.

Besides its fundamental interest, our work has several applications in different fields:
\begin{enumerate}
    \item \textbf{Device independent energy witnesses.}
    Through our hierarchies of SDP relaxations, it is possible to certify, solely from its dynamical statistical behavior, that certain energy assumptions on an untrusted device do not hold. For instance, given an experimental dataset, obtained by probing a complex quantum system, we might certify that said system has average energy greater than $\bar{E}$. This could play a role in future experiments that aim to clarify the quantum nature of gravity \cite{BMM+17, MV17}. Note that the measurement statistics of classical systems are not limited by energy considerations: if models of classical gravity (see, e.g., \cite{TD16, oppenheim}) approximately describe the actual dynamics of quantum massive systems, it might be possible to devise an experiment to disprove the quantumness of gravity under energy assumptions.

    \item \textbf{Design of optimal clocks.} Our characterization of the set of timelines generated by a Hamiltonian of known spectrum allows identifying the physical setup that best approximates  arbitrary ``ideal'' measurement statistics, such as those of a perfect clock. That is: given the available energy levels $E_1,...,E_n$ of the system, we can efficiently solve the problem of determining what state to prepare and what measurement to conduct on the system to optimally approximate (in $2$-norm) the ``clock'' timeline $P_c(a|t)=\Theta(a-t)\Theta(t+1-a)$, for $a=0,1,2,...$. 

    \item \textbf{Semi-device independent randomness certification.} In \cite{jones2024, jones2025certifiedrandomnessquantumspeed} it is shown how to certify and quantify the randomness of the outcomes of a quantum measurement under different energy assumptions on the underlying quantum system. Both results rely on the characterization of measurement statistics in timed demolition scenarios with one measurement of two possible outcomes and two different measurement times. The techniques used, based on quantum speed limits \cite{speed_limits1} and the characterization of correlations generated by probing two different states of bounded overlap \cite{Van-Himbeeck_2019}, do not generalize to scenarios with more than two measurement times. In our work, we show that the set of feasible datasets under the energy restriction considered in \cite{jones2024} can be characterized in scenarios with arbitrarily many measurements, outcomes and measurement times with a single SDP. Similarly, through hierarchies of SDP relaxations, it is possible to characterize, up to the desired precision, the set of datasets in arbitrary timed demolition scenarios under the variance energy constraint invoked in \cite{jones2025certifiedrandomnessquantumspeed}. Our work therefore allows extending the randomness certification protocols in \cite{jones2024, jones2025certifiedrandomnessquantumspeed} to more complicated measurement scenarios, with the potential to certify more randomness under less demanding experimental requirements.

    \item \textbf{Semi-device independent quantum key distribution.} The randomness certification results of \cite{jones2024, jones2025certifiedrandomnessquantumspeed} suggest that timed demolition scenarios could similarly be useful for semi-device independent quantum key distribution (QKD) \cite{semi_QKD, semi_QKD2}, in which case our characterization of sets of feasible datasets will be needed. Think of a prepare-and-measure QKD protocol where we take for granted a constraint on the energy of the transmitted quantum system. At the receiving site, there is just one measurement setting, but one can choose when to conduct said measurement. In this scenario, suitable measurement statistics would suffice to certify secrecy, without further assumptions on the preparation and measurement devices, or on the quantum system itself. 

    \item \textbf{Dynamics of complex quantum systems.} Computing the exact evolution of many-body or continuous variable quantum systems is an intractable problem \cite{hamil_intract1, hamil_intract2}. Some methods allow, however, determining the average value of an observable at different times, with rigorous uncertainty bounds \cite{MPS_evolution, CV_evolution, differential_paper}. Sometimes the uncertainty bounds explode, e.g., when the state of the many-body system can no longer be approximated by a matrix product state \cite{MPS} of low bond dimension \cite{MPS_evolution}, or when the continuous variable state cannot be approximated by a short expansion into number states \cite{CV_evolution}. It is worth remarking that, for many such systems, the average energy of the initial state can be computed exactly. As we show in section \ref{sec:extrapolation}, in such circumstances our general extrapolation techniques allow estimating the future values of any bounded observable beyond the breaking point of the numerical method.

\end{enumerate}
Since this paper aims to properly introduce the timed measurements framework, these applications will be sketched rather than fully worked out, with one exception: in Section \ref{sec:clocks} we will compute the optimal clock signals achievable by manipulating the first levels of standard atomic spectra.

The structure of this paper is as follows: in Section \ref{sec:timed_demolition}, we will introduce the framework of timed demolition measurements, which we will use to formulate all our subsequent results. We will motivate a number of different energy constraints and point to experimental platforms where such constraints faithfully represent our knowledge of the system. We will also discuss our central assumption of handling a closed quantum system and explain how our results still apply (with minor modifications) when the system interacts with its environment weakly. In Section \ref{sec:charact} we will introduce a general method to derive SDP relaxations for sets of feasible time series under arbitrary energy constraints. We will apply this method to formulate complete hierarchies of SDP relaxations for the constraints listed above. In section \ref{sec:self_testing} we will define robust self-testing in the timed demolition measurements scenario and present a family of self-testing datasets. In section \ref{sec:extrapolation}, we will introduce the extrapolation problem, explain how to solve it and prove the existence of Aha! datasets and fog banks. We will present our final conclusions in section \ref{sec:conclusions}.

\section{Timed demolition measurements}
\label{sec:timed_demolition}
\subsection{Basic framework}
\label{sec:framework}
Consider a closed quantum system, initially in state $\rho\in B(\H)$, for some Hilbert space $\H$, evolving through the action of a Hamiltonian $H$. Since we are free to choose the origin of energies, we will frequently assume that $H\geq 0$. At time $t$, we conduct a demolition measurement labeled by $x\in\{1,...,X\}$, obtaining an outcome $a\in\{1,...,A\}$ with some probability $P(a|x,t)$. Measurement $M_x$ can be described by a Positive Operator Valued Measure (POVM), i.e., a tuple of positive semidefinite operators $(M_{a|x})_a\subset B(\H)^A$, with $\sum_aM_{a|x}=\id_\H$. With this notation, it holds that
\begin{equation}
P(a|x,t)=\tr(e^{-iHt}\rho e^{iHt}M_{a|x}),\forall a,x,t.
\end{equation}
For fixed $t$, we will call the vector of probabilities $(P(a|x,t))_{a,x}$ a \emph{datapoint}. The set of all datapoints will be denoted by ${\cal D}$. For finite ${\cal T}\subset \R$, any function $P:{\cal T}\to{\cal D}$, with $P(t):=(P(a|x,t))_{a,x}$, for all $t\in{\cal T}$, will be called a \emph{dataset} or a \emph{time series}. If we allow the domain of $P$ to include all $\R$, we will speak of a \emph{timeline}. Contrarily to datasets, timelines cannot be estimated through a finite number of experiments. Both for timelines and datasets, we call the tuple $(\H,\rho,H,M)$ a \emph{realization}, where $M$ denotes the collection of measurements $M_x$. Figure~\ref{fig:setup} illustrates the setup for collecting datasets.
\begin{figure}
     \centering  
   \includegraphics[trim={4cm 6cm 4cm 0.5cm},clip,width=\linewidth]{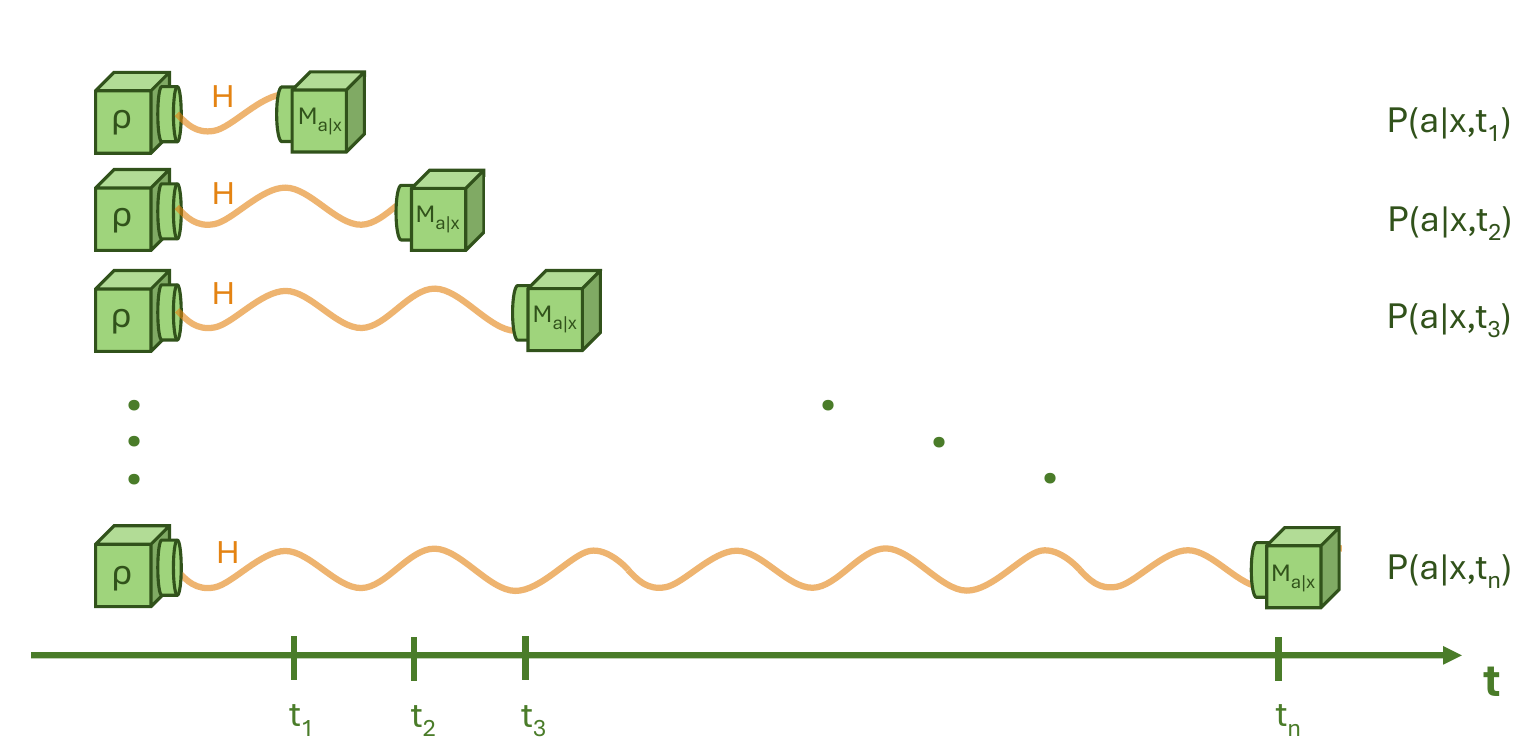}
     \caption{\textbf{Datasets.} A datapoint $(P(a|x,t_i))_{a,x}$ characterises the outcome probabilities when a state $\rho$ is prepared, evolved for a time $t_i$ with Hamiltonian $H$ and measured with one of the measurements $M_x$. For the different datapoints making up a dataset only the time the system evolves for differs (shown in different rows in the figure).
     In order to obtain an estimate $(\tilde{P}(a|x,t_i))_{a,x}$ for the datapoint $(P(a|x,t_i))_{a,x}$, one may prepare many copies of the state $\rho$, evolve each with Hamiltonian $H$ for a time $t_i$ and perform one of the different measurements $M_x$ on it. }
     \label{fig:setup}
\end{figure}

In practice, through finitely many experiments, we can only hope to arrive at \emph{noisy datasets}, consisting of the probability distribution estimates $(\tilde{P}(a|x,t_j))_{a,x,j}$, together with a list of error bars $(\delta(x,j))_{x,j}$. We say that a dataset or timeline $P$ \emph{fits} the noisy dataset $(\tilde{P}, \delta)$, defined over times $\vec{t}=(t_1,..,t_n)$, it if it holds that
\begin{equation}
\sum_a|P(a|x,t_j)-\tilde{P}(a|x,t_j)|\leq \delta(x,j).
\label{dist_fixed_x}
\end{equation}
for $j=1,...,n$. For any $\delta\in\R^+$, we will use the expression $(\tilde{P},\delta)$ to mean the noisy dataset with error bars $\delta(x,j)=\delta$, for all $x,j$.

Connected to the previous equation, we next define a notion of distance between datapoints:
\begin{equation}
\|P(t)-\tilde{P}(t)\|:=\max_{x}\sum_a|P(a|x,t)-\tilde{P}(a|x,t)|,
\end{equation}
To understand the intuition behind this distance, think of a device that, with probability $\frac{1}{2}$ has measurement statistics described by either $P$ or $\tilde{P}$. If we have just one use of the device at time $t$, then the maximum probability of correctly identifying the datapoint is
\begin{equation}
\frac{1}{2}+\frac{\|P(t)-\tilde{P}(t)\|}{4}.
\end{equation}
Correspondingly, eq.~\eqref{dist_fixed_x} is related to the probability of distinguishing the two datapoints when were are obliged to measure $x$.

The closure of the set of datasets $P$ admitting a quantum realization is only limited by the requirement that, for each $x$ and all $t\in{\cal T}$, $P(a|x,t)$ is a probability distribution in $a$. Indeed, let $\vec{t}=(t_1,...,t_N)$, $P:\{t_j\}_j\to{\cal D}$ and consider any approximate realization of an ideal quantum ``clock" \cite{Woods_2018, Woods_2022}, i.e., any tuple $(\H_c, \rho_c, H_c, M(t)dt)$ such that 
\begin{equation}
\tr(e^{-iH_ct}\rho_c e^{iH_ct} M(t'))\approx\delta(t-t'),
\end{equation}
for $t\in [-R,R]$, with $R>\max_j |t_j|$. Let $Q(a|x,t)$ be any smooth function of $t$ such that, for fixed $x,t$, $Q(a|x,t)$ is a probability distribution and $Q(a|x,t_j)=P(a|x,t_j)$, for all $a,x,j$. Then we can define the $A$-outcome POVM $M_{a|x}:=\int Q(a|x,t) M(t)dt$ and verify that the dataset
\begin{equation}
P'(a|x,t_j):=\tr(e^{-iH_ct_j}\rho_c e^{iH_ct_j} M_{a|x})
\end{equation}
satisfies $P'(a|x,t_j)\approx P(a|x,t_j)$. Since the approximation can be made as accurate as one wishes, it follows that $P$ can be approximated arbitrarily well by datasets $P'$ admitting a quantum realization. 

Accurate clocks require, however, very energetic quantum states to operate \cite{Woods_2018}. Given that energy is a resource, this suggests concerning ourselves with timelines and datasets whose realizations $(\H,\rho,H,M)$ satisfy certain energy constraints. The nature of those constraints is the subject of the next section.

\subsection{Energy constraints and their corresponding sets of correlations}
\label{sec:sets}
There exist many inequivalent ways to constrain the energy of a quantum system: which one is most suitable will depend on the considered experimental platform. In the following, we list natural energy constraints for standard quantum setups.

\subsubsection{Constraints for discrete spectra}
In some experimental scenarios, the spectrum of the Hamiltonian is finite and known, call it $\E=(E_1,...,E_n)$ (from now on, for any presented list of energy levels, we will assume that it is in increasing order, i.e., $E_1<E_2< \ldots < E_n$). If the dynamics of the system only depends on a finite number of known energy levels, then we speak of a \emph{spectral constraint}. The relevant set of correlations is the set $S(\E)$ of timelines or datasets with realization $(\H, \rho, H, M)$ such that $\mbox{spec}(H)\subset \{E_j\}_j$; or, alternatively, such that the energy distribution $\sigma(E)dE$ of $\rho$ has support in $\{E_j\}_j$. This set of datasets has already appeared in the literature, see, e.g., \cite{jones2024}, where the $X=1, A=N=2$ scenario is studied for the truncated harmonic oscillator, i.e., $\E=(0,\Delta, 2\Delta,...,(n-1)\Delta)$. 

Trapped ions \cite{ion_traps_review} and cold atoms \cite{cold_atoms_review} are obvious experimental platforms where some form of the spectral constraint holds. For starters, by tuning the shape of the trap, one can modify the spectrum of the motion Hamiltonian of the system at will. Furthermore, one can make sure that the ion or atom be prepared in an initial state with support on just the first energy levels $\E$ of the trap. Another experimental platform of interest where the set $S(\E)$ plays a role are optical resonators \cite{circuit_QED_review}: coupling them to a superconducting circuit allows engineering arbitrary superpositions of finitely many photon number states at a given frequency \cite{Hofheinz2009Synthesizing}. In that case, the energy levels of the Hamiltonian are those of a truncated harmonic oscillator. 
    
A ``soft" version of the spectral constraint is the assumption that the first $n$ energy levels of $H$ are $\E=(E_1,...,E_n)$; that $H$ has a gap from $E_n$ to $\eta> E_n$; and that energy levels beyond $\eta$ are populated with probability at most $\epsilon$. We would then be interested in timelines and datasets whose realizations $(\H, \rho, H, M)$ satisfy
\begin{align}
&\mbox{spec}(H)\cap [0,\eta)=\{E_j\}_j,\label{first_levels_cond_temp}\\
&\int_{\eta}^\infty dE\sigma(E)\leq \epsilon\label{density_cond_temp},
\end{align}
where $\sigma(E)dE$ denotes the energy distribution of $\rho$. Call $S_\eta(\E; \epsilon)$ the set of datasets or timelines satisfying such a soft spectral constraint. Curiously enough, for datasets, the closure $\overline{S_\eta(\E;\epsilon)}$ is independent of $\eta$, see Appendix \ref{app:independence_of_eta}. Thus, from now on we will take $\eta\to E_n$ and focus on the set $S(\E;\epsilon)$, defined as:
\begin{align}
&\mbox{spec}(H)\cap [0,E_n]=\{E_j\}_j,\label{first_levels_cond}\\
&\lim_{\delta\to 0^+}\int_{E_n+\delta}^\infty dE\sigma(E)\leq \epsilon\label{density_cond},
\end{align}

The set $S(\E;\epsilon)$ becomes relevant when we prepare atoms or ions in superposition of both bound and free states, like in \cite{exp_cold_atoms}. In that case, it is possible to estimate, through motion state Rabi spectroscopy, the probability $\epsilon$ that the system has energy beyond some threshold $E_n$. 

\subsubsection{Constraints for continuous spectra}
\label{sec:constraints}
Suppose that the energy levels in $\E$ are numerous and very close to each other. Then we can set the origin of energies to be the ground state energy, i.e., $E_1=0$, define $E^+:=E_n$ and consider the slightly larger set $S(E^+)$ of datasets and timelines with realization $(\H, \rho, H, M)$ such that $0\leq H\leq E^+$, or, alternatively, such that the energy density $\sigma(E)dE$ of $\rho$ is supported in $[0,E^+]$. Such a \emph{hard energy support constraint} is considered in \cite{energy_max_bound} in the context of quantum speed limits. 

The soft version of this constraint defines the set $S(E^+;\epsilon)$ of datasets and timelines with realization $(\H,\rho, H,M)$ such that $H\geq 0$ and the energy distribution $\sigma(E)dE$ of $\rho$ satisfies eq. (\ref{density_cond}), with $E_n=E^+$. Note that $S(E^+)=S(E^+;0)$.

In some physical situations, bounding the average energy of the state is also a rather natural assumption, i.e.:
\begin{equation}
\tr(\rho H)\leq \bar{E}
\label{aver_cond}
\end{equation}
(here we assume, of course, that $H\geq 0$). Such scenarios naturally occur in quantum thermodynamics: imagine, for instance, that, starting from a large number $N$ of identical and independent quantum systems at zero temperature, we conduct an operation over the ensemble with energy cost $N\bar{E}$. Next, we take a random constituent system. Assuming that part of the energy might be lost, the only constraint we can give a priory on the state $\rho$ we hold is given by eq. (\ref{aver_cond}). This is the scenario considered in the seminal paper by Margolus and Levitin \cite{margolus_levitin} on quantum speed limits. Translated to our language, their bound on state overlaps reads:
\begin{equation}
\|P(0)-P(t)\|\leq 2\sin\left(\frac{\pi\bar{E}t}{2}\right),
\end{equation}
for $t\in \left[0,\frac{\pi}{2\bar{E}}\right]$.

A concrete experimental scenario where one might encounter condition (\ref{aver_cond}) is in ion or cold atom experiments under a very loose trap. In such setups, the determination of the system's average energy can easily be achieved through Rabi spectroscopy, see \cite{ion_traps_review}, \cite{cold_atoms_review}.

We will denote by $A(\bar{E})$ the set of datasets/timelines with realization $(\H, \rho, H, M)$ satisfying the \emph{average energy constraint} (\ref{aver_cond}) and the condition $H\geq 0$. Note that $A(\bar{E})$ equals the set of datasets with realization $(\H,\rho,H,M)$ such that $H\geq 0$ and $\tr(\rho H)=\bar{E}$ \emph{exactly} \footnote{The reason is that, if some feasible realization in $A(\bar{E})$ satisfies $\tr(\rho)<\bar{E}$, then the realization $(\H\otimes \C^2,\tilde{\rho},\tilde{H}, \tilde{M})$, with $\tilde{\rho}=\rho\otimes \proj{1}$, $\tilde{H}=H\otimes \id_2+(\bar{E}-\tr(\rho H))\id_{\H}\otimes \proj{1}$ and $\tilde{M}_a=M_a\otimes \id_2$ generates the same timeline and satisfies $\tr(\tilde{\rho}\tilde{H})=\bar{E}$.}.

If we lift the condition that the Hamiltonian is bounded from below (this is a good approximation when we are working in a range of very high energies), another natural assumption appears, the \emph{energy variance constraint}:
\begin{equation}
\sqrt{\tr(\rho H^2)-\tr(\rho H)^2}\leq \Delta E.
\label{variance_cond}
\end{equation}
This energy constraint already features in the first paper ever written on quantum speed limits \cite{speed_limits1}, where the Mandelstam-Tamm bound is derived. Adapting the strongest form of the bound \cite{speed_limits2} to our framework, we have:
\begin{equation}
\|P(0)-P(t)\|\leq 2\sin\left(\Delta Et\right),
\end{equation}
for $t\in \left[0,\frac{\pi}{2\Delta E}\right]$. Otherwise, the energy variance constraint appears quite often in works on resource theories, see \cite{speed_limits3, variance1, variance2, variance3}. Both the energy variance assumption and the Mandelstam-Tamm bound are used in \cite{jones2025certifiedrandomnessquantumspeed} to certify measurement-outcome randomness through a 2-point dataset.

Like $\bar{E}$, the quantity $\Delta E$ can be estimated in ion trap or cold atom experiments through Rabi spectroscopy, see, e.g., \cite{exp_cold_atoms} for a relatively recent experiment.

We will call $U(\Delta E)$ the set of datasets/timelines with realization $(\H, \rho, H, M)$ satisfying the energy variance constraint. Similarly to the case of $A(\bar{E})$, one can show that $U(\Delta E)$ equals the set of datasets with realizations saturating the inequality (\ref{variance_cond}).

\subsection{Basic properties of the sets of timed correlations}
\label{sec:relations_sets}
The sets of datasets/timelines defined above are all convex: this can be seen by noting that, for any such set $T$, given two realizations $(\H^j,\rho^j,H^j,M^j)$, $j=1,2$, of feasible timelines $P_1(t), P_2(t)\in T$, for any $\lambda\in [0,1]$, the tuple 
\begin{equation}
(\H^1\oplus\H^2,\lambda\rho^1\oplus (1-\lambda)\rho^2,H^1\oplus H^2,M^1\oplus M^2)
\end{equation}
is a feasible realization for $\lambda P_1(t)+(1-\lambda) P_2(t)$.

In addition, there exist obvious inclusion relations among the sets. Namely, for $\E=(E_1,...,E_n)$, $E^+=E_n$, it holds that 
\begin{itemize}
    \item $S(E^+)\subset A(E^+), U\left(\left(\frac{E^+}{2}\right)^2 \right)$,
    \item $S(\E;\epsilon), S(E^+;\epsilon)\subset S(E^+;\epsilon')$, for $\epsilon'\geq\epsilon>0$,
    \item $S(E^+;\epsilon)\subset S(E^++\Delta;\epsilon)$, for $\epsilon,\Delta>0$,
    \item $S(\E+\lambda;\epsilon)=S(\E;\epsilon)$ for $\epsilon>0,\lambda\in\R$.    
\end{itemize}

Beyond exact relations, there also exist approximate inclusion relations among the sets of correlations listed in Section \ref{sec:constraints}. Given two sets $A, \tilde{A}$ of datasets/timelines and a function $f:\R\to \R^+$, by $A\subset_f \tilde{A}$ we mean that, for any $P\in A$, there exists $\tilde{P}\in \tilde{A}$ such that $\|P(t)-\tilde{P}(t)\|\leq f(t)$, for all $t\in \R$ (for timelines) or $t\in \{t_1,...,t_N\}$ (for datasets). Similarly, $A=_f \tilde{A}$ will be shorthand for $A\subset_f \tilde{A}\subset_f A$.

With this notation, one can easily show that 
\begin{equation}
S(E^+;\epsilon)\subset_{2\sqrt{\epsilon}} S(E^+).    
\label{rel_E_plus_E_plus_epsilon}
\end{equation}
Indeed, let $P\in S(E^+;\epsilon)$, with realization $(\H,\rho,H,M)$. First, note that we can choose $\rho$ to be a pure state, i.e., $\rho=\proj{\psi}$. This is so because, given a mixed state $\rho\in B(\H)$, there exists $\ket{\psi}\in \H^{\otimes 2}$ such that $\rho=\tr_2(\proj{\psi}_{12})$. Thus, the realization $(\H^{\otimes 2},\proj{\psi},H\otimes \id,M\otimes\id)$ generates the same timeline and still complies with the soft energy support constraint. 

We can decompose $\ket{\psi}$ as $\ket{\psi}=\alpha\ket{\psi^-}+\beta\ket{\psi^+}$, where the normalized ket $\ket{\psi^-}$ ($\ket{\psi^+}$) belongs to the subspace of $\H$ spanned by (generalized) eigenstates of $H$ with eigenvalue smaller than or equal to $E^+$ (greater than $E^+$), and $\alpha,\beta$ are real numbers such that $\alpha^2+\beta^2=1$, $\beta\leq \sqrt{\epsilon}$. Then, the realization $(\H,\proj{\psi^-},H,M)$ generates a timeline $\tilde{P}(a|x,t)\in S(E^+)$ with
\begin{align}
&\sum_a|P(a|x,t)-\tilde{P}(a|x,t)|\nonumber\\
&\leq \|e^{-iHt}(\proj{\psi}-\proj{\psi^-})e^{iHt}\|_1\nonumber\\
&=2\sqrt{1-|\braket{\psi}{\psi^-}|^2}\leq 2\sqrt{\epsilon}.
\end{align}

To prove eq.~\eqref{rel_E_plus_E_plus_epsilon}, we relied on the fact that any feasible timeline $P(t)$ admits a realization $(\H,\rho,H,M)$ where $\rho$ is pure. We will call any such realization a \emph{pure-state realization}.

Similarly to (\ref{rel_E_plus_E_plus_epsilon}), one can prove that
\begin{align}
S(\E;\epsilon)\subset_{2\sqrt{\epsilon}} S(\E).
\end{align}

Let us now focus on sets of datasets. Let $\vec{t}=(t_1,...,t_N)$ denote a set of measurement times. For $\E=(E_1,...,E_n)$, with $E_1=0$, $E_n=E^+$, define the maximum spectral gap $g(\E):=\max_j E_{j}-E_{j-1}$. In Appendix \ref{app:discretization}, we prove that, if $g(\E)\leq \frac{\pi}{\max_j|t_j|}$, then, for all $1\geq\epsilon\geq 0$, it holds that
\begin{align}
&S(E^+;\epsilon) \subset_f S(\E;\epsilon),\mbox{ with } f(t)=2\sin\left(\frac{g(\E) |t|}{2}\right).
\label{approx_E_plus_epsilon_by_discrete_epsilon}
\end{align}

So far, we have explored how sets of datasets change when we change the energy assumptions. But, how do they react to changes on the measurement times? Given a set of datasets $T(\bullet)$, let us denote, for the rest of the section, the set of datasets in $T(\bullet)$ with measurement times $\vec{t}$ as $T(\bullet;\vec{t})$. It is easy to prove that, for any converging sequence $(P^j, \vec{t}^j)_j$, such that $P^j\in S(E^+;\vec{t}^j)$, it holds that $\lim_{j\to\infty}P^j\in \overline{S(E^+;\vec{t})}$. That is, the set $\overline{S(E^+;\vec{t})}$ is continuous on $\vec{t}$. So are the closures of $S(\E;\vec{t})$, $A(\bar{E};\vec{t})$ and $U(\Delta E;\vec{t})$, see Appendix \ref{app:discretization}.

This is not the case, however, for soft sets of the form $S(\E;\epsilon)$ or $S(E^+;\epsilon)$: in Appendix \ref{app:continuity_datasets}, we prove that, for $0<\epsilon<\frac{1}{2}$, there exists a dataset $P$ such that, for any $\E$, $P\in S\left(\E;\epsilon;\vec{t}^j\right)$, for all $j\in\N^+$, with $\vec{t}^j=(0,\Delta,2\Delta+\frac{\Delta}{2j+1})$. And yet, for $E^+ \Delta$ small enough, $P$ is at a finite distance from the set $S\left(E^+;\epsilon;\vec{t}\right)$, with $\vec{t}=\lim_{j\to\infty}\vec{t}^j=(0,\Delta, 2\Delta)$.

The set of datasets $S\left(E^+;\frac{1}{2};\vec{t}\right)$ does not depend continuously on the measurement times $\vec{t}$, how can that be? Well, conducting a measurement at a specific time $t$ is an idealization: the closest one can do in practice is conducting a measurement \emph{approximately} at time $t$. If we model a measurement at time $t$ as an instantaneous measurement that takes place at a random time $\tilde{t}$, following a distribution $\gamma(\tilde{t})d\tilde{t}$ centered in $t$, then the sets $S(E^+;\epsilon; \vec{t})$ become continuous under the averaging of the measurement times. 

An exact characterization of $S(E^+;\epsilon;\vec{t})$ will tell us, for a fixed realization $(\H,\rho,H,M)$ compatible with the soft constraint, the limiting dataset generated by $(\H,\rho,H,M)$ as we decrease the width of $\gamma(\tilde{t})d\tilde{t}$. In information processing scenarios where an adversary might control the measurement times and adapt the realization correspondingly, such as in randomness certification \cite{jones2025certifiedrandomnessquantumspeed}, the set 
\begin{equation}
\tilde{S}(\vec{E};\epsilon;\vec{t}):=\lim_{\delta\to 0}\mbox{conv}\left(\cup_{\|\vec{t}-\vec{s}\|\leq \delta}S(E^+;\epsilon;\vec{s})\right)
\label{def_tilde_S_epsilon}
\end{equation}
would be more relevant.

\subsection{Weak interactions with the environment}
Closed quantum systems are an idealization. In practice, a physical system $S$ is expected to jointly evolve with its environment $E$ through a time-independent Hamiltonian of the form:
\begin{equation}
H_T=H_S+H_E+H_{ES}.
\end{equation}
In the expression above, the free Hamiltonians $H_S$, $H_{E}$ are understood to respectively act non-trivially only on $S$ and $E$, while the interaction term $H_{ES}$ acts on the Hilbert spaces of both $S$ and $E$. With this notation, the ``open systems'' observed measurement statistics are
\begin{equation}
P(a|x,t)=\tr\{e^{-iH_Tt}\rho_{ES}e^{iH_Tt}(\id_E\otimes M_{a|x})\},
\end{equation}
where $\rho_{ES}$ denotes the original joint state of system and environment.

In the quantum setups motivating the definitions of the sets $S(\E;\epsilon), S(E^+;\epsilon), A(\bar{E}), U(\Delta E)$ (namely, cold atoms \cite{cold_atoms_review}, trapped ions \cite{ion_traps_review} and frequency resonators coupled to superconducting circuits \cite{circuit_QED_review}), the interaction with the environment is small and controlled. Hence, for short times $t$, the evolution of the system under observation is well approximated by the ``free'' timeline
\begin{equation}
P_0(a|x,t):=\tr\{e^{-iH_St}\rho_{S}e^{iH_St}M_{a|x}\}.
\end{equation}
In fact, experimental demonstrations of quantum speed limits for closed systems have already been successfully conducted in cold atoms \cite{exp_cold_atoms} and frequency resonators \cite{exp_supercond}.

In the following pages, we will provide several results on the structure of $(\H_S,\rho_S,H_S, M)$ or on unobserved measurement statistics statistics $P(t)$, for $t\not\in \{t_1,...,t_N\}$, that hold under the assumption that system $S$ is closed, satisfies certain initial energy constraints and its timeline fits some noisy dataset $(\tilde{P}, \delta)$ for small enough $\delta$. Suppose that we drop the first assumption, but we can nonetheless quantify the deviation from the free evolution of the system, as:
\begin{equation}
\|P(t)-P_0(t)\|\leq \delta(t),
\end{equation}
for some computable function $\delta(t)$. Then we would conclude that $P_0$ fits the dataset $(\tilde{P},\delta')$, with
\begin{equation}
\delta'_{x,j}:=\delta_{x,j}+\delta(t_j).
\end{equation}
If the error bars $\{\delta'_{x,j}\}_{x,j}$ are small enough, structural results on $(\H_S,\rho_S,H_S, M)$ might thus still apply. Similarly, for $\{\delta'_{x,j}\}_{x,j}$ small enough, we will be able to predict $P_0(\tau)$ for future, unobserved times $\tau$. $P_0(\tau)$ is, however, unobserved: to arrive at experimental predictions at time $\tau$, one would have to add a correction $\delta(\tau)$.

In this direction, we next prove a general result: for times $t\sim \|H_{SE}\|^{-1}$, $P(a|x,t)$ is well approximated by
$P_0(a|x,t)$. Let the operators $U(t):=e^{-itH_T}$, $U_0(t):=e^{-it(H_S+H_E)}$ respectively denote the actual and the non-interacting evolution. Expressing $U(t)$ as a Dyson series with respect to the free Hamiltonian $H_S+H_E$, one finds that
\begin{equation}
\|U(t)-U_0(t)\|\leq e^{|t|\|H_{SE}\|}-1.
\end{equation}
If the interaction term is small, i.e., $\|H_{ES}\|\leq \lambda$, with $\lambda |t|\ll 1$, then 
\begin{equation}
\|U(t)-U_0(t)\|\leq O(\lambda |t|).
\end{equation}
This implies that
\begin{align}
&\|U(t)\rho_{ES}U(t)-U_0(t)\rho_{ES}U_0(t)\|_1\nonumber\\
&\leq \|(U(t)-U_0(t))\rho_{ES}U(t)\|_1\nonumber\\
&+\|U_0(t)\rho_{ES}(U(t)-U_0(t))\|_1\nonumber\\
&\leq O(\lambda |t|).
\end{align}
Consequently,
\begin{equation}
\|P(t)-P_0(t)\|\leq O(\lambda |t|).
\label{inter_diff}
\end{equation}

From the above discussion we conclude that, under weak environmental interactions, experimental datasets generated by probing for short times a quantum system initially satisfying the considered energy constraint will approximately belong to the corresponding set of free time-like correlations. Moreover, any general result that holds for finite noise $\delta\not=0$, derived under the free evolution assumption, will also hold in all physical situations where both the statistical noise and the interaction with the environment are sufficiently weak. This applies to all the phenomena we will be reporting next, namely: self-testing datasets (Section \ref{sec:self_testing}), aha! datasets (Section \ref{sec:aha}) and fog banks (Section \ref{sec:fog}). 

\section{Characterization of the correlation sets}
\label{sec:charact}

In this section we explain how to characterize the sets of correlations that we have listed so far through complete sequences of semidefinite programming (SDP) relaxations \cite{sdp} (if the reader is unfamiliar with SDP, they will find a short introduction thereabout in Appendix \ref{app:SDP}). More specifically, we will prove that the set $S(\E)$ is SDP-representable. In addition, for $T\in\{S(E^+),S(\E;\epsilon), S(E^+;\epsilon), A(\bar{E}), U(\Delta E)\}$ we will show that there exists a sequence of SDP-representable sets $(T_k)_k$ such that $T\subset T_k$, and $T_k\subset_{\omega_k} T$, for some monotonically decreasing sequence of positive numbers $(\omega_k)_k$ with $\lim_{k\to \infty}\omega_k=0$. As we will see, for $T\in\{S(E^+), A(\bar{E}), U(\Delta E)\}$ the sequence $(\omega_k)_k$ converges so fast, that the problem of deciding if an $N$-point dataset belongs to $T$ or is at a distance from $T$ greater than $\delta$ (namely, the weak membership problem) can be solved in time polynomial in $\frac{1}{\delta}$ and $N$. For $T\in\{S(\E;\epsilon), S(E^+;\epsilon)\}$, both the form of the SDP relaxations $(T_k)_k$ and the convergence bounds $(\omega_k)_k$ depend on the \emph{exact} measurement times in complicated ways.

Any SDP relaxation of $T$ can be used to falsify the corresponding energy constraint device-independently. For example: the hierarchy of SDP relaxations for $A(\bar{E})$ is $(A_m(\bar{E}))_m$. For any $m$, deciding if there exists $P\in A_m(\bar{E})$ fitting an experimental dataset $(\tilde{P},\delta)$ can be cast as an SDP feasibility problem and thus it can be solved efficiently. If, for some $m$, we find that no dataset in $A_m(\bar{E})$ fits $(\tilde{P},\delta)$, then neither will a dataset in $A(\bar{E})\subset A_m(\bar{E})$. We will thus have proven \emph{device-independently} that the system that generated the noisy dataset has average energy \emph{beyond} $\bar{E}$. Moreover, from $\lim_{m\to\infty}A_m(\bar{E})=\overline{A(\bar{E})}$, such a hierarchy of infeasibility detection tests is complete.

In addition, SDP relaxations of the corresponding set of feasible datasets can help certify randomness in energy-constrained quantum systems. Indeed, recent works \cite{jones2024, jones2025certifiedrandomnessquantumspeed} study the problem of demonstrating the randomness of measurement outcomes in the most basic timed demolition measurement scenario of $X=1, A=N=2$, under different assumptions on the system's energy: a truncated harmonic oscillator in \cite{jones2024} and a system with bounded energy variance in \cite{jones2025certifiedrandomnessquantumspeed}. Through our SDP formulations of $S(\E), U(\Delta E)$, it is possible to extend these protocols to more complex scenarios, with more settings, outcomes and measurement times. Moreover, we feel that our characterization of $A(\bar{E})$ is more relevant than that of $U(\Delta E)$ to analyze the motivating examples in \cite{jones2025certifiedrandomnessquantumspeed} involving coherent states of light.

Finally, our SDP characterization of $S(\E)$ has an immediate application: the design of optimal atomic clocks. We expand on this in Section \ref{sec:clocks}.

\subsection{SDP relaxations: main ingredients}
\label{sec:ingredients}
To derive complete hierarchies of SDP relaxations for the sets of correlations we discussed in the previous section, we will rely on four ideas or ingredients:
\begin{enumerate}
    \item Start from a pure-state realization $(\H,\ket{\psi},H,M)$ of a general element $P$ of the set $T$ and choose an energy cut-off $\eta>0$.
    \item If $\mbox{spec}(H)\cap [-\eta,\eta]$ is not a finite set, define a \emph{discretization function} $\varphi:\R\to \{\tilde{E}_k:k\in\Z\}$. This function will vary according to the set we wish to characterize; in general, $\varphi$ will satisfy these three conditions:
    \begin{itemize}
    \item The timeline $\tilde{P}$ generated by the realization $(\H,\psi, \tilde{H},M)$, with $\tilde{H}=\varphi(H)$, fulfills $\tilde{P}\in T$.
    \item $\varphi(\eta)=\eta$ and $|E-\varphi(E)|\leq \alpha,\forall E\in\spec(H)$, for some $\alpha\in\R^+$ that we call the \emph{resolution} of $\varphi$.
    \item The set of energy levels $\{\tilde{E}_k\}_k\cap [-\eta,\eta]$ is finite.
    \end{itemize}
    From the second condition, it follows by Appendix \ref{app:discretization} (Proposition \ref{prop:discretization}) that 
    \begin{equation}
    \|P(t)-\tilde{P}(t)\|\leq 2\sin\left(\alpha |t|\right),
    \end{equation}
    for all $|t|\leq \frac{\pi}{2\alpha}$. 
    
    Whether $H$ was postulated to be discrete or not, we end up with a realization $(\H,\ket{\psi}, \tilde{H},M)$, with $\spec(\tilde{H})=\{\tilde{E}_k\}_k$. In the following, we assume that either $\tilde{E}_m=\eta$ or that $\tilde{E}_{m-1}=\eta$. In either case, we define $\E:=(\tilde{E}_0,...,\tilde{E}_{m-1})$. From this point on, we focus on characterizing $\tilde{P}$, since it either equals $P$ or approximates $P$ for small $\alpha$.
    
    \item Now there comes the most complicated step. First, we express the initial wave-function in this manner:
    \begin{equation}
    \ket{\psi}=\sum_{k=0}^{m-1}\sqrt{p_k}\ket{\tilde{E}_k}+\sqrt{p_{m}}\ket{\psi^+},
    \end{equation}
    for some normalized states $\{\ket{\tilde{E}_k}\}_k$ with $\tilde{H}\ket{\tilde{E}_k}=\tilde{E}_k\ket{\tilde{E}_k}$. With respect to $\tilde{H}$, the normalized state $\ket{\psi^+}$ has energy support in either $(-\infty,\eta] \cup [\eta,\infty)$ or $(-\infty,\eta) \cup (\eta,\infty)$. If the discretization function $\varphi$ was properly chosen, the vector of weights $p=(p_j)_{j=0}^m$ will be related to the energy distribution of $\ket{\psi}$ under the original Hamiltonian $H$ and thus to the type of energy constraint. Such weights will thus satisfy some constraints $R(p)\leq 0$, which we here assume to be SDP-representable.

    Next, let us define the $(N+m)\times (N+m)$ positive semidefinite matrices
    \begin{equation}
    \tilde{M}_{a|x}:=\Phi M_{a|x}\Phi^\dagger,
    \end{equation}
    with
    \begin{equation}
    \Phi:=\sum_{k=0}^{m-1}\sqrt{p_k}\ket{k}\bra{\tilde{E_k}}+\sqrt{p_{m}}\sum_{j=0}^{N-1}\ket{j+m}\bra{\psi^+_j},
    \end{equation}
    and the vectors $\ket{\psi^+_j}:=e^{-i\tilde{H}t_j}\ket{\psi^+}$. Further define the function $p(a,x|t_j,\tilde{M},\E)$:
    \begin{align}
    p(a|x,t_j,\tilde{M},\E)&:=\bra{\hat{\psi}_j}\tilde{M}_{a|x}\ket{\hat{\psi}_j},
    \label{def_p_spec}
    \end{align}
    with
    \begin{equation}
    \ket{\hat{\psi}_j}:=\sum_{k=0}^{m-1}e^{-i\tilde{E}_kt_j}\ket{k}+\ket{m+j};
    \end{equation}
    denote by $p(t_j,\tilde{M},\E)$ the corresponding datapoint.
    
    Then it holds that
    \begin{align}
    p(a|x,t_j,\tilde{M},\E)&=\bra{\hat{\psi}_j}\tilde{M}_{a|x}\ket{\hat{\psi}_j}\nonumber\\
    &=\bra{\psi}e^{i\tilde{H}t_j}\tilde{M}_{a|x}e^{-i\tilde{H}t_j}\ket{\psi}\nonumber\\
    &=\tilde{P}(a|x,t_j),\forall j,a,x.
    \label{tilde2POVM}
    \end{align}
    In addition, 
    \begin{equation}
    \sum_a\tilde{M}_{a|x}=\Phi\Phi^\dagger=\left(\begin{array}{cc}\mbox{diag}(p_0,...,p_{m-1})&0\\0&\gamma\end{array}\right),
    \label{norm_cond_aux}
    \end{equation}
    where $\gamma_{jk}:=p_{m}\braket{\psi^+_j}{\psi^+_k}$. 
    \item The matrix $\gamma$ belongs to $G(\vec{t})$, the (complex) cone of \emph{decay matrices}. Their exact nature is the subject of Appendix \ref{app:overlaps} and will be discussed in Section \ref{sec:charact_soft}. To arrive at an SDP relaxation for $T$, the last step is choosing an SDP relaxation of $G(\vec{t})$, call it $\tilde{G}(t)$.    
\end{enumerate}

With these ingredients, the resulting SDP-representable relaxation of $T$ is $\tilde{T}$: the set of all datasets $P$ (with times $t_1,...,t_N$) such that there exist a normalized distribution $p=(p_j)_{j=0}^m$, an $N\times N$ complex matrix $\gamma$, and $(m+N)\times (m+N)$ matrices $\{\tilde{M}_{a|x}\}_{a,x}$ satisfying:
\begin{align}
&R(p)\leq 0,\label{conds_ps}\\
&\tilde{M}_{a|x}\geq 0, \forall a,x,\label{tilde_POVMs_pos}\\
&\sum_a\tilde{M}_{a|x}=\left(\begin{array}{cc}\mbox{diag}(p_0,...,p_{m-1})&0\\0&\gamma\end{array}\right),\forall x,\label{tilde_POVMs_norm}\\
&\gamma\in \tilde{G}(\vec{t}),\label{G_relax}\\
&\|P(t_j)-p(t_j,\tilde{M},\E)\|\leq 2\sin\left(\alpha |t_j|\right),\forall j\label{approx_Ps}.
\end{align}
Clearly $T\subset \tilde{T}$. Moreover, as we will next see, for suitably chosen $\varphi,\tilde{G}(\vec{t})$, approximate converse relations of the form $\tilde{T}\subset_{f} T$ can be derived.

In the following sections, we show how to use these ideas to characterize the sets of time-like correlations defined in Section \ref{sec:sets}. The method is, however, very general: it could also be used to characterize, e.g., the datasets generated by Hamiltonians whose spectrum has both discrete eigenvalues as well as finite bands, or the datasets generated by a harmonic oscillator with bounded average energy.

\subsection{Characterization of $S(\E)$}
\label{sec:charact_finite}
To characterize $S(\E)$, with $\E=\{E_0,...,E_{n}\}$, we set $\eta=E_n$, $m=n+1$. There is no need for a discretization function $\varphi$, since the Hamiltonian's spectrum is already finite. Apart from normalization and positivity, the only condition on the weights $(p_j)_{j=0}^{m}$ is $p_m=0$, which implies that $\gamma=0$. We can thus replace the $N+m$-dimensional vectors $\ket{\hat{\psi}_j}$ in eq. (\ref{def_p_spec}) by the $j$-independent, $m$-dimensional vector
\begin{equation}
\ket{\hat{\psi}(t)}:=\sum_{k=0}^{m-1}e^{-iE_kt}\ket{k}. 
\end{equation}

The SDP relaxation we arrive is thus the set of $P(a|x,t)$ such that:
\begin{align}
&P(a|x,t)=\bra{\hat{\psi}(t)}\tilde{M}_{a|x}\ket{\hat{\psi}(t)},\label{P_func_hat_psi}\\
&\sum_{k=0}^{n}p_k=1,\label{norm_coeffs_spec}\\
&\tilde{M}_{a|x}\geq 0, \forall a,x,\label{posi_spec}\\
&\sum_a\tilde{M}_{a|x}=\mbox{diag}(p_0,...,p_{n}),\forall x.\label{norm_spec}
\end{align}
Two things to note here:
\begin{enumerate}
\item This is an SDP relaxation for timelines, not just datasets. This opens many possibilities: we can, for instance, integrate $P(a|x,t)$ in time against a number of functions $\{g_j(t)\}_j$ and characterize the set of feasible results. As we will soon see, this property plays an important role in the design of optimal clocks.
\item This SDP relaxation is, in fact, an SDP characterization. Indeed, let $\{p_k\}_k$, $\{\tilde{M}_{a|x}\}_{a,x}$ satisfy eqs. (\ref{norm_coeffs_spec}), (\ref{posi_spec}) and (\ref{norm_spec}). To prove that the timeline $P(a|x,t)$ defined by eq. (\ref{P_func_hat_psi}) belongs to $S(\E)$, it suffices to define $\Lambda:=\mbox{diag}(p^{1/2}_0,...,p^{1/2}_{m-1})$ and note that $(\C^m,\proj{\psi},\sum_{j}E_j\proj{j},M)$, with $\ket{\psi}=\Lambda\ket{\hat{\psi}}$, $M_{a|x}=\Lambda^{-1}\tilde{M}_{a|x}\Lambda^{-1}$ is a realization of $P$.

\end{enumerate}

The set of \emph{timelines} (not just datasets) $S(\E)$ is therefore SDP representable. This solves, in particular, the problem of characterizing the set ${\cal Q}_{J,A}$ of spin-bounded quantum rotation boxes raised in \cite{Aloy_2024}. In our language, that would be the set of datasets $S(\E)$, with $\E=(-J,-J+1,...,J)$.

\subsubsection{Application: devising an optimal clock}
\label{sec:clocks}
Suppose that we have a system with energies $\E=(E_0,...,E_n)$. We are allowed to prepare whatever initial state $\rho$ we wish and to conduct any POVM $M=(M_a)_a$, thus generating the timeline $P(a|t)$. For time $0\leq t\leq T$, we would like $P(a|t)$ to approximate an ``ideal'' timeline $\hat{P}(t)$. For example, $\hat{P}(t)$ could equal the perfect clock statistics $P_c(a|t)=\Theta(t+1-a)\Theta(a-t)$. In short, we wish to solve the optimization problem:
\begin{align}
&\min_{\rho, M}\int_0^Tdt\sum_a|P(a|t)-\hat{P}(a|t)|^2\nonumber\\
\mbox{such that }&P(a|t)=\tr(e^{-iHt}\rho e^{iHt} M_a),
\end{align}
where $H=\sum_{j=0}^n E_j\proj{j}$.

From the SDP characterization of $S(\E)$ in Section \ref{sec:charact_finite} and its proof of convergence, this problem is equivalent to:
\begin{align}
&\min_{P}\int_0^Tdt\sum_a|P(a|t)-\hat{P}(a|t)|^2\nonumber\\
\mbox{such that }&P(a|t)\in S(\E).
\label{casi_SDP}
\end{align}
In turn, this problem can be cast as an SDP. First, expand eq. (\ref{P_func_hat_psi}) as a sum of imaginary exponentials in $t$:
\begin{equation}
P(a|t)= \sum_{j,k=0}^n\bra{j}\tilde{M}_{a}\ket{k}e^{-i(E_k-E_j)t}.
\end{equation}
Define thus the functions $f_{j,k}(t):=e^{-i(E_k-E_j)t}$ and construct a Gram matrix $F$, with
\begin{equation}
F_{(j,k), (l,m)}=\int_0^Tdt f_{jk}^*(t) f_{lm}(t),
\end{equation}
as well as the vectors $\ket{c^a},\ket{d^a}$, with
\begin{align}
&\ket{c^a}:=\sum_{j,k}\bra{j}\tilde{M}_{a}\ket{k}\ket{(j,k)},\nonumber\\
&\ket{d^a}:=\sum_{j,k}\int_0^Tdt f^*_{jk}(t) \hat{P}(a|t)\ket{(j,k)}.
\end{align}

Then it can be verified that
\begin{align}
&\int_0^Tdt|P(a|t)-\hat{P}(a|t)|^2 \nonumber\\ &=
\int_0^Tdt |\hat{P}(a|t)|^2
-\braket{c^a}{d^a}-\braket{d^a}{c^a}+\bra{c^a}F\ket{c^a}.
\end{align}
The first term on the right-hand side of the equation above is constant; the second term depends linearly on the SDP variables $(\tilde{M}_a)_a$. The third term depends \emph{quadratically} in the SDP variables $(\tilde{M}_a)_a$. In Appendix \ref{app:SDP}, we nonetheless show that constraints of the form 
\begin{equation}
\mu\geq \bra{v}A\ket{v},
\end{equation}
where $A$ is a constant, positive semidefinite matrix and $\ket{v}$ is an affine linear function of the problem variables, can be modeled within SDP. It follows that problem (\ref{casi_SDP}) is equivalent to the SDP:
\begin{align}
\min_{\tilde{M},\vec{p},\vec{\mu}}&\sum_a \mu^a-\braket{c^a}{d^a}-\braket{d^a}{c^a}+\int_0^Tdt|\hat{P}(a|t)|^2\nonumber\\
\mbox{such that }&\tilde{M}_a\geq0,\nonumber\\
&\sum_{a}\tilde{M}_a=\mbox{diag}(p_0,...,p_n),\nonumber\\
&\sum_j p_j=1,\nonumber\\
&\mu^a\geq \bra{c^a}F\ket{c^a},\forall a.
\label{SDP_clock}
\end{align}
Moreover, from the minimizer $(\tilde{M}_a)_a$ it is possible to extract the optimal state and POVM, see the previous section.

For the spectra of the truncated harmonic oscillator and of the hydrogen atom we solved the problem of finding the best approximation of the clock $\check{P}(0|t)=\sum_{k=1}^M\Theta(2k-t)\Theta(t-(2k-1))$. We display the corresponding timelines for $M=5$ and a few different N-values in Figures~\ref{fig:HO} and~\ref{fig:HA}, respectively. We observe that for a harmonic oscillator, as we increase $N$, the clock function is better and better approximated. For the hydrogen atom on the other hand, we were not able to find comparable improvements above $N=6$. Notice also, however, that we were not able to reliably compute the timeline for high $N$ in case of the hydrogen atom (e.g.\ for  $N=20$ like for the harmonic oscillator): this is due to the small energy difference between the higher levels that led to a badly conditioned Gram matrix in the implementation of the last constraint in~\eqref{SDP_clock}, in the sense that many eigenvalues were discarded as being within numerical tolerance of zero.
\begin{figure}[htbp]
    \centering
    \begin{subfigure}{0.49\columnwidth}
        \centering
        \includegraphics[width=\linewidth, trim = 3.2cm 10cm 4.5cm 10.1cm, clip]{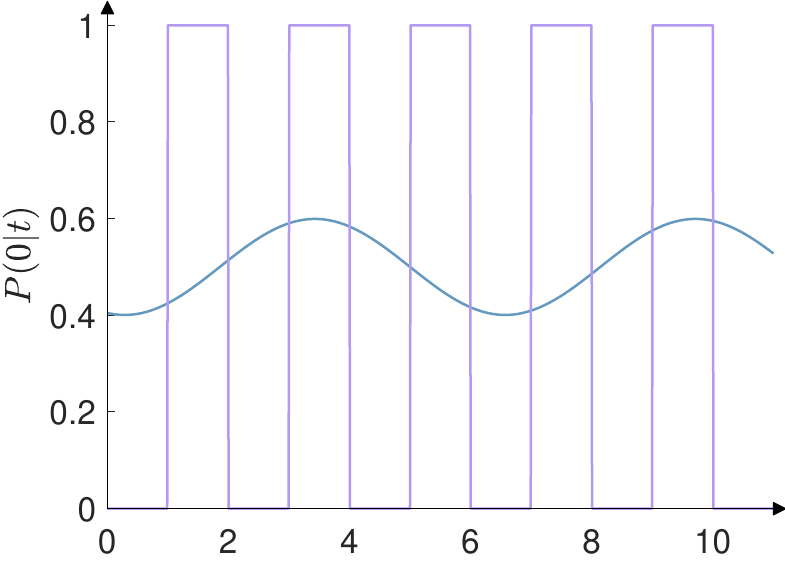}
        \label{fig:sub1}
    \end{subfigure}
    \hfill
    \begin{subfigure}{0.49\columnwidth}
        \centering
        \includegraphics[width=\linewidth, trim = 3.2cm 10cm 4.5cm 10.1cm, clip]{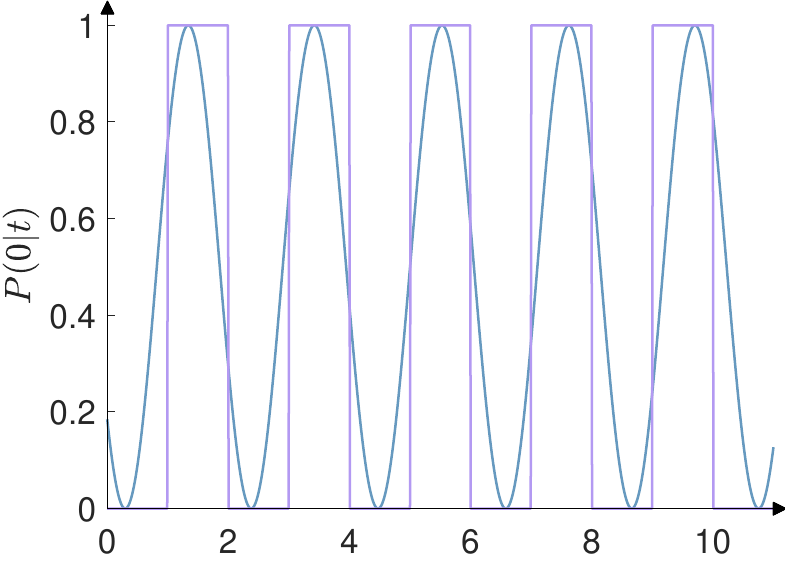}
        \label{fig:sub2}
    \end{subfigure}
    \begin{subfigure}{0.49\columnwidth}
        \centering
        \includegraphics[width=\linewidth, trim = 3.2cm 10cm 4.5cm 10.1cm, clip]{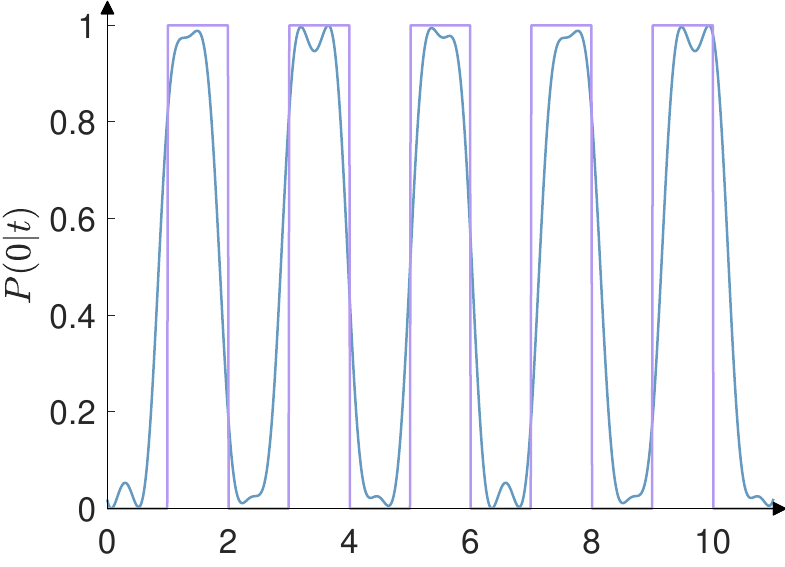}
        \label{fig:sub3}
    \end{subfigure}
    \hfill
    \begin{subfigure}{0.49\columnwidth}
        \centering
        \includegraphics[width=\linewidth, trim = 3.2cm 10cm 4.5cm 10.1cm, clip]{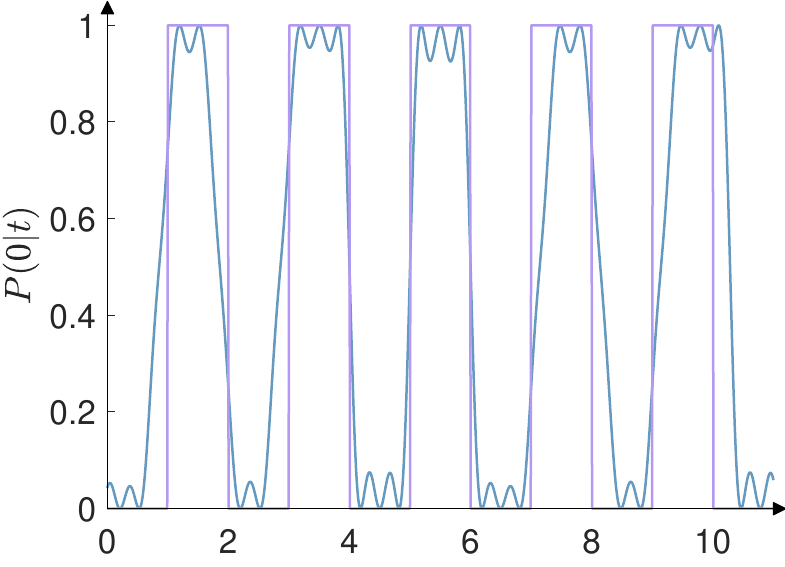}
        \label{fig:sub4}
    \end{subfigure}
    \caption{{\bf Harmonic oscillators as quantum clocks.} The timeline of the clock function with $M=5$ (purple curve) is approximated by a harmonic oscillator with truncated spectrum $(E_0, E_1, \ldots, E_{N-1})$, where $E_k= k$\footnote{This corresponds to setting $\hbar \omega =1$ in the usual harmonic oscillator spectrum, as the constant $\frac{1}{2}$ does not affect the result.}. In blue we display the optimal timeline of the truncated harmonic oscillator 
    of the first $N=2,5,10,20$ energy levels, from left to right, top to bottom. The horizontal axis displays time; the vertical axis, $P(0|t)$.}
    \label{fig:HO}
\end{figure}

\begin{figure}[htbp]
    \centering
    \begin{subfigure}{0.49\columnwidth}
        \centering
        \includegraphics[width=\linewidth, trim = 3.2cm 10cm 4.5cm 10.1cm, clip]{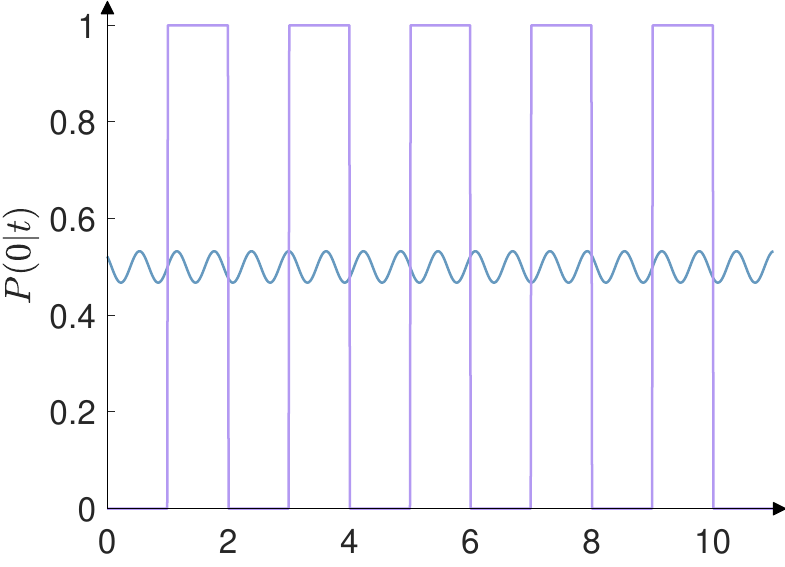}
        \label{fig:sub1}
    \end{subfigure}
    \hfill
    \begin{subfigure}{0.49\columnwidth}
        \centering
        \includegraphics[width=\linewidth, trim = 3.2cm 10cm 4.5cm 10.1cm, clip]{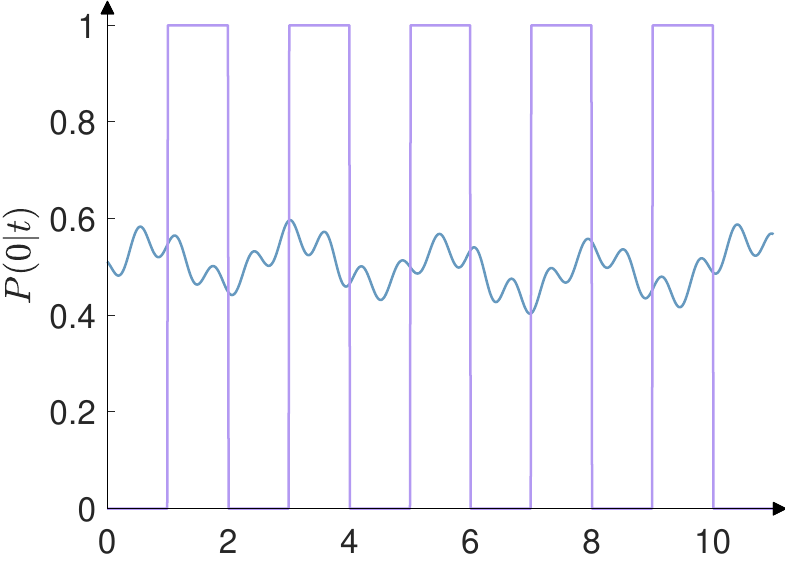}
        \label{fig:sub2}
    \end{subfigure}
    \begin{subfigure}{0.49\columnwidth}
        \centering
        \includegraphics[width=\linewidth, trim = 3.2cm 10cm 4.5cm 10.1cm, clip]{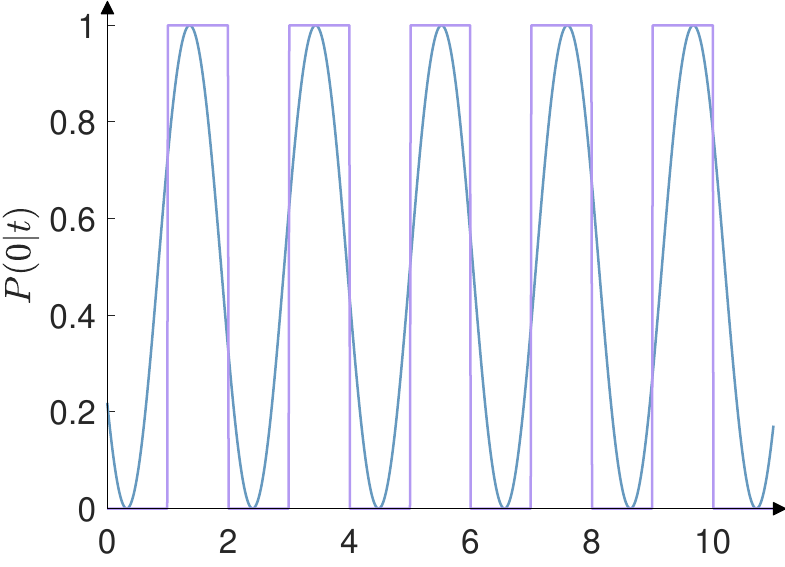}
        \label{fig:sub3}
    \end{subfigure}
    \hfill
    \begin{subfigure}{0.49\columnwidth}
        \centering
        \includegraphics[width=\linewidth, trim = 3.2cm 10cm 4.5cm 10.1cm, clip]{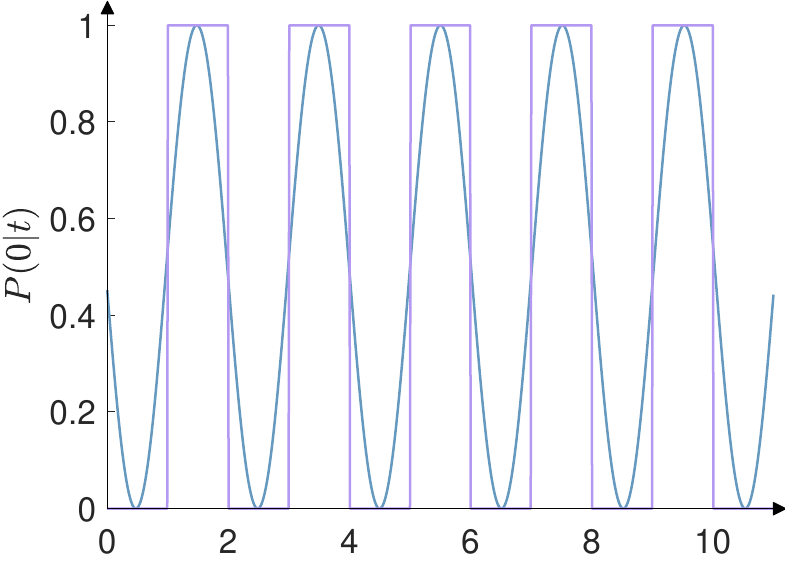}
        \label{fig:sub4}
    \end{subfigure}
    \caption{{\bf Hydrogen atoms as quantum clocks.}  The timeline of the clock function with $M=5$ (purple curve) is approximated by a hydrogen atom with truncated spectrum $(E_1, E_2, \ldots, E_{N})$, where $E_k= \frac{13.6}{k^2}$. In blue we display the optimal timeline of the truncated hydrogen atom for the first 
    $N=2,4,6,8$ energy levels, from left to right, top to bottom. The x-axis displays time, the y-axis $P(0|t)$.}
    \label{fig:HA}
\end{figure}

\subsection{Characterization of $S(E^+)$}
To find an SDP relaxation of $S(E^+)$, we set $\eta=E^+$, choose $m\in \N$ and define $\Delta_m:=\frac{E^+}{m-1}$. The discretization function is:
\begin{equation}
\varphi(E):=\Delta_m\mbox{round}\left(\frac{E}{\Delta_m}\right).    
\end{equation}
The function $\varphi$ maps the real line to $\{\tilde{E}_k:=k\Delta_m:k\in\Z\}$, with resolution $\alpha=\frac{\Delta_m}{2}$. It also sets $\tilde{E}_{m-1}=E^+$. If the initial state $\ket{\psi}$ has energy distribution (with respect to to the original Hamiltonian $H$) supported in $[0,E^+]$, then its energy distribution with respect to $\tilde{H}=\varphi(H)$ will be supported in $\{E_0,...,E_{m-1}\}\subset [0,E^+]$. The conditions on the weights $(p_j)_{j=0}^m$ are the same as in the case of $S(\E)$: namely, normalization and $p_m=0$. 

We arrive at the following SDP relaxation:
$S_m(E^+)$ is the set of timelines $P(a|x,t)$ for which
\begin{align}
&\exists (p_k)_{k=0}^{m-1}:\sum_{k=0}^{m-1}p_k=1,\label{norm_coeffs_E_plus}\\
&\exists \tilde{M}_{a|x}: \tilde{M}_{a|x}\geq 0, \forall a,x,\label{posi_E_plus}\\
&\sum_a\tilde{M}_{a|x}=\mbox{diag}(p_0,...,p_{m-1}),\forall x,\label{norm_E_plus}\\
&\|P(t)-\tilde{P}(t)\|\leq 2\sin\left(\frac{E^+t}{2(m-1)}\right),\nonumber\\
&\mbox{ for } |t|\leq \frac{(m-1)\pi}{E^+},\label{approx_Ps}\\
&\mbox{with }\tilde{P}(a|x,t):=\bra{\hat{\psi}(t)}\tilde{M}_{a|x}\ket{\hat{\psi}(t)},\label{P_func_hat_psi_E_plus}\\
&\ket{\hat{\psi}(t)}:=\sum_{j=0}^{m-1}e^{-ij\Delta_mt}\ket{j}. 
\end{align}
This is once more a set of timelines, although condition (\ref{approx_Ps}) restricts the (continuous) range of times where the SDP relaxation is sound. Note that, by construction, $\tilde{P}\in S(E^+)$. For measurement times $\vec{t}=(t_1,...,t_N)$, with $\texttt{t}:=\max_j|t_j|< \frac{(m-1)\pi}{E^+}$, we thus have that the corresponding set of datasets $S_m(E^+)$ satisfies:
\begin{align}
&S(E^+)\subset S_m(E^+),\nonumber\\
&S_m(E^+)\subset_{f_m} S(E^+),
\end{align}
with
\begin{equation}
f_m=2\sin\left(\frac{E^+\texttt{t}}{2(m-1)}\right).
\end{equation}

\subsection{Characterization of $A(\bar{E})$}
\label{sec:average}
Choose $m\in\N$, with
\begin{equation}
m\geq \max 2\bar{E}\texttt{t},
\label{minimum_m}
\end{equation}
and set
\begin{equation}
\eta=\left(\frac{\bar{E}}{4\texttt{t}^2}\right)^{1/3}m^{2/3}.
\label{eta_optim}
\end{equation}
This time, our discretization function is
\begin{equation}
\varphi_-(E):=\bar{\Delta}_m\left\lfloor\frac{E}{\bar{\Delta}_m}\right\rfloor,
\label{floor_disct}
\end{equation}
with $\bar{\Delta}_m:=\frac{\eta}{m}$. This function has the property 
\begin{equation}
0\leq E-\varphi_-(E)\leq \bar{\Delta}_m,\forall E\geq 0.
\label{decreasing_energy}
\end{equation}
In this relaxation, $\tilde{E}_m=\eta$. Also, condition (\ref{decreasing_energy}) implies that the weights $(p_k)_k$ of the state $\ket{\psi}$ will satisfy
\begin{equation}
\sum_{k=0}^mp_k\tilde{E}_k\leq \bar{E}.    
\end{equation}
This time we cannot eliminate $\gamma$ in eq. (\ref{tilde_POVMs_norm}). We approximate the set of decay matrices by demanding $\gamma$ to be positive semidefinite and have identical diagonal entries. 

The final SDP relaxation of $A(\bar{E})$ is hence $A_m(\bar{E})$, the set of all datasets $P$ (with times $t_1,...,t_N$) such that there exist a normalized distribution $(p_j)_{j=0}^m$, an $N\times N$ complex matrix $\gamma$, and $(m+N)\times (m+N)$ matrices $\{\tilde{M}_{a|x}\}_{a,x}$ satisfying:
\begin{align}
&\sum_{j=0}^mp_{j}\frac{\eta j}{m}\leq \bar{E},\label{aver_spec}\\
&\tilde{M}_{a|x}\geq 0, \forall a,x,\\
&\sum_a\tilde{M}_{a|x}=\left(\begin{array}{cc}\mbox{diag}(p_0,...,p_{m-1})&0\\0&\gamma\end{array}\right),\forall x,\\
&\gamma\geq 0,\gamma_{kk}=p_m,\forall k,\\
&\|P(t_j)-p(t_j,\tilde{M},\E)\|\leq 2\sin\left(\frac{\eta |t_j|}{m}\right),\forall j.\label{approx_aver_E}
\end{align}
This is by construction a relaxation of $A(\bar{E})$, i.e., $A(\bar{E})\subset A_m(\bar{E})$. In Appendix \ref{app:average_and_var_constraints} we further prove that
\begin{align}
&A_m(\bar{E})\subset_{g_m} A(\bar{E}),
\end{align}
with
\begin{equation}
g_m=2\left(\frac{1}{4^{1/3}}+2^{1/3}\right)\left(\frac{\bar{E}\texttt{t}}{m}\right)^{1/3}.
\label{error_aver_hier}
\end{equation}

\subsection{Characterization of $U(\Delta E)$}
Given a pure-state feasible realization $(\H, \ket{\psi},H, M)$ of $P\in U(\Delta E)$, we set the origin of energies so that $\bra{\psi}H\ket{\psi}=0$. This implies that
\begin{equation}
\bra{\psi}H^2\ket{\psi}\leq (\Delta E)^2.
\label{second_moment_bounded}
\end{equation}
Choose $\tilde{m}\in\N$, with
\begin{equation}
\tilde{m}\geq \Delta E \texttt{t},
\label{min_m_var}
\end{equation}
and set
\begin{equation}
\eta=\sqrt{\frac{\Delta E\tilde{m}}{\texttt{t}}}.
\label{def_eta_var}
\end{equation}
The discretization function is
\begin{equation}
\varphi_{><}(E)=\sign(E)\frac{\eta}{\tilde{m}}\left\lfloor\frac{|E|\tilde{m}}{\eta}\right\rfloor.
\end{equation}
The spectrum of the new Hamiltonian is thus $\frac{\eta}{\tilde{m}}\Z$. We cut off all energies $E$, with $|E|\geq \eta$, which implies that $m=2\tilde{m}$. The function $\varphi_{><}$, of resolution $\alpha=\frac{\eta}{\tilde{m}}$, has the effect of shrinking the energies towards $0$; from eq. (\ref{second_moment_bounded}) it thus follows that the weights $(p_k)_k$ of $\ket{\psi}$ under the new Hamiltonian $\varphi_{><}(H)$ satisfy 
\begin{equation}
\sum_{k=0}^mp_k\tilde{E}^2_k\leq (\Delta E)^2.
\end{equation}
The resulting SDP relaxation for $U(\Delta E)$ is $U_{\tilde{m}}(\Delta E)$, the set of all datasets $P$ (with times $t_1,...,t_N$) such that there exist a normalized distribution $(p_j)_{j=0}^{\tilde{m}}$, an $N\times N$ complex matrix $\gamma$, and $(2{\tilde{m}}+N)\times (2{\tilde{m}}+N)$ matrices $\{\tilde{M}_{a|x}\}_{a,x}$ satisfying:
\begin{align}
&\sum_{j=-\tilde{m}+1}^{\tilde{m}}p_{j}\left(\frac{\eta j}{\tilde{m}}\right)^2\leq (\Delta E)^2,\\
\label{var_spec}
&\tilde{M}_{a|x}\geq 0, \forall a,x,\\
&\sum_a\tilde{M}_{a|x}=\left(\begin{array}{cc}\mbox{diag}(p_{-\tilde{m}+1},...,p_{\tilde{m}-1})&0\\0&\gamma\end{array}\right),\forall x,\\
&\gamma\geq 0, \gamma_{kk}=p_{\tilde{m}},\forall k,\\
&\|P(t_j)-p(t_j,\tilde{M},\E)\|\leq 2\sin\left(\frac{\eta |t_j|}{\tilde{m}}\right),\forall j.\label{approx_var}
\end{align}
In Appendix \ref{app:average_and_var_constraints}, we prove that
\begin{align}
&U_{\tilde{m}}(\Delta E)\subset_{s_{\tilde{m}}} U(\Delta E),
\end{align}
with
\begin{equation}
s_{\tilde{m}}=2\sqrt{\frac{\Delta E\texttt{t}}{\tilde{m}}}.
\label{error_var_hier}
\end{equation}

\subsection{Characterization of $S(\E;\epsilon)$}
\label{sec:charact_soft}
To characterize $S(\E;\epsilon)$, with $\E=(E_0,...,E_{m-1})$, choose $\eta=E_{m-1}$. There is no need for a discretization function. Apart from normalization and positivity, the weights must satisfy the soft energy condition
\begin{equation}
p_m\leq \epsilon.
\label{p_epsilon}
\end{equation}
We arrive at the following convex relaxation:
\begin{align}
&P(t_j)=p(t_j,\tilde{M},\E),\label{begin_soft}\nonumber\\
&p_m\leq \epsilon,\\
&\tilde{M}_{a|x}\geq 0, \forall a,x,\\
&\sum_a\tilde{M}_{a|x}=\left(\begin{array}{cc}\mbox{diag}(p_0,...,p_{m-1})&0\\0&\gamma\end{array}\right),\forall x,\label{tilde_POVMs_norm}\\
&\gamma\in G(\vec{t}), p_m=\gamma_{11}.
\label{end_soft}
\end{align}
Here, $G(\vec{t})$ denotes the closure of the cone of square matrices $\gamma$ satisfying
\begin{equation}
\gamma_{jk}=\int_{-\infty}^\infty dE e^{-iE(t_k-t_j)}\rho^+(E),j,k=1,...,N,
\label{def_G}
\end{equation}
for some non-normalized measure $\rho^+(E) dE$. Intuitively, for any $\gamma\in G(\vec{t})$, each entry $\gamma_{jk}$ represents the overlap $\bra{\psi^+}e^{iH(t_j-t_k)}\ket{\psi^+}$, where $\ket{\psi^+}$ is the high-energy part of the initial wave-function $\ket{\psi}$. In this regard, the reader might be surprised to see that the region of integration of the above equation to include energy values below $E_{m-1}$. However, by Lemma \ref{lemma_indent_Gs} in Appendix \ref{app:overlaps} this set of matrices is independent of where one sets the lower limit of integration.

As proven in Appendix \ref{app:soft_constraints} (Lemma \ref{lemma_high_energy}), for any dataset $P$ satisfying eqs. (\ref{begin_soft}-\ref{end_soft}), it holds that $P\in \overline{S(\E;\epsilon;\vec{t})}$. The discontinuity of the set $S(\E;\epsilon,\vec{t})$ in $\vec{t}$ manifests on the fact that the set of decay matrices $G(\vec{t})$ is similarly discontinuous in $\vec{t}$, see Appendix \ref{app:disco}. 

It is thus crucial to characterize the set of decay matrices $G(\vec{t})$ exactly, so let us quickly review what we know about it. In the characterizations of $A(\bar{E})$ and $U(\Delta E)$, we exploited the fact that $G(\vec{t})\subset C_N$, where $C_N$ denotes the cone of correlation matrices, i.e., the set of positive semidefinite matrices $\gamma$ with identical diagonal entries. In Appendix \ref{app:overlaps}, we further show that, for some time vectors $\vec{t}$, the set $G(\vec{t})$ is SDP representable. This is, in fact, the case if, for some $\Delta>0$, it holds that $t_j=k_j\Delta$, for $j=1,...,N$, where $k_j\in\mathbb{N}$. Similarly, for $N=3$, if the real numbers $t_1-t_3,t_2- t_3$ are not congruent (namely, if no non-zero linear combination thereof with rational coefficients vanishes), then the set $G(\vec{t})$ coincides with the SDP representable set $C_3$. In these two cases, the conditions of Lemma \ref{lemma_high_energy} define an SDP to characterize $S(\vec{E};E^+,\epsilon)$.

In general, though, $G(\vec{t})$ needs to be characterized through hierarchies of SDP relaxations. Namely, there exist SDP representable sets $(G^k(\vec{t}))_{k\in\N}$ such that $G^1(\vec{t})\supset G^2(\vec{t})\supset...\supset G(\vec{t})$, with $\lim_{k\to\infty}G^k(\vec{t})=G(\vec{t})$. Each set $G^k(\vec{t})$ corresponds to the Lasserre-Parrilo SDP relaxation of order $k$ \cite{lasserre, parrilo} of the set of measures on complex variables subject to certain polynomial constraints, see Appendix \ref{app:overlaps} for details. 

If the real numbers $\{t_j-t_N\}_{j=1}^{N-1}$ are non-congruent, then the resulting set of decay matrices is known in the literature on convex optimization as $\mbox{CUT}^\infty_N$ \cite{lifting2} (note that the set of decay matrices does not depend on the exact values of $t_1,...,t_N$, as long as the non-congruency condition holds). In \cite{lifting} it is proven that $\mbox{CUT}^\infty_N\subsetneq C_N$, for all $N\geq 4$; in particular, $\mbox{CUT}^\infty_N$ satisfies a number of linear inequalities, violated by the larger set $C_N$. It is easy to see that $\mbox{CUT}^\infty_N\supset G(\vec{t})$, for any vector of measurement times $\vec{t}\in \R^N$, see Appendix \ref{app:overlaps}. It follows that the set $\tilde{S}(\E;\epsilon;\vec{t})$, defined in eq. (\ref{def_tilde_S_epsilon}), demands the requirement $\gamma\in \mbox{CUT}^\infty_N$.

Coming back to the original sets $S(\E;\epsilon)$, for their characterization there are two possibilities:
\begin{enumerate}
    \item $\vec{t}$ is such that $G(\vec{t})$ admits an SDP representation. In that case, $S(\E;\epsilon)$ can be characterized by a single SDP, given by eqs. (\ref{begin_soft}-\ref{end_soft}).
    \item $\vec{t}$ is such that $G(\vec{t})$ admits a complete hierarchy of SDP relaxations $(G^k(\vec{t}))_k$. In that case, $S(\E;\epsilon)$ can be characterized by a hierarchy of SDP relaxations $(S^k(\E;\epsilon))_k$, defined through equations eqs. (\ref{begin_soft}-\ref{end_soft}), but replacing $G(\vec{t})$ by $G^k(\vec{t})$.
\end{enumerate}

With regards to convergence, we prove that
\begin{equation}
S^k(\E;E^+,\epsilon)\subset_{h_k} S(\E;E^+,\epsilon),
\label{conv_soft_hier}
\end{equation}
with $h_k=O\left(\frac{N^2}{k^2}\right)$ in the limit $k\gg 1$; see Appendix \ref{app:soft_constraints} for a proof and the exact expression of $h_k$. 

\subsection{Characterization of $S(E^+;\epsilon)$}
\label{sec:charact_soft_continuous}
Choose $\eta=E^+$ and use the `floor' discretization function $\varphi_-$ in eq. (\ref{floor_disct}), with $\Delta_m=\frac{E^+}{m-1}$, which implies $\tilde{E}_{m-1}=E^+$. Starting from a feasible pure-state realization $(\H,\psi,H,M)$ of a dataset $P\in S(E^+;\epsilon)$, we find that the energy weights $(p_k)_k$ of $\ket{\psi}$ according to $\tilde{H}=\varphi_-(H)$ will satisfy $p_m\leq \epsilon$. If $G(\vec{t})$ is SDP-representable, we end up with the SDP $S_m(E^+;\epsilon)$, which we define as follows: a dataset $P$ belongs to $S_m(E^+;\epsilon)$ iff
\begin{align}
&p_m\leq \epsilon,\sum_kp_k=1,\\
&\tilde{M}_{a|x}\geq 0, \forall a,x,\\
&\sum_a\tilde{M}_{a|x}=\left(\begin{array}{cc}\mbox{diag}(p_0,...,p_{m-1})&0\\0&\gamma\end{array}\right),\forall x,\\
&\gamma\in G(\vec{t}),\\
&\|P(t_j)-p(t_j,\tilde{M},\E)\|\leq 2\sin\left(\frac{E^+|t_j|}{m-1}\right),\forall j.
\end{align}

If, on the contrary, $G(\vec{t})$ admits the complete hierarchy of SDP relaxations $(G^k(\vec{t}))_k$, then one defines the SDP set $S^k_m(E^+;\epsilon)$ replacing $G(\vec{t})$ by $G^k(\vec{t})$ in the above equation.

For $m\geq \frac{E^+\texttt{t}}{2\pi}$, the SDP representable sets $S_m(E^+;\epsilon,\vec{t})$, $S^k_m(E^+;\epsilon,\vec{t})$ are thus relaxations of $S(E^+;\epsilon,\vec{t})$. Moreover, they satisfy:
\begin{align}
S_m(E^+;\epsilon,\vec{t})\subset_{\tilde{h}_m} S(E^+;\epsilon,\vec{t}),\nonumber\\
S_m(E^+;\epsilon,\vec{t})\subset_{h_k+\tilde{h}_m} S^k_m(E^+;\epsilon,\vec{t}),
\end{align}
with
\begin{equation}
\tilde{h}_m:=2\sin\left(\frac{E^+\texttt{t}}{m-1}\right).
\end{equation}

\section{Self-testing datasets}
\label{sec:self_testing}
Noisy datasets contain non-trivial information about the energy distribution of the considered quantum state. The physics underlying an experimental dataset $P$ is, however, much richer: we can identify it with the realization $(\H,\rho,H,M)$ of said dataset. Could certain datasets, under suitable energy constraints, single out a specific realization? Generalizing space-like terminology, one may call such datasets \emph{self-testing}.


Self-testing is a well-studied phenomenon in quantum nonlocality \cite{self_testing_review}, where certain space-like correlations are known to certify the preparation and measurement devices generating them. In this regard, it is possible to self-test any bipartite pure quantum state \cite{coladangelo2017all}, any pure multi-partite qubit state \cite{balanzójuandó2024allpuremultipartite} and any set of projective measurements \cite{self_testing_all_meas}. By analogy with the space-like case, we say that a dataset $P$ \emph{self-tests} the realization $(\bar{\H},\bar{\rho},\bar{H},\bar{M})$ under the energy constraint $T$ if, for any other realization $(\H,\rho,H,M)$ generating $P$, compatible with the energy constraint, there exists an isometry $V:\H\to\bar{\H}\otimes \H_{\text{junk}}$ and a quantum state $\gamma_{\text{junk}}\in B(\H_{\text{junk}})$ such that
\begin{align}
&VM_ae^{-iHt}\rho e^{iHt}V^\dagger\nonumber\\
&-\left(\bar{M}_ae^{-i\bar{H}t}\proj{\bar{\psi}} e^{i\bar{H}t}\right)\otimes \gamma_{\text{junk}}=0
\label{perfect_ST}
\end{align}
holds for all $t$. The dataset $P$ \emph{robustly self-tests} the realization $(\bar{\H},\bar{\rho},\bar{H},\bar{M})$ under the energy constraint $T$ if there exists a continuous error function $s:\R^+\times\R\to\R^+$, with $\lim_{\delta\to 0^+}s(\delta, t)=0$, such that, for any $\tilde{P}\in T$ fitting the dataset $(P,\delta)$, with feasible realization $(\H,\rho,H,M)$, there exists an isometry $V$ such that the left-hand side of the equation above is bounded in trace norm by $s(\delta,t)$.

The reader might wonder how one can guarantee that $\gamma_{\text{junk}}$ in eq. (\ref{perfect_ST}) does not depend on time. How could one assert this, as the measurement device is not acting on the system $\H_{\text{junk}}$? However, in quantum mechanics time evolution demands energy; if generating $P$ already requires saturating the considered energy constraint, then the system holding the junk state cannot have a non-zero Hamiltonian $\bar{H}'$: otherwise, the joint system with Hamiltonian $\bar{H}+\bar{H}'$ would violate the energy assumption.

In the following, we show that, under the hard energy support constraint, any finite-dimensional Hilbert space $\H$ with $\mbox{dim}(\H)\geq 2$ can be self-tested. For the remainder of this section, we thus consider datasets in $S(E^+)$.

For $N\in\N, N\geq 2$, consider the dataset $\mathbf{D}_N\in S(E^+)$
\begin{equation}
\mathbf{D}_N(0|t_0)=1, \mathbf{D}_N(0|t_j)=0,j=1,...,N-1 ,
\label{clock_datasets}
\end{equation}
with $t_j:=\frac{2\pi(N-1)j}{NE^+}$. This dataset can be seen to admit the realization $(\C^N,\bar{H}^N,\proj{\bar{\psi}^N},\bar{M}^N)$, with
\begin{align}
&\bar{H}^N:=\frac{E^+}{N-1}\sum_{k=0}^{N-1}k\proj{k},\ket{\bar{\psi}^N}:=\frac{1}{\sqrt{N}}\sum_{k=0}^{N-1}\ket{k},\nonumber\\
&\bar{M}^N_1:=\proj{\bar{\psi}^N},\bar{M}^N_2:=\id_N-\proj{\bar{\psi}^N}.
\label{reference_realization}
\end{align}
In Figure~\ref{fig:selftest} we display $\mathbf{D}_N$ and the timelines of these realizations for $N=2,3,4,5$.
\begin{figure}[htbp]
    \centering
    \begin{subfigure}{0.49\columnwidth}
        \centering
        \includegraphics[width=\linewidth, trim = 3.2cm 10cm 4.5cm 10.1cm, clip]{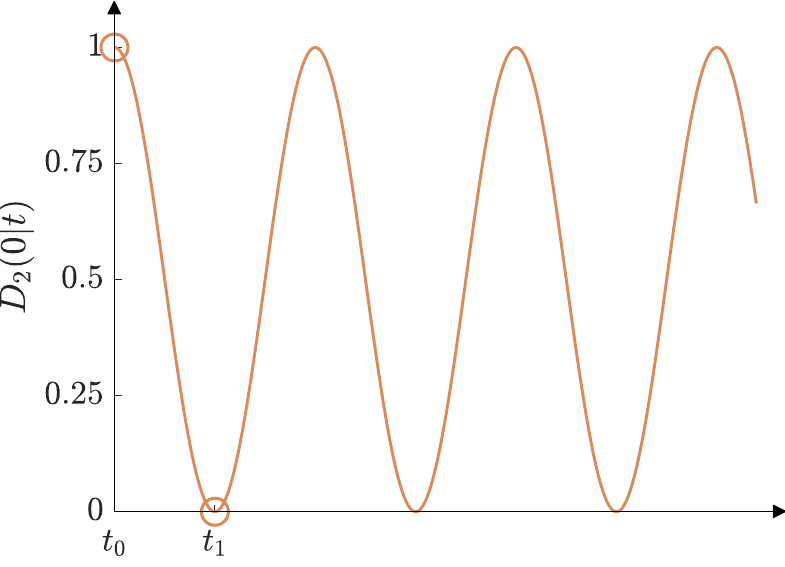}
        \label{fig:sub1}
    \end{subfigure}
    \hfill
    \begin{subfigure}{0.49\columnwidth}
        \centering
        \includegraphics[width=\linewidth, trim = 3.2cm 10cm 4.5cm 10.1cm, clip]{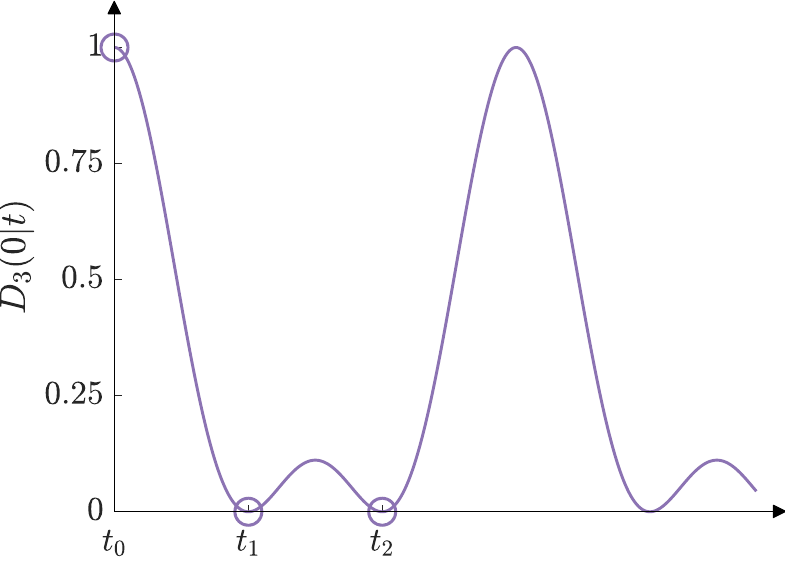}
        \label{fig:sub2}
    \end{subfigure}
    \begin{subfigure}{0.49\columnwidth}
        \centering
        \includegraphics[width=\linewidth, trim = 3.2cm 10cm 4.5cm 10.1cm, clip]{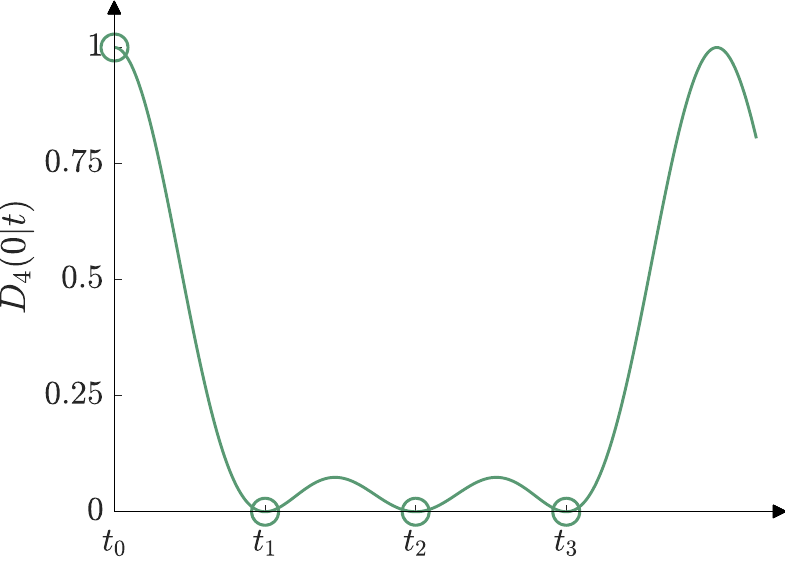}
        \label{fig:sub3}
    \end{subfigure}
    \hfill
    \begin{subfigure}{0.49\columnwidth}
        \centering
        \includegraphics[width=\linewidth, trim = 3.2cm 10cm 4.5cm 10.1cm, clip]{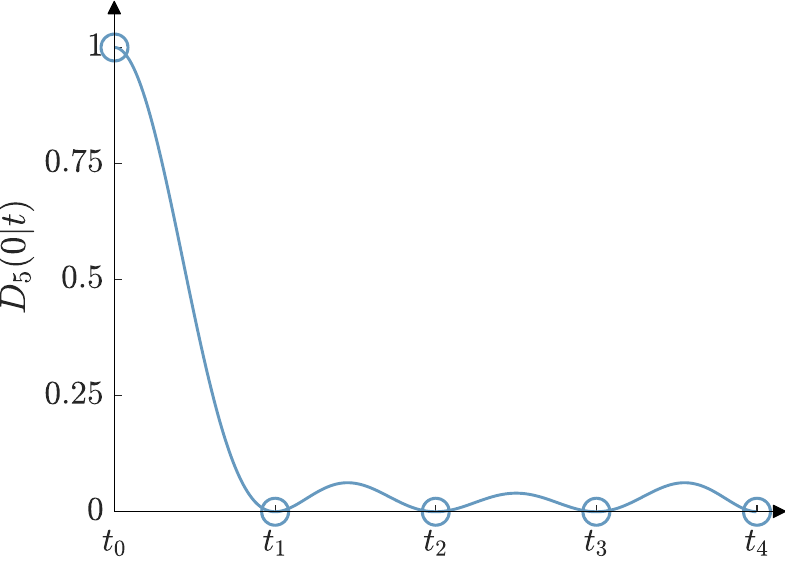}
        \label{fig:sub4}
    \end{subfigure}
    \caption{{\bf Datasets $\mathbf{D}_N$ and timelines of realizations \eqref{reference_realization}.} The x-axis displays the time, the y-axis the probability to observe outcome $0$, the units are chosen the same for all four plots, where $E^+=1$. From left to right, top to bottom we consider $N=2,3,4,5$.}
    \label{fig:selftest}
\end{figure}
In Appendix \ref{app:ST} we prove that, under the hard energy support condition, $\mathbf{D}_N$ robustly self-tests $(\C^N,\bar{H}^N,\proj{\bar{\psi}^N},\bar{M}^N)$ with error function $s(\delta, t)$ satisfying
\begin{align}
&s(\delta,t)=O\left(\delta^{1/64}\right)+O\left(\delta^{1/8}t\right).
\end{align}
This result relies on a technical lemma that will be of use later, so we include it here.
\begin{lemma}
\label{lemma_density}
Let $N\in\N$, let $\ket{\psi}$ be a normalized state and let $\sigma(E)dE$ be its energy density with respect to some Hamiltonian $H$. Suppose that the hard support constraint holds, i.e., $\mbox{supp}(\sigma(E))\subset [0,E^+]$ and that
\begin{equation}
\bra{\psi}e^{iH(t_0-t_k)}\ket{\psi}=0, k=1,...,N-1,
\end{equation}
with
\begin{equation}
t_k:=\frac{2\pi (N-1)k}{N E^+}.
\end{equation}
Then, 
\begin{equation}   \operatorname{supp}(\sigma{(E)}) \subseteq \left\{\frac{E^+k}{N-1}:k=0,...,N-1\right\}.
\end{equation}
\end{lemma}
We refer to Appendix~\ref{app:ST} for an error-tolerant formulation of the lemma and its proof.

Note that the timeline generated by $(\C^2,\bar{H}^2,\proj{\bar{\psi}^2},\bar{M}^2)$ corresponds to Rabi oscillations. In this regard, our result implies that, under the hard-energy support constraint, it suffices to generate a two-point dataset to certify high visibility.

\section{Extrapolation of quantum measurement data}
\label{sec:extrapolation}
In the previous sections we have characterized relevant sets of time-like quantum correlations through semidefinite programming relaxations and formulated self-testing statements in the timed demolition measurements framework. Such constructions can be regarded as time-like analogs of standard results in quantum nonlocality, where self-testing was first conceived \cite{self_testing_review} and the set of correlations is bounded through hierarchies of SDP relaxations \cite{npa, npa2, Berta_2016}.

The timed framework, though, suggests a line of research with no space-like counterpart: the extrapolation of past system behavior. In this section, we will: (A) introduce and motivate the problem of extrapolating quantum measurement data; (B) explain how to solve it with the tools developed so far; (C) discuss how predictable physical systems can be; and (D) prove that the hard energy support constraint leads to two unexpected extrapolation phenomena.

\subsection{The extrapolation problem}
\label{sec:extrapol_framework}
Informally, in an extrapolation problem one seeks to bound the future expectation value of a quantum observable, given the past statistical behavior of the considered quantum system.

Formally, an extrapolation problem, denoted $(\tilde{P}, \delta, f)$, is specified by an objective function $f:{\cal D}\to\R$, which we assume linear, and a noisy dataset $(\tilde{P},\delta)$. It entails solving the two optimization problems
\begin{align}
\mu^{+/-}:=&\max/\min \sum_{a,x}f_{a,x}P(a|x,\tau)\nonumber\\
\mbox{such that }&P\mbox{ fits }(\tilde{P},\delta),\nonumber\\
&P\in T,
\label{extrapol_optim}
\end{align}
where the set $T$ of datasets reflects our assumption on the system's energy. The interpretation is that any $P\in T$ that coincides with $\tilde{P}$ at times $t_1,...,t_N$ up to error $\delta$ will satisfy $\mu^-\leq f(P(\tau))\leq \mu^+$, see Figure~\ref{fig:extra} for an illustration. The solution $[\mu^-,\mu^+]$ of the extrapolation problem is \emph{trivial} if $\mu^-=\min_{P\in{\cal D}}f(P)$, $\mu^+=\max_{P\in{\cal D}}f(P)$. In that case, knowledge of the noisy dataset does not increase predictability. 
\begin{figure}
     \centering  
   \includegraphics[trim={2.6cm 8.8cm 7cm 0.5cm},clip,width=\linewidth]{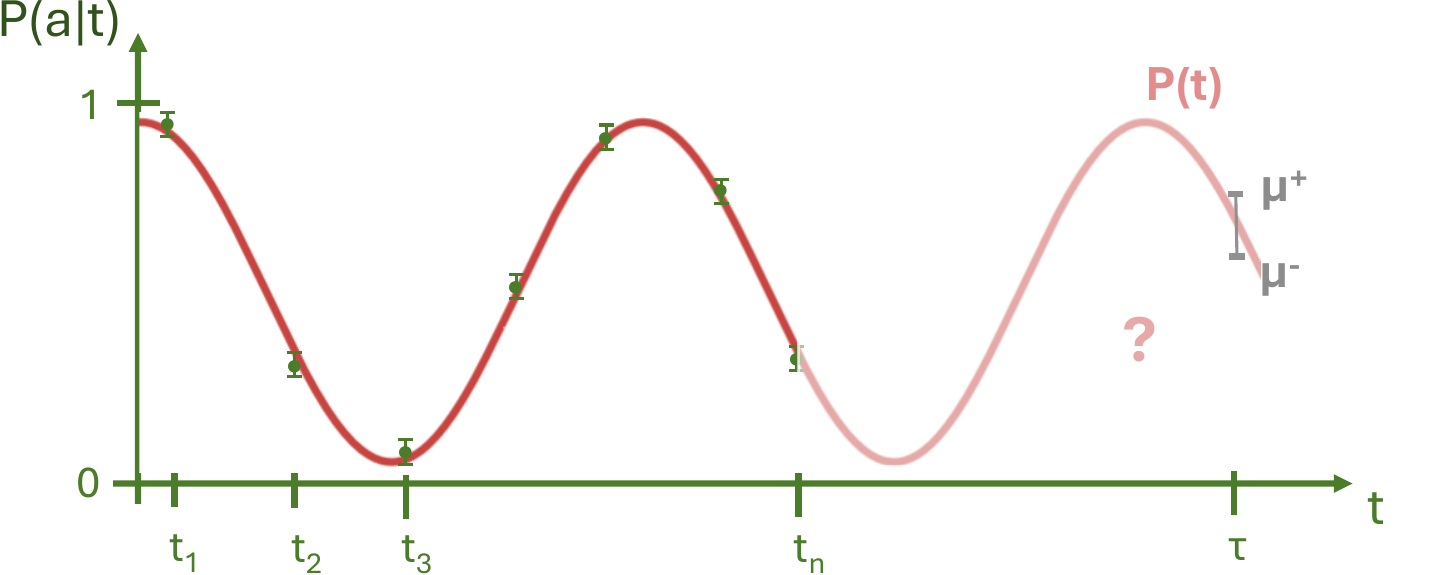}
     \caption{\textbf{Extrapolation problem.} We display the simplest case where we consider only one two-outcome measurement $M$. In this case the pairs $(\tilde{P}(a|t_i), \delta_i)$ for $i=1, \ldots, n$ (displayed by the data points and error bars in the illustration) are sufficient to represent each noisy datapoint. Our goal in this example is to bound the range of possible ${P}(a|\tau)$ (see $\mu^-,\mu^+$). The red curve $P(t)$ displays one example of a timeline that fits the dataset. The part of the illustration where extrapolation is required is coloured palely.}
     \label{fig:extra}
\end{figure}


A main motivation to study the extrapolation problem is that general extrapolation techniques can boost the range of applicability of numerical simulations of complex quantum systems. Consider, for instance, the problem of predicting the magnetization per site ${\cal M}$ at different times of a 1D spin chain subject to a local Hamiltonian evolution. The standard approach to tackle this problem consists of discretizing time and approximating the many-body quantum state at each time-step by a matrix product state (MPS) \cite{MPS} of low bond dimension \cite{MPS_evolution}. This method not only provides a guess on the magnetization-per-site at each time, but also rigorous error bars. A completely different method, recently developed by one of us \cite{differential_paper}, returns a similar output. Very often, the ever-growing error bars suddenly explode at a certain point, signalling, perhaps, that the MPS approximation breaks. From that point on, further predictions become trivial. 

In this class of problems, the energy of the state can be computed exactly, and the range of energy values of the Hamiltonian can be easily estimated. 
The magnetization-per-site operator ${\cal M}$ satisfies the inequalities $-\id\leq{\cal M}\leq \id$, so $M_a:=\frac{\id+(-1)^a{\cal M}}{2}$, with $a=1,2$, defines a POVM. Thus, there is a one-to-one correspondence between two-outcome datasets and time estimates of the magnetization-per-site (extendible to any bounded observable). The extrapolation of the corresponding dataset for some future time $\tau$ would therefore provide us with non-trivial estimates of $\langle {\cal M}\rangle$ at time $\tau$. In this context, our general extrapolation methods for systems under an average energy constraint (and perhaps extra constraints, such as Hamiltonian boundedness) can help us make non-trivial predictions for the average value of any bounded observable of interest beyond the breaking time at which standard numerical methods cease to be useful. 

We will dedicate the rest of this section to anticipate the kind of phenomena one might observe in the resolution of an extrapolation problem.

Let $\tau$ be a desired extrapolation time. If it holds that, for any $Q\in{\cal D}$, there exists a timeline $P(t)$ fitting $(\tilde{P},\delta)$ such that $P(a|x,\tau)=Q(a|x)$, then we say that the dataset $(\tilde{P},\delta)$ exhibits a \emph{Knightian uncertainty} at time $\tau$ \cite{uncertainty}. Under Knightian uncertainty, all possible extrapolation problems are trivial, i.e., the noisy dataset offers no help in making future predictions. The complete opposite of Knightian uncertainty is \emph{full certainty}: the situation that, for some $Q\in{\cal D}$, all timelines $P(t)$ fitting $(\tilde{P},\delta)$ satisfy $P(a|x,\tau)=Q(a|x)$. 

Due to self-testing, the datasets $(\mathbf{D}_N,\delta)$, with $\mathbf{D}_N$ as defined in eq. (\ref{clock_datasets}), provide, in the limit $\delta\to 0$, an example of full certainty at \emph{all} times $t$. Similarly, the empty dataset (with no datapoints) would predict Knightian uncertainty at all times. Beyond such simple examples, it is possible that certain datasets exhibit very complicated patterns of full certainty and Knightian uncertainty. In this regard, we anticipate the following (for the time being, hypothetical) extrapolation effects:
\begin{enumerate}
\item Informally, an \emph{aha! dataset} is a dataset for a single measurement that improves the predictability of another dataset referring to another, independent measurement on the same system. More concretely, consider a dataset $\mathbf{D}_1=(\tilde{P}_1,\delta_1)$ consisting of a single measurement, i.e., $\tilde{P}_1=(\tilde{P}_1(a|t_j))_{a,j}$, that exhibits a Knightian uncertainty at $\tau>\max_j t_j$. Let $\mathbf{D}_2=(\tilde{P}_2,\delta_2)$ represent measurement data from another, independent measurement. It could well be that knowledge of the full dataset $\left(\tilde{P},\delta \right)$, defined through $\tilde{P}(a|x,t_j)=\tilde{P}_x(a|t_j)$, $\delta(x,j)=\delta_x(j)$, allows us to predict the future distribution $\{P_1(a|\tau)\}_a$ with certainty. If so, then we say that $\mathbf{D}_2$ is an aha! dataset for $\mathbf{D}_1$, see Figure~\ref{fig:ahaDs} for an illustration. 

\begin{figure}
    \centering
    \includegraphics[width=\linewidth]{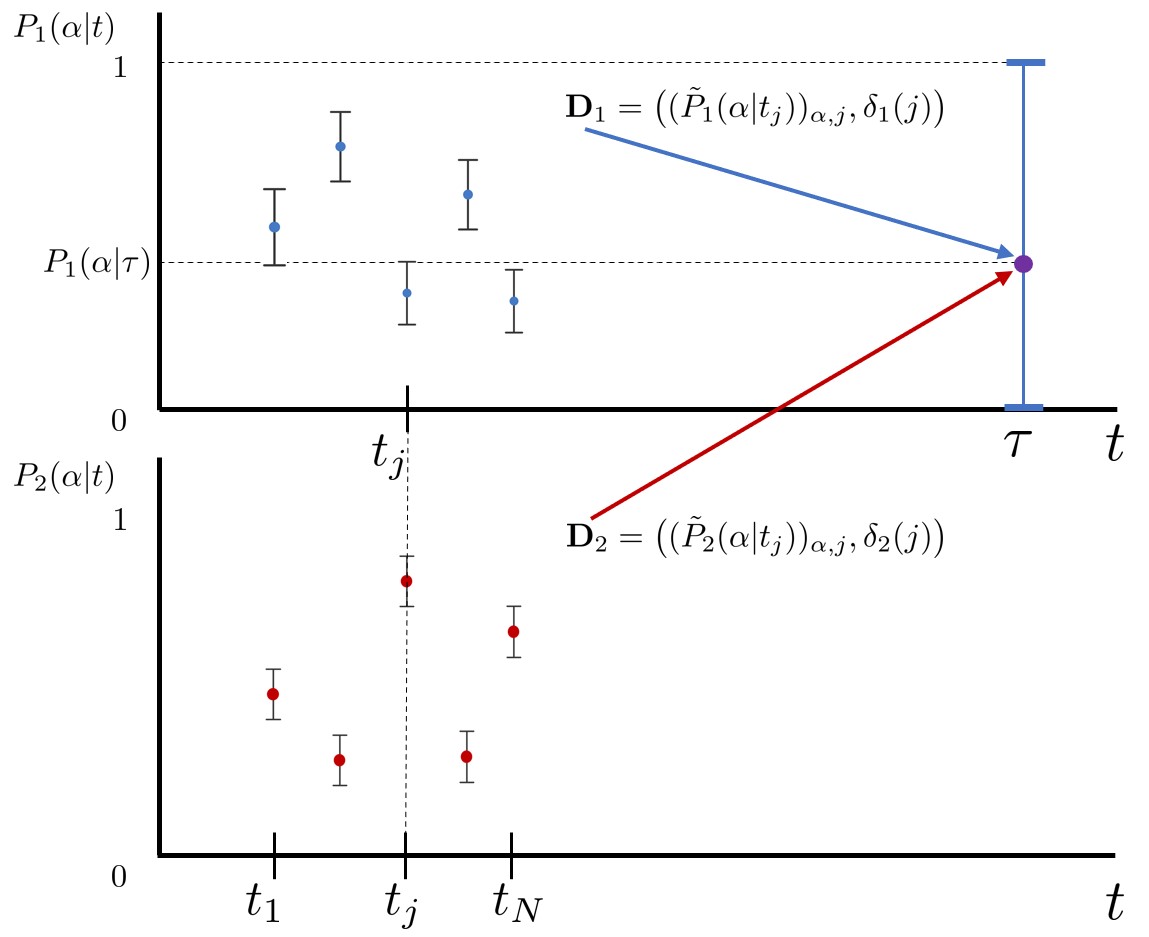}
    \caption{\textbf{An aha! dataset}. The blue dataset $\textbf{D}_1$, generated by conducting measurement $x=1$, exhibits a Knightian uncertainty at time $\tau$. Combining $\textbf{D}_1$ with the red dataset $\textbf{D}_2$, generated by conducting measurement $x=2$ on the same system, it is possible to predict the datapoint $P(a|x=1,\tau)$ with full certainty (purple dot). $\mathbf{D}_2$ is thus an aha! dataset for $\mathbf{D}_1$.}
    \label{fig:ahaDs}
\end{figure}

\item A \emph{fog bank} corresponds to a situation where we lose predictability completely and then regain it, also completely, some time later. More formally, consider a dataset $\tilde{P}$ with Knightian uncertainty at some future time $\tau=\tau_1$. Naively, one would guess that the dataset would present Knightian uncertainty at any later extrapolation time. Imagine, though, that, for some time $\tau_2>\tau_1$, we had full certainty. Then we would say that the dataset $\tilde{P}$ exhibits a fog bank at $\tau=\tau_1$, see Figure~\ref{fig:FB} for an illustration.

\begin{figure}
    \centering
    \includegraphics[width=\linewidth]{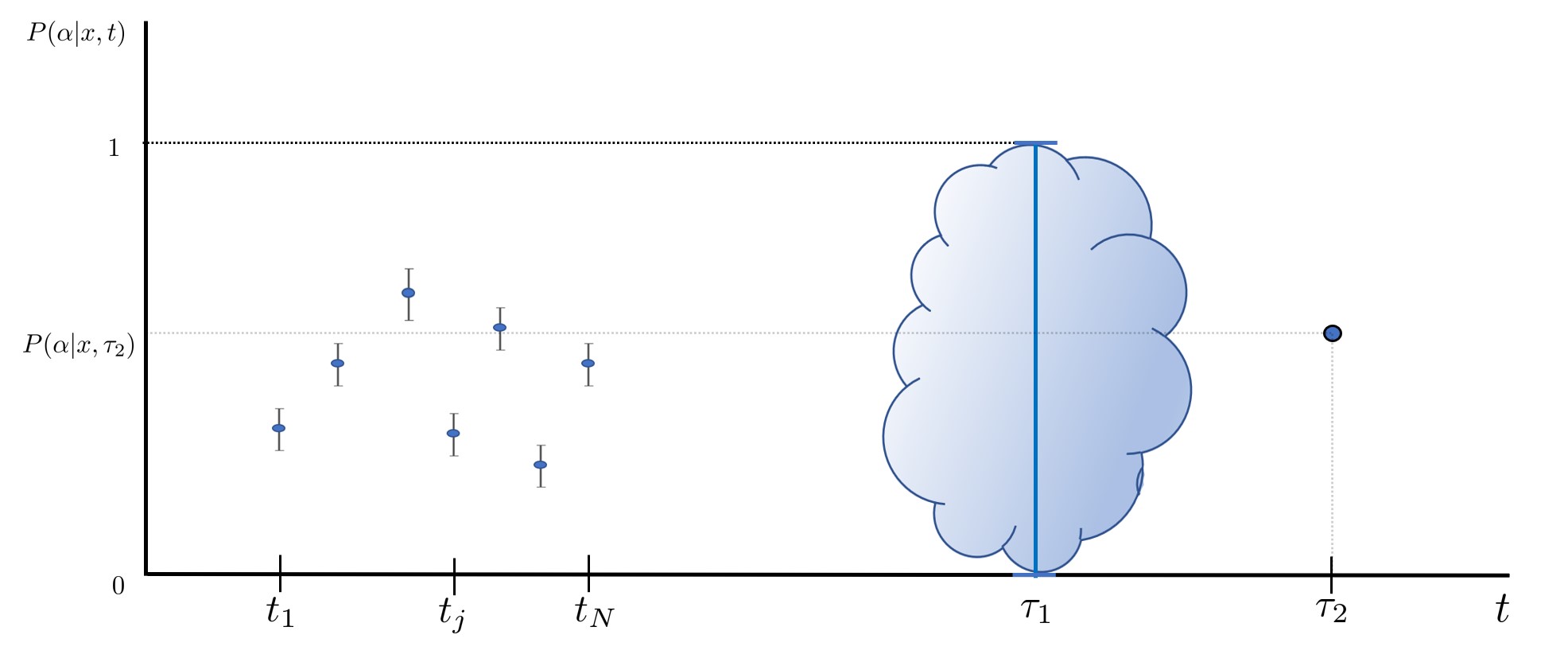}
    \caption{\textbf{A fog bank}. The dataset exhibits a Knightian uncertainty at time $\tau_1$, but full certainty at $\tau_2>\tau_1$.}
    \label{fig:FB}
\end{figure}
\end{enumerate}
In Section \ref{sec:phenomena} we will present examples of datasets, generated by probing low-dimensional quantum systems, that, under the hard energy support constraint, exhibit these two phenomena.

\subsection{Resolution of the extrapolation problem}
In Section \ref{sec:charact}, we provided an SDP formulation of $S(\E)$: replacing $T$ by $S(\E)$ in eq. (\ref{extrapol_optim}), we thus find that extrapolations under the spectral constraint can be cast as an SDP. For all the other sets $T$, Section \ref{sec:charact} provides complete hierarchies of SDP relaxations $(T_m)_m$, together with bounds on the speed of convergence, of the form $T_m\subset_{\omega_m} T$. Given sufficient computational power, such SDP sequences of relaxations will allow us to solve the extrapolation problem for all practical purposes. Indeed, call $\mu^\pm_m$ the solutions of the semidefinite program that results when one replaces $T$ by $T_m$ in (\ref{extrapol_optim}). Since $T\subset T_m$, we have that $f(P(\tau))\in [\mu^-_m,\mu^+_m]$ for any $P\in T$ compatible with the experimental data $(\tilde{P},\delta)$. In other words, each SDP relaxation $T_m$ allows one to make a testable prediction for time $\tau$. In addition, since $\lim_{m\to\infty} T_m=\overline{T}$, we have that $\lim_{m\to\infty}[\mu^-_m,\mu^+_m]$ tends to the optimal extrapolation interval. A quantitative estimation of the quality of the SDP bounds as a function of $m$ is given by the lemma below:
\begin{lemma}
\label{lemma_T_m_approx}
Let $(P,\delta,f)$ be an extrapolation problem, and let $T_m\supset T$ be  a relaxation of $T$ such that $T_m\subset_{\epsilon_m} T$. Further assume that there exists $P^0\in T$ and $r>0$ with
\begin{equation}
\delta(x,j)-\sum_a|P^0(a|x,t_j)-P(a|x,t_j)|\geq r,\forall x,j.
\end{equation}
Then it holds that
\begin{equation}
|\mu_m^{\pm}-\mu^{\pm}|\leq \epsilon_m\left(\frac{|\mu^{\pm}-f(P^0(\tau))|}{r}+\sum_{x}\max_{a}|f(a|x)|\right).
\end{equation}
\end{lemma}
For a proof of the lemma, see Appendix \ref{app:T_m_approx}. 

For the sets $S(E^+), A(\bar{E}), U(\Delta E)$, it thus follows that the extrapolation of an $N$-point dataset can be solved up to precision $\delta$ in time $\mbox{poly}\left(N,\frac{1}{\delta}\right)$. The computational complexity of the extrapolation problem in $S(\E;E^+,\epsilon), S(E^+;\epsilon)$ depends on the congruence relations between the considered measurement times.

\subsection{The predictability of physical systems}
\label{sec:non_trivial}

Knowing how to solve the extrapolation problem does not necessarily mean being able to predict the future behavior of an experimental setup: certain tuples of initial states, Hamiltonians and measurement operators might generate timelines impossible to extrapolate, unless a large number of datapoints thereof with almost vanishing error bars are provided. In this section, we explore how difficult the problem of predicting the future measurement statistics of a quantum device can get, in terms of the number of system resets required to make a useful projection. For simplicity, we will restrict our analysis to scenarios with a hard energy support constraint.

First of all, the good news: thanks to the self-testing results of Section \ref{sec:self_testing}, we know that quantum setups with realization $(\C^N, \ket{\bar{\psi}^N}, \bar{H}^N,\bar{M}^N)$ are so well-behaved that $N$ datapoints thereof with error bars of size $\delta=\mbox{poly}\left(\frac{\alpha}{\tau}\right)$ suffice to make an extrapolation of accuracy $\alpha$ at time $\tau$. Such systems will thus allow conducting long-time extrapolations under reasonable experimental demands, namely, reasonable statistical robustness and a small number of datapoints.

Self-testing datasets might not be telling us the whole story, though: could it be that most systems are, in fact, extremely difficult to extrapolate? Could there be systems, in fact, that are \emph{impossible} to extrapolate?

With regards to the second question, under the hard energy support constraint, a successful extrapolation to arbitrarily long times is always possible. Consider, for starters, the problem of extrapolating a \emph{time segment}, i.e., a map $P:[t_0,t_f]\to {\cal D}$. Under the assumption $P\in S(E^+)$, we have that the function
\begin{equation}
f(t):=\tr(e^{-iHt}\rho e^{iHt}\sum_{a,x}f_{a,x}M_{a|x})
\end{equation}
is analytic for all $t\in\C$. Given $P$ we can therefore expand $f(t)$ in a Taylor series and obtain $P(a|x,\tau)$ \emph{exactly}.

The above example is unrealistic, in that experimental data has noise and spans a finite number of datapoints. What happens in the (just demanding) regime $\infty>N\gg 1$, $1\gg\delta(x,j)> 0$? The problem of extrapolating noisy samples of an unknown analytic function $f(t)$ has been considered before in the computer science literature \cite{extr_analytic}. In \cite{extr_analytic2, extr_analytic3,extr_analytic}, it is shown that a least-squares fit by a polynomial of degree $D$ suffices to provide good bounds, in the afore-mentioned regime. The choice of $D$ solely depends on the maximum value of the function along a region in the complex plane that encompasses the times $t_1,...,t_N$: in our case, $D=D(E^+,\max_j|t_j|)$. The extrapolation method described in \cite{extr_analytic} certainly does not provide optimal extrapolations, but it shows that, under the hard energy support constraint, the measurement statistics of any quantum system can be extrapolated to arbitrary times, as long as the experimental dataset contains sufficiently many low-error datapoints.

Let us now address the first question, i.e., whether the behavior of a quantum system can always be extrapolated in realistic scenarios where both $N$ and $\delta$ take reasonable values. Consider thus the physical system with realization $(\C^2, \frac{\id_2}{2}, 0, M)$, with $M_0=\proj{0},M_1=\proj{1}$. If we probe this system at times $t_j=\frac{j}{N}T$, $j=1,...,N$, aiming for statistical robustness $\delta>0$, we will generate the noisy dataset $\mathbf{O}(N, T,\delta):=(\tilde{P},\delta)$, with 
\begin{align}
&\tilde{P}(1|t_j)=\tilde{P}(2|t_j)=\frac{1}{2}.
\label{def_dataset_O}
\end{align}
Our goal is to extrapolate the value $f(P(\tau)):=P(1|\tau)$. This defines the extrapolation problem $\mathbb{E}(N,T,\delta,\tau)$.

By symmetry, for any $N,\delta$,  the optimal prediction will be of the form $P(1|\tau)\in [\frac{1}{2}-\delta',\frac{1}{2}+\delta']$, for some $\delta'=\delta'(N, T,\delta,\tau)$. By the previous analyticity argument, we have that
\begin{equation}
\lim_{\delta\to 0}\lim_{N\to\infty}\delta'=0.
\end{equation}
In Appendix \ref{app:limits}, we estimate how $\delta'$ depends on $\delta$ in the limit $N\to \infty$. In this regard, we prove that, under the hard energy support constraint with energy threshold $E^+=\frac{2}{T}$, it holds that
\begin{equation}
\delta'\gtrsim \frac{1}{2}\delta\left(\frac{\tau}{T}\right)^{\frac{\log\left(\frac{1}{\delta}\right)}{\log\log\left(\frac{1}{\delta}\right)}}
\end{equation}
in the limit $\delta\ll \frac{T}{\tau}$. This implies that
\begin{equation}
\lim_{\delta\to 0}\lim_{N\to\infty}\frac{\delta'}{\delta}=\infty,
\end{equation}
or, in other words, that high precision in the extrapolation demands \emph{extremely} high precision on the measurement statistics.

This is bad news, at least if one expected to achieve very accurate predictions on this system. Alternatively, one might be content with \emph{any} future prediction, as long as it is non-trivial. Such a reasonable hope is squandered by the following result, also proven in Appendix \ref{app:limits}: under the hard energy support constraint with $E^+=\frac{1}{T}$, avoiding a Knightian uncertainty at $\tau>T$ necessarily demands setting $\delta\leq \frac{1}{g(\tau)}$, for some super-exponential function $g$, independent of $N$. This result precludes long-time extrapolations for the dataset $\mathbf{O}(N, T,\delta)$.

To make matters worse, due to the inclusion relations $S(E^+)\subset S(E^+,\epsilon),A(E^+)$, the above two no-go results also hold under the other two types of energy constraints.

In conclusion: under the hard energy support constraint it is possible to extrapolate the behavior of any quantum system for arbitrarily long times by increasing the number of datapoints and their statistical robustness. However, under reasonable experimental demands, the possibility to make accurate or non-trivial predictions for the future strongly depends on the physical system. Some quantum systems will allow making long-time extrapolations under reasonable accuracy of the measurement statistics and very few datapoints. Some others, even in the limit of infinitely many datapoints, require extremely high precision in the measurement data just to avoid a Knightian uncertainty. To make matters worse, the best criterion we have so far to discriminate one type of systems from the other is directly tackling the extrapolation problem with the computational tools developed in Section \ref{sec:charact}.

\subsection{New extrapolation phenomena}
\label{sec:phenomena}
In Section \ref{sec:extrapol_framework} we postulated the existence of aha! datasets and fog banks. In this section, we show that these two phenomena, far from being hypothetical, can be demonstrated in experiments with low-dimensional quantum systems. For the explanation, we will consider noise-free datasets $(\tilde{P},\delta)$, with $\delta=0$. We note, however, that the two coming examples are \emph{noise-robust}.
Both aha! datasets and fog banks are therefore experimentally observable. Whether fog banks and aha! datasets with full certainty exist in scenarios with noisy measurement data is an interesting question for future research.

\subsubsection{Two instances of aha! datasets}
\label{sec:aha}
In Appendix \ref{app:aha_2} we show that the self-testing dataset $\mathbf{D}_2$ is an aha! dataset for some other, two-point dataset. In this section, we use three-level quantum systems to engineer a much more interesting example: a pair of datasets $\mathbf{A}_1$, $\mathbf{A}_2$ such that each is an aha! dataset for the other.

Define the measurement times $t_k:=\frac{4\pi k}{3E^+}$, for $k=0,1,2$, and consider the following datasets:
\begin{align}
    \mathbf{A}_1&:=\{\tilde{p}(1|t_0)=1,\tilde{p}\left(1|t_1\right)=\frac{1}{3},\tilde{p}\left(1|t_2\right)=0\},\\
    \mathbf{A}_2&:=\{\tilde{p}(1|t_0)=1,\tilde{p}\left(1|t_1\right)=0,\tilde{p}\left(1|t_2\right)=1\}.
    \label{aha_partial}
\end{align}

They represent the time series of two different measurements on the same physical system. As such, they can be integrated into a joint dataset $\mathbf{A}$:
\begin{align}
&\mathbf{A}(\bullet|x=1,\bullet):=\mathbf{A}_1(\bullet|\bullet),\nonumber\\
&\mathbf{A}(\bullet|x=2,\bullet):=\mathbf{A}_2(\bullet|\bullet).
\label{aha_full}
\end{align}
Our choice for extrapolation time is $\tau=\frac{4\pi}{E^+}$.

In Appendix \ref{app:aha_3}, we provide a timeline $P^+\in S(E^+)$, which one can generate by conducting measurements on a three-level quantum system, that fits $\mathbf{A}$. It follows that 
$\mathbf{A}\in S(E^+)$. In Appendix \ref{app:aha_3} we similarly define the four-dimensional timelines $P^-_{1}, P^+_{2}$, which respectively fit the datasets $\mathbf{A}_1, \mathbf{A}_2$. Timelines $P^+,P^-_1, P_2^-$ are plotted in Figure \ref{fig:aha_example}.
\begin{figure}
    \centering
    \includegraphics[width=\linewidth]{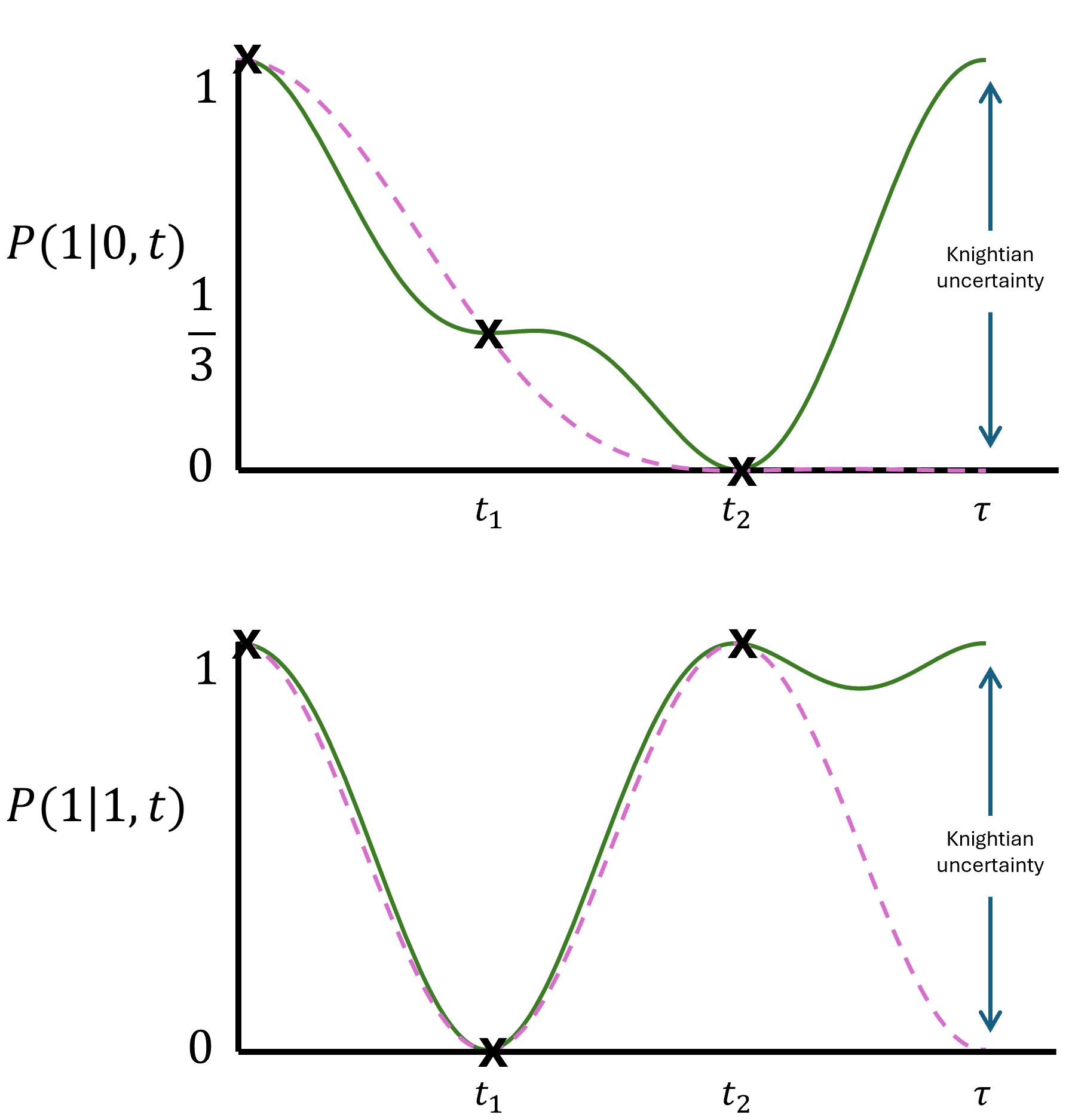}
    \caption{\textbf{A pair of aha! datasets.} The black crosses represent the datapoints of $\mathbf{A}_1, \mathbf{A}_2$. The timeline $P^+$ fitting $\mathbf{A}$ is depicted in dark green, while timelines $P^-_1,P^-_2$, respectively fitting $\mathbf{A}_1, \mathbf{A}_2$, are drawn as magenta dashed lines.}
    \label{fig:aha_example}
\end{figure}

On one hand, the timeline $P^+$ predicts $P^+(1|1,\tau)=P^+(1|2,\tau)=1$. On the other hand, the timelines $P_1^-, P_2^-$ predict $P_1^-(1|\tau)=P_2^-(1|\tau)=0$. For each single-measurement dataset $\mathbf{A}_x$, we therefore have two fitting timelines: one that predicts $P(1|\tau)=1$ and another one predicting $P(1|\tau)=0$. Since any convex combination of the two defines a fitting dataset, we conclude that, on their own, both datasets $\mathbf{A}_1$, $\mathbf{A}_2$ exhibit a Knightian uncertainty at time $\tau$.

Now, let us try to predict the datapoint $P(1|x,\tau)$, given the joint dataset $\mathbf{A}$. Let $(\H,\ket{\psi},H,M)$ be the pure-state realization of any timeline $P$ fitting $\mathbf{A}$. From $P(1|1,t_0)=1$, $P(1|1,t_2)=0$, we infer that $\braket{\psi(t_0)}{\psi(t_2)}=0$. Similarly, from $P(1|2,t_0)=1$, $P(1|2,t_1)=0$ we have that $\braket{\psi(t_0)}{\psi(t_1)}=0$. Invoking Lemma \ref{lemma_density}, we deduce that the energy density of $\ket{\psi}$ has support in $\left\{0, \frac{E^+}{2},E^+\right\}$. Thus, there exist vectors $\{\ket{\psi_k}\}_{k=0}^2$ such that
\begin{equation}
\ket{\psi(t)}=\ket{\psi_0}+e^{-i\frac{E^+t}{2}}\ket{\psi_1}+e^{-iE^+t}\ket{\psi_2}.
\end{equation}
Note that $\ket{\psi(t)}=\ket{\psi(t+\frac{4\pi}{E^+})}$. It follows that $\ket{\psi(t_0)}=\ket{\psi(\tau)}$, which implies that
\begin{equation}
P(1|1,\tau)=P(1|2,\tau)=1.
\end{equation}
The joint dataset $\mathbf{A}$ exhibits full certainty at $\tau$, despite the fact that, separately, each dataset $\mathbf{A}_x$ showed Knightian uncertainty. Hence, each dataset $\mathbf{A}_x$ is an aha! dataset for the other.

Notice that this example is robust to noise: for $\delta>0$ sufficiently small,  the noisy dataset $(\tilde{P}_1,\delta)$ would still exhibit a Knightian uncertainty at $\tau$, while the composite dataset $(\tilde{P},\delta)$ would allow predicting $P_1(a|\tau)$ with \emph{almost} full certainty.

\subsubsection{An instance of a fog bank}
\label{sec:fog}
Consider an experimental scenario with a single measurement setting and three outcomes $a=1,2,3$. For times $t_0,t_2,t_3$ (note that $t_1$ is missing), with $t_k=\frac{3k\pi}{2E^+}$, we define the dataset
\begin{equation}
\mathbf{F}:=\{\tilde{P}(1|t_0)=\tilde{P}(2|t_2)=\tilde{P}(3|t_3)=1\}.
\label{fog_dataset}
\end{equation}
We define the extrapolation times $\tau_1:=t_1+\frac{6\pi}{E^+}$, $\tau_2:=t_2+\frac{6\pi}{E^+}$. Note that $\tau_2>\tau_1$.

In Appendix \ref{app:fog}, we provide three timelines $P^0, P^1, P^2\in S(E^+)$, generated by conducting measurements on a four-level quantum system, that fit $\mathbf{F}$, see Figure \ref{fig:fog_example}. 
\begin{figure}
    \centering
    \includegraphics[width=\linewidth]{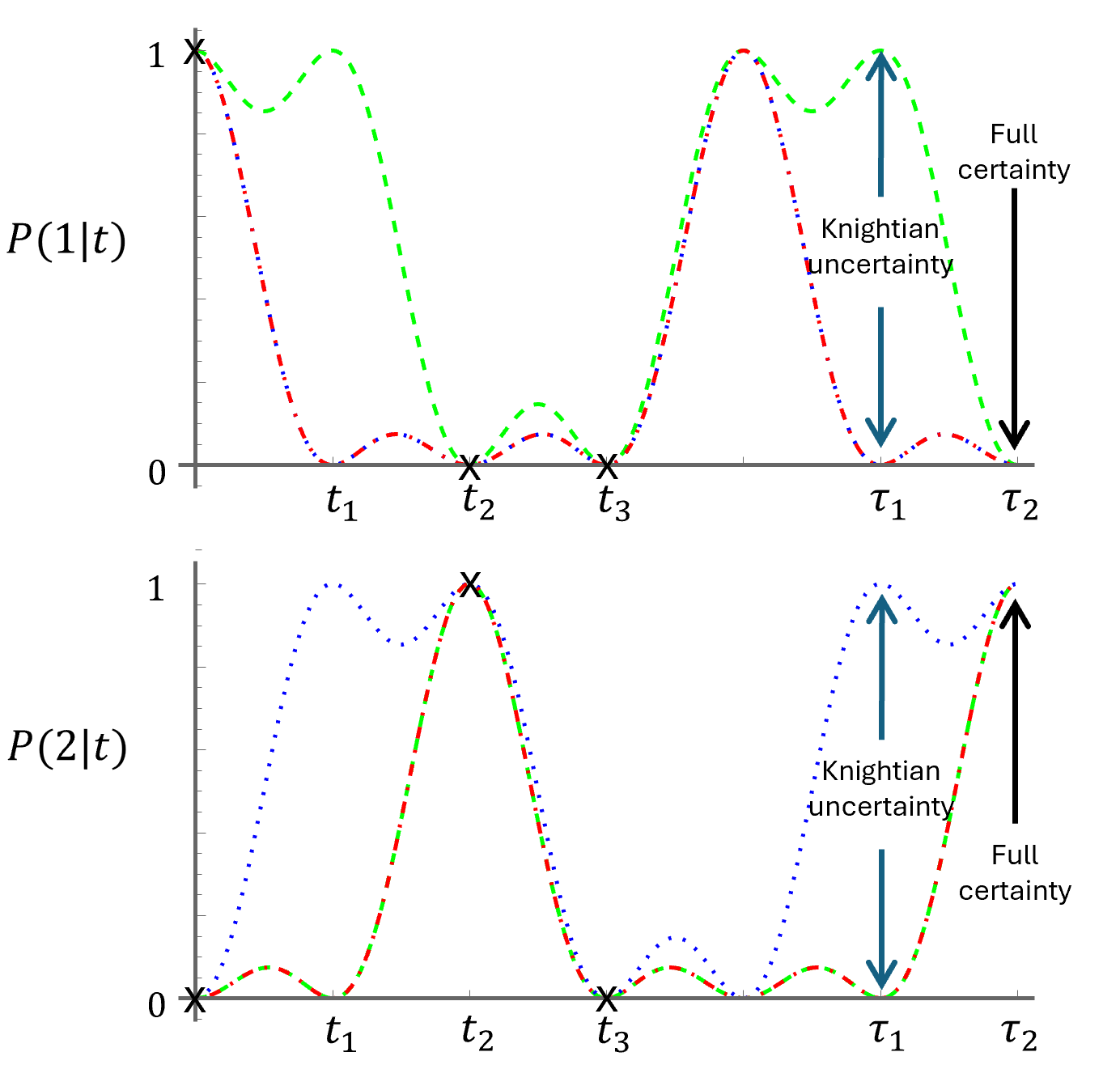}
    \caption{\textbf{An instance of a fog bank.} The datapoints of $\mathbf{F}$ are depicted with black crosses. Timelines $P^0, P^1, P^2$ are respectively represented with green dashed, blue dotted and red dot-dashed lines. }
    \label{fig:fog_example}
\end{figure}
At $t=\tau_1$, $P^0, P^1, P^2$ respectfully predict the outcome distributions $(1,0,0)$, $(0,1,0)$, $(0,0,1)$. Any convex combination $q_0P^0+q_1 P^1+q_2P^2$ of these timelines will fit $\mathbf{F}$ and predict the distribution $(q_0, q_1, q_2)$ at $\tau_1$. There is, therefore, a Knightian uncertainty at $\tau_1$. 

The three timelines, however, converge at $t=\tau_2$, predicting outcome $a=2$ with certainty. In the following, we prove that this is not a coincidence: any feasible timeline fitting $\mathbf{F}$ will make the same prediction.

Let $(\H,\psi,H,M)$ be the pure state realization of any timeline $P(t)$ fitting $\mathbf{F}$. Then $\tilde{P}(1|t_0)=1$, $\tilde{P}(1|t_2)=0$ implies that the state $\ket{\psi(t_2)}$ is orthogonal to $\ket{\psi(t_0)}$. Similarly, $\tilde{P}(1|t_0)=1$, $\tilde{P}(1|t_3)=0$ implies that the state $\ket{\psi(t_3)}$ is orthogonal to $\ket{\psi(t_0)}$. 

The state $\ket{\psi(t_1)}$ is similarly orthogonal to $\ket{\psi(t_0)}$, since
\begin{align}
&\braket{\psi(t_0)}{\psi(t_1)}=\bra{\psi}e^{-iHt_1}\ket{\psi}\nonumber\\
&=\bra{\psi}e^{-iH(t_3-t_2)}\ket{\psi}=\braket{\psi(t_2)}{\psi(t_3)}=0,
\end{align}
where the last relation follows from $\tilde{P}(2|t_2)=1,\tilde{P}(2|t_3)=0$.

We can then apply Lemma \ref{lemma_density} and conclude that the energy density of $\ket{\psi}$ has support in $\left\{\frac{kE^+}{3}:k=0,1,2,3\right\}$. Consequently, $\ket{\psi(t)}=\ket{\psi\left(t+\frac{6\pi}{E^+}\right)}$, and so, for $\tau_2=t_2+\frac{6\pi}{E^+}$, we have that
\begin{equation}
P(2|\tau_2)= P(2|t_2)=1.
\end{equation}
Hence, there is full certainty at $\tau_2>\tau_1$, which implies that we have a fog bank at $\tau=\tau_1$.

Notice that  replacing $(\tilde{P},0)$ with $(\tilde{P},\delta)$, there would still be a Knightian uncertainty at $\tau_1$ in this example and \emph{almost} full predictability at $\tau_2$. Thus the example is robust to noise.

\section{Conclusion}
\label{sec:conclusions}
In this paper, we have introduced the framework of timed demolition measurements to model all quantum experiments where an untrusted, closed quantum system is measured at an arbitrary time. Different assumptions on the energy of the system lead to different sets of time-like correlations. We showed that, like their space-like counterparts, time-like correlations can be self-testing, in the sense that their approximate realization allows one to approximately pinpoint their underlying Hilbert space, initial quantum state, time-independent Hamiltonian and measurement operators. Contrarily to the quantum case, the closure of all such sets can be characterized efficiently: this allows devising optimal clocks, extending previous semi-device independent randomness certification protocols and bounding energy invariants device independently. 

The characterization of the sets allows solving the problem of extrapolating the past measurement data of arbitrary quantum systems. In this regard, we showed that, while some systems allow long-time extrapolations with few datapoints and reasonable noise in the measurement data, other systems demand almost noise-free measurement statistics to make short-time, non-trivial predictions. Finally, we identified two unexpected extrapolation phenomena: aha! datasets and fog banks, which can be already observed in few-level quantum systems.

The potential of timed measurements for QKD or how much the error tolerance of timed measurements-based protocols for randomness certification \cite{jones2024, jones2025certifiedrandomnessquantumspeed} improves as we increase the number of measurement times are natural topics for future research. With regards to the extrapolation problem, it would be interesting to study it in scenarios with memory, such as the temporal correlations scenario \cite{temporal1, temporal2, temporal3}, where the constraint lies, not in the system's energy, but on its dimensionality. Another research line worth investigating,  is whether the characterization of feasible quantum datasets could shed some light on the debate on the nature of gravity \cite{BMM+17, MV17}. If gravity is a quantum field \cite{Rovelli_2004, Becker_Becker_Schwarz_2006}, then any experimental dataset generated during the gravitational interaction of two or more quantum particles is subject to the constraints that we derive in this paper. On the contrary, if gravity is classical \cite{TD16, oppenheim}, one such experiment could generate datasets impossible to account for, under both the quantum hypothesis and some assumption on the energy of the interaction. Could this train of thought lead to an experimental disproof of quantum gravity?

\section*{Acknowledgements}

This project was funded within the QuantERA II Programme that has received funding from the European Union's Horizon 2020 research and innovation programme under Grant Agreement No 101017733 and by the Swiss National Science Foundation via Ambizione PZ00P2\_208779. This work was funded in whole or in part by the Austrian Science Fund (FWF) 10.55776/COE1 and the European Union – NextGenerationEU.

\begin{center}
\includegraphics[width=2cm]{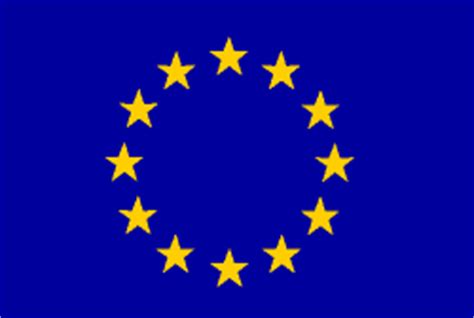}    
\end{center}

\bibliography{extrapol}

\begin{appendix}

\section{Independence of $\overline{S_\eta(\E;\epsilon)}$ on $\eta$}
\label{app:independence_of_eta}
Clearly,  $S_\eta(\E;\epsilon)\supset S_{\tilde{\eta}}(\E;\epsilon)$, for $\eta\leq \tilde{\eta}$, so We next prove that $\overline{S_\eta(\E;\epsilon)}\subset \overline{S_{\tilde{\eta}}(\E;\epsilon)}$. Let thus $P\in S_\eta(\E;\epsilon)$ be a dataset, with measurement times $\vec{t}=(t_1,...,t_N)$ and realization $(\H,\rho,H,M)$. For any $\delta>0$, there exists $r>\tilde{\eta}-\eta$ such that 
\begin{equation}
\|r(t_1,...,t_N)-2\pi(a_1,...,a_N)\|_\infty\leq \delta,
\label{resonance}
\end{equation}
for some $a_1,...,a_n\in\Z$ \cite{chávez2020}.
Let $H$ decompose as $H=\sum_{k=1}^nE_k\Pi_k+ H_U$, where $\{\Pi_k\}_k$ are orthogonal projectors and $H_U\geq \eta (\id-\sum_k\Pi_k)$, with $H_U\Pi_k=0$, for all $k$, represents the high-energy part of $H$. Define the Hamiltonian
\begin{equation}
\tilde{H}= \sum_{k=1}^nE_k\Pi_k+ \tilde{H}_U,
\end{equation}
with $\tilde{H}_U:=H_U + r(\id-\sum_k\Pi_k)$; that is, we have moved the high energy part so that $\tilde{H}_U\geq \tilde{\eta}(\id-\sum_k\Pi_k)$. Call $\tilde{\sigma}(E)dE$ the energy density of $\ket{\psi}$ with respect to $\tilde{H}$. Then it holds that
\begin{align}
&\mbox{spec}(H)\cap [0, \tilde{\eta}]=\{E_k\}_k\nonumber\\
&\int_{\tilde{\eta}}^\infty \tilde{\sigma}(E)dE\leq \epsilon.
\end{align}
Hence, the dataset $\tilde{P}$ with realization $(\H, \rho,\tilde{H},M)$ belongs to $S_{\tilde{\eta}}(\E;\epsilon)$. By eq. (\ref{resonance}), it follows that $\|e^{-iHt_j}-e^{-i\tilde{H}t_j}\|\leq O(\delta)$, for all $j$. Hence, 
\begin{equation}
\|e^{-iHt_j}\rho e^{-iHt_j}-e^{-i\tilde{H}t_j}\rho e^{i\tilde{H}t_j}\|_1\leq O(\delta),
\end{equation}
which implies that $\|P(t_j)-\tilde{P}(t_j)\|\leq O(\delta)$, for $j=1,...,N$. Since we can take $\delta$ as small as we wish, we conclude that $P\in \overline{S_{\tilde{\eta}}(\E;\epsilon)}$.

\section{Approximate relations between the sets of correlations}
\label{app:discretization}
For any function $\varphi:\R\to \R$, define its \emph{resolution} $\alpha$ as
\begin{equation}
\alpha:=\sup_{E\in\R}|E-\varphi(E)|.
\label{def:resolution}
\end{equation}
We next prove a general result regarding timelines and resolutions.
\begin{prop}
\label{prop:discretization}
Let $\varphi:\R\to \R$ have resolution $\alpha$, and let $P$ be a timeline with pure-state realization $(\H,\ket{\psi},H,M)$. Define the timeline $\tilde{P}$ with realization $(\H,\ket{\psi}, \varphi(H),M)$. Then, 
\begin{equation}
\|P(t)-\tilde{P}(t)\|\leq 2 \sin(\alpha |t|),
\end{equation}
for $|t|\leq \frac{\pi}{2\alpha}$.
\end{prop}
\begin{proof}
Define $\tilde{H}:=\varphi(H)$, $\ket{\psi(t)}:=e^{-iHt}\ket{\psi}$, $\ket{\tilde{\psi}(t)}:=e^{-i\tilde{H}t}\ket{\psi}$. Call $\sigma(E)dE$ the energy distribution of $\ket{\psi}$ with respect to the Hamiltonian $H$. Then we have that
\begin{align}
&|\braket{\psi(t)}{\tilde{\psi}(t)}|=\left|\int dE\sigma(E)e^{i(E-\varphi(E))t}\right|\nonumber\\
&\geq\mbox{Re}\left(\int dE\sigma(E)e^{i(E-\varphi(E))t}\right)\nonumber\\
&= \int_{0} \sigma(E)\cos\left((E-\varphi(E))t\right).
\end{align}
In turn, for $|t|\leq \frac{\pi}{2\alpha}$, the integrand is lower bounded by $\cos\left(\alpha t\right)$. This is a consequence of the inequality
\begin{equation}
|E-\varphi(E)||t|\leq\alpha|t|\leq \frac{\pi}{2}    
\end{equation}
and the fact that the function $\cos(|y|)$ is positive and decreasing for $|y|\leq\frac{\pi}{2}$.
Consequently, 
\begin{align}
&|\braket{\psi(t)}{\tilde{\psi}(t)}|\geq \cos\left(\alpha |t|\right).
\end{align}
For any $x$, the above relation implies that 
\begin{align}
\sum_a \left|P(a|x,t)- \tilde{P}(a|x,t)\right| &\leq \|\proj{\psi(t)}-\proj{\tilde{\psi}(t)}\|_1\nonumber\\
&=2\sqrt{1-|\braket{\psi(t)}{\tilde{\psi}(t)}|^2}\nonumber\\
&\leq 2\sin\left(\alpha |t|\right).
\label{diff_Ps}
\end{align}    
\end{proof}

Proposition \ref{prop:discretization} implies that eq. (\ref{approx_E_plus_epsilon_by_discrete_epsilon}) in the main text holds. Indeed, given $\E=(E_1,...,E_n)$, with $E_1=0$, $E_n= E^+$, we define the function $\varphi:\R\to\R$ through:
\begin{align}
\varphi(E):=&E,\mbox{ if }  E>E^+,\nonumber\\
&\arg\min_{E'\in\{E_j\}_j\cup\{E^+\}}|E-E'|, \mbox{ otherwise}.
\end{align}
Clearly,
\begin{equation}
|E-\varphi(E)|\leq \frac{g(\E)}{2},\forall E,
\end{equation}
with $g(vec{E})=\max_j E_{j}-E_{j-1}$. 

Now, let $P\in S(E^+;\epsilon)$ admit the feasible realization $(\H,\rho,H,M)$. Call $\sigma(E)dE$, $\tilde{\sigma}(E)dE$ the energy distributions of $\rho$ with respect to the Hamiltonians $H\geq 0$ and $\tilde{H}:=\varphi(H)\geq 0$. It is immediate to see that if $\sigma(E)dE$ satisfies the soft condition
\begin{equation}
\lim_{\delta\to 0}\int_{E+\delta}^\infty \sigma(E)dE\leq \epsilon,
\end{equation}
then so will $\tilde{\sigma}(E)dE$. In addition, $\tilde{\sigma}(E)dE$ satisfies
\begin{equation}
\mbox{Supp}\left(\tilde{\sigma}(E)\right)\cap [0,E^+]=\{E_j\}_{j=1}^n.
\end{equation}
It follows that the timeline $\tilde{P}$ with realization $(\H, \rho, \tilde{H},M)$ satisfies $\tilde{P}\in S(\E;\epsilon)$. Moreover, by Proposition \ref{prop:discretization} we have that
\begin{equation}
\|P(t)-\tilde{P}(t)\|\leq 2 \sin\left(\frac{g(\E) |t|}{2}\right).
\end{equation}
This proves eq. (\ref{approx_E_plus_epsilon_by_discrete_epsilon}).

\section{Continuity in time of the sets of correlations}
\label{app:continuity_datasets}
In this appendix we explore how the closures of $S(\E), S(E^+), A(\bar{E}), S(\bar{E};\epsilon)$ vary with the measurement times.

To prove that the closure $\overline{T(\vec{t})}$ of some set of correlations $T(\vec{t})$ is continuous in time, it is enough to show that there exists a function $K:\R\to\R$, with $\lim_{y\to 0^+}K(y)=0$, such that, for any pure-state realization $(\H,\ket{\psi},H,M)$ of a timeline $P\in T$, it holds that
\begin{equation}
\|\proj{\psi (t)}-\proj{\psi (\tilde{t})}\|_1\leq K(|t-\tilde{t}|).
\label{existence_K}
\end{equation}
Indeed, in that case 
\begin{equation}
\max_x\sum_a|P(a|x,\tilde{t})-P(a|x,t)|\leq K(|t-\tilde{t}|)    
\end{equation}
follows from the contractive property of the trace distance. Given a converging sequence of pairs of datasets and measurement times $(P^j, \vec{t}^j)_j$, with $\vec{t}^j\to \vec{t}$ and $P^j\in T(\vec{t}^j)$, eq. (\ref{existence_K}) implies that the sequence of datasets $(\hat{P}^j)_j$, with $\hat{P}^j(a|x,t_k):=P^j(a|x,t_k)$, satisfies 
\begin{equation}
\lim_{j\to\infty}\hat{P}^j(a|x,t_k)=\lim_{j\to\infty}P^j(a|x,t^j_k)=:\hat{P}(a|x,t_k).    
\end{equation}
Since $\hat{P}^j\in T(\vec{t})$, for all $j$ and $(\hat{P}^j)_j$ is convergent, it holds that $\hat{P}\in \overline{T}$.  

In turn, that eq. (\ref{existence_K}) holds for some appropriate function $K$ is equivalent to the existence of $s:\R\to\R$, with $\lim_{y\to 0^+}s(y)=1$, such that
\begin{equation}
|\braket{\psi (t)}{\psi (\tilde{t})}|\geq s(|t-\tilde{t}|).
\label{suff_cont}
\end{equation}
This is due to the general relation 
\begin{equation}
\|\proj{\varphi}-\proj{\phi}\|_1=2\sqrt{1-|\braket{\varphi}{\phi}|^2}.
\end{equation}
We will next use this sufficient criterion to prove the continuity of the closures of $S(\E), S(E^+),A(\bar{E}), U(\Delta E)$.

That eq. (\ref{suff_cont}) holds for timelines $P\in S(\E)$ is straightforward. We next see that this is also the case for timelines in $S(E^+)$. Let $P\in S(E^+)$, with pure-state realization $(\H,\psi, H, M)$, and define the time-evolved state vector $\ket{\psi(t)}:=e^{-iHt}\ket{\psi}$. If $E^+|t-\tilde{t}|\leq \frac{\pi}{2}$, then we have that
\begin{align}
&|\braket{\psi(t)}{\psi(\tilde{t})}|=\left|\int_0^{E^+}\sigma(E)dE e^{iE(t-\tilde{t})} \right|\nonumber\\
&\geq \int_0^{E^+}\sigma(E)dE \cos(E(t-\tilde{t}))\geq \cos(E^+(t-\tilde{t})).
\end{align}
The right-hand side fits eq. (\ref{suff_cont}), so the continuity of $\overline{S(E^+;\vec{t})}$ in $\vec{t}$ is proven.

Consider now $P\in A(\bar{E})$, with pure realization $(\H,\ket{\psi}, H, M)$. We can decompose $\ket{\psi}=\alpha\ket{\psi^-}+\beta\ket{\psi^+}$, with $|\alpha|^2+|\beta|^2=1$, such that the energy distribution $\sigma^-(E)dE$ of $\ket{\psi^-}$ has support in $[0,E^+]$ and that of $\ket{\psi^+}$, in $(E^+,\infty)$. Then the condition $\bra{\psi}H\ket{\psi}\leq\bar{E}$ implies that $|\beta|^2\leq \frac{\bar{E}}{E^+}$. Doing a similar calculation as above, we find that 
\begin{align}
&|\braket{\psi(t)}{\psi(\tilde{t})}|\geq |\alpha|^2\cos(E^+(t-\tilde{t}))-|\beta|^2\nonumber\\
&\geq \left(1-\frac{\bar{E}}{E^+}\right)\cos(E^+(t-\tilde{t}))-\frac{\bar{E}}{E^+}.
\end{align}
Choosing $E^+\propto |t-\tilde{t}|^{-1/2}$, the right-hand side of the equation above becomes a function of $|t-\tilde{t}|$ that tends to $1$ as $\tilde{t}$ tends to $t$. The closure of $A(\bar{E};\vec{t})$ is thus continuous in $\vec{t}$.

Finally, let $P\in U(\Delta E)$, with realization $(\H,\ket{\psi}, H, M)$. We choose the origin of energies so that
\begin{equation}
\bra{\psi}H\ket{\psi}=0,
\end{equation}
which implies that
\begin{equation}
\bra{\psi}H^2\ket{\psi}\leq (\Delta E)^2.
\end{equation}
Next, we choose a cut-off $E^+>\Delta E$, and decompose $\ket{\psi}=\alpha\ket{\psi^-}+\beta\ket{\psi^+}$, with $|\alpha|^2+|\beta|^2=1$, such that the energy distribution $\sigma^-(E)dE$ of $\ket{\psi^-}$ has support in $[-E^+,E^+]$ and that of $\ket{\psi^+}$, in $(-\infty, -E^+)\cup (E^+,\infty)$. By the equation above, we have that $|\beta|^2(E^+)^2\leq (\Delta E)^2$. It thus follows that
\begin{align}
&|\braket{\psi(t)}{\psi(\tilde{t})}|\geq |\alpha|^2\cos(E^+(t-\tilde{t}))-|\beta|^2\nonumber\\
&\geq \left(1-\left(\frac{\Delta E}{E^+}\right)^2\right)\cos(E^+(t-\tilde{t}))-\left(\frac{\Delta E}{E^+}\right)^2.
\end{align}
Again, taking $E^+\propto |t-\tilde{t}|^{-1/2}$, the right-hand side becomes a function of $|t-\tilde{t}|$ only, which tends to $1$ as as $\tilde{t}$ tends to $t$, hence proving the continuity of $\overline{\U(\Delta E;\vec{t})}$ in the measurement times.

We next investigate the continuity in $\vec{t}$ of the set of datasets $S(E^+;\epsilon; \vec{t})$, with $0<\epsilon< \frac{1}{2}$. Let $\Delta> 0$, and consider the following dataset $P^j$, defined in the $X=1,A=2$ timed measurement scenario:
\begin{align}
&P^j(1|t^j_1)=\frac{1}{2}\left(1+\sqrt{1-(1-2\epsilon)^2}\right),\nonumber\\
&P^j(1|t^j_2)=P^j(1|t^j_3)=\frac{1}{2}\left(1-\sqrt{1-(1-2\epsilon)^2}\right),
\end{align}
with $t^j_1=0,t^j_2=\Delta$, $t_3^j=2\Delta+\frac{\Delta}{2m+1}$. 
For any vector of energies $\vec{E}\in \R^n$, a feasible realization for this dataset in $S(0;\epsilon;\vec{t}^j)$ is $(\C^2,\proj{\psi}, H_j, M)$, with
\begin{align}
&\ket{\psi}:=(\sqrt{1-\epsilon}\ket{0}+\sqrt{\epsilon}\ket{1}),\nonumber\\
&H_j:=E_n\proj{0}+\left(E_n+\frac{(2j+1)\pi}{\Delta}\right)\proj{1},\nonumber\\
&M_1=\proj{\varphi^+},M_2=\proj{\varphi^-},
\end{align}
where $\ket{\varphi^\pm}$ are the eigenvectors of the matrix
\begin{equation}
\proj{\psi(t^j_1)}-\proj{\psi(t^j_2)},
\end{equation}
with positive and negative eigenvalues. 

Define $\vec{t}:=\lim_{j\to\infty}\vec{t}^j=(0,\Delta,2\Delta)$. If, for some $E^+>0$, $S(E^+;\epsilon;\vec{t})$ were continuous in $\vec{t}$, then $P$ would belong to the closure of $S(E^+;\epsilon;\vec{t})$.

However, this is not the case. We next prove the following result:
\begin{lemma}
\label{lemma:discontinuity_soft}
Let $\Delta \leq \frac{\pi}{2E^+}$, $0<\epsilon< \frac{1}{2}$, and let $\tilde{P}\in S\left(E^+;\frac{1}{2}\right)$ fit the first two datapoints of $P$ with error $\delta$. Then, 
\begin{equation}
|\tilde{P}(1|t_1)-\tilde{P}(1|t_3)|\leq \sqrt{O(\delta)+O((E^+\Delta)^2)}.
\label{inequality_soft}
\end{equation}
\end{lemma}
If $\tilde{P}$ fits the three points of the noisy dataset $(P,\delta)$, the left-hand side of the above equation will equal $\sqrt{1-(1-2\epsilon)^2}+O(\delta)$. For $\delta, \Delta$ small enough, inequality (\ref{inequality_soft}) will thus be violated, proving that $\tilde{P}\not\in S(E^+;\epsilon;\vec{t})$. It follows that no converging sequence of datasets in $S(E^+;\epsilon;\vec{t})$ will tend to $P$, and so $P\not\in\overline{S(E^+;\epsilon;\vec{t})}$. 

\begin{proof}
Suppose that $\tilde{P}\in S(E^+;\epsilon;\vec{t})$, with pure-state realization $(\tilde{\H},\ket{\tilde{\psi}}, \tilde{H},\tilde{M})$, fits the first two points of the dataset $P$, with error $\delta$. By definition,
\begin{equation}
\ket{\tilde{\psi}(t_1)}=\ket{\tilde{\psi}}=\sqrt{1-\epsilon'}\ket{\tilde{\psi}^-}+\sqrt{\epsilon'}\ket{\tilde{\psi}^+},
\end{equation}
where $0\leq\epsilon'\leq \epsilon$, and the energy densities of $\ket{\tilde{\psi}^-}$ and $\ket{\tilde{\psi}^+}$ have support in  $[0,E^+]$ and $(E^+,\infty)$, respectively.

From the first two datapoints of $P$, we have that:
\begin{align}
&2\sqrt{1-(1-2\epsilon)^2}+O(\delta)\nonumber\\
&=\sum_{a}|P(a|t_1)-P(a|t_2)|+O(\delta)\nonumber\\
&=\sum_{a}|\tilde{P}(a|t_1)-\tilde{P}(a|t_2)|\nonumber\\
&\leq \|\proj{\tilde{\psi}(t_1)}-\proj{\tilde{\psi}(t_2)}\|_1\nonumber\\
&=2\sqrt{1-|\braket{\tilde{\psi}(t_1)}{\tilde{\psi}(t_2)}|^2}.
\end{align}
Thus,
\begin{align}
&1-2\epsilon+O(\delta)\geq |\braket{\tilde{\psi}(t_1)}{\tilde{\psi}(t_2)}|\nonumber\\
&\geq \mbox{Re}\left(\braket{\tilde{\psi}(t_1)}{\tilde{\psi}(t_2)}\right)\nonumber\\
&\geq (1-\epsilon')\cos(E^+\Delta)+\epsilon'\mbox{Re}\left(\braket{\tilde{\psi}^+(t_1)}{\tilde{\psi}^+(t_2)}\right).
\label{bound_over}
\end{align}
We next prove that $\ket{\tilde{\psi}(t_2)}$ can be approximated by the vector
\begin{equation}
\ket{\tilde{\psi}'}:=\sqrt{1-\epsilon'}\ket{\tilde{\psi}^-}-\sqrt{\epsilon'}\ket{\tilde{\psi}^+}.
\end{equation}
Indeed, 
\begin{align}
&\|\ket{\tilde{\psi}(t_2)}-\ket{\tilde{\psi}'}\|^2\nonumber\\
&\leq 2\left(1-(1-\epsilon')\cos(E^+ \Delta)+\epsilon'\mbox{Re}(\braket{\tilde{\psi}^+(t_2)}{\tilde{\psi}^+})\right)
\end{align}
Invoking eq. (\ref{bound_over}), we have that
\begin{align}
&\|\ket{\tilde{\psi}(t_2)}-\ket{\tilde{\psi}'}\|^2\nonumber\\
&\leq 4\left(1-(1-\epsilon')\cos(E^+ \Delta)-\epsilon+O(\delta)\right)\nonumber\\
&\leq 4\left(1-(1-\epsilon')\cos(E^+ \Delta)-\epsilon'+O(\delta)\right),\nonumber\\
&=O((E^+ \Delta)^2)+O(\delta),
\end{align}
where we made use of $\epsilon'\leq\epsilon$ in the third line.

Since $\ket{\tilde{\psi}^-}, \ket{\tilde{\psi}^-(t_2)}$ are orthogonal to $\ket{\tilde{\psi}^+}, \ket{\tilde{\psi}^+(t_2)}$ this implies, in particular, that
\begin{equation}
\sqrt{\epsilon'}\|\ket{\tilde{\psi}^+}+ \ket{\tilde{\psi}^+(t_2)}\|\leq \sqrt{O((E^+ \Delta)^2)+O(\delta)}.
\end{equation}
In turn, it follows that $\ket{\tilde{\psi}(t_3)}\approx \ket{\tilde{\psi}(t_1)}$. Let's see why. To begin, we bound the difference between the high energy parts of the two wave functions
\begin{align}
&\|\ket{\tilde{\psi}^+(t_3)}-\ket{\tilde{\psi}^+(t_1)}\|\nonumber\\
&\leq \|\ket{\tilde{\psi}^+(t_3)}+\ket{\tilde{\psi}^+(t_2)}\| +\|\ket{\tilde{\psi}^+(t_2)}+\ket{\tilde{\psi}^+(t_1)}\|\nonumber\\
&\leq \|e^{-iH\Delta}(\ket{\tilde{\psi}^+(t_2)}+\ket{\tilde{\psi}^+(t_1)})\|+\|\ket{\tilde{\psi}^+(t_2)}+\ket{\tilde{\psi}^+(t_1)}\|\nonumber\\
&=2\|\ket{\tilde{\psi}^+(t_2)}+\ket{\tilde{\psi}^+(t_1)}\|
\end{align}
Thus,
\begin{align}
&\sqrt{\epsilon'}\|\ket{\tilde{\psi}^+(t_3)}-\ket{\tilde{\psi}^+(t_1)}\|\nonumber\\
&\leq \sqrt{O((E^+ \Delta)^2)+O(\delta)}.
\end{align}
Next, we consider the whole wave function. The result is
\begin{align}
&\|\ket{\tilde{\psi}(t_3)}-\ket{\tilde{\psi}(t_1)}\|^2\nonumber\\
&=(1-\epsilon')\|\ket{\tilde{\psi}^-(t_3)}-\ket{\tilde{\psi}^-(t_1)}\|^2+\epsilon'\|\ket{\tilde{\psi}^+(t_3)}-\ket{\tilde{\psi}^+(t_1)}\|^2\nonumber\\
&\leq 2(1-\epsilon')(1-\cos(2E^+ \Delta))+O((E^+ \Delta)^2)+O(\delta)\nonumber\\
&=O((E^+ \Delta)^2)+O(\delta).
\end{align}
Finally, 
\begin{align}
&|\tilde{P}(1|0)-\tilde{P}(1|2\Delta)|\nonumber\\
&=|\bra{\tilde{\psi}(t_1)}\tilde{M}_1\ket{\tilde{\psi}(t_1)}-\bra{\tilde{\psi}(t_3)}\tilde{M}_1\ket{\tilde{\psi}(t_3)}|\nonumber\\
&\leq|\bra{\tilde{\psi}(t_1)}\tilde{M}_1(\ket{\tilde{\psi}(t_1)}-\ket{\tilde{\psi}(t_1)})|\nonumber\\
&+|(\bra{\tilde{\psi}(t_3)}-\ket{\tilde{\psi}(t_1)})\tilde{M}_1\ket{\tilde{\psi}(t_3)}|\nonumber\\
&\sqrt{O(\delta)+O((E^+ \Delta)^2)}.
\end{align}

\end{proof}

\section{The basics of semidefinite programming}
\label{app:SDP}
A semidefinite program is an optimization problem of the form \cite{sdp}:
\begin{align}
p^\star:=&\min_{\vec{x}\in\R^n} \vec{c}\cdot \vec{x}\nonumber\\
\mbox{such that }&F_0+\sum_{i=1}^n x_i F_i\geq 0,
\label{primal_problem}
\end{align}
where $F_0, F_1,...,F_n$ are symmetric square matrices. This problem is called the \emph{primal problem}. The \emph{dual problem}, which can also be cast as an SDP, is:
\begin{align}
d^\star:=&\max_{Z} -\tr(F_0Z)\nonumber\\
\mbox{such that }&Z\geq 0,\nonumber\\
&\tr(F_i Z)=c_i, i=1,...,n.
\end{align}
It is easy to see that $p^\star\geq d^\star$. Moreover, if either problem has \emph{strictly feasible points}, i.e., if there exists $\vec{x}$ such that the matrix $F_0+\sum_{i=1}^n x_i F_i$ is positive definite, or a positive definite matrix $Z$ such that $\tr(F_i Z)=c_i$ holds for all $i$, then $p^\star=d^\star$ \cite{sdp}. This condition is called \emph{strong duality}

To numerically solve a strongly dual problem, SDP solvers try to find feasible points of both primal and dual. Given any such pair of points $(\vec{x}, Z)$, it holds that $-\tr(ZF_0)\leq p^\star\leq \vec{c}\cdot \vec{x}$. SDP solvers stop the optimization when the difference between the found primal and dual values is small, of order $10^{-7}$. At that stage, the SDP can be considered solved for all practical purposes.

Sometimes our goal is simply to find out whether there exist $\vec{x}$ satisfying the constraints of (\ref{primal_problem}). This is a \emph{feasibility problem}. The easiest way to tackle it is to define the SDP:
\begin{align}
p^\star:=&\min_{\vec{x}\in\R^n,\lambda\in\R} \lambda\nonumber\\
\mbox{such that }&\lambda\id+F_0+\sum_{i=1}^n x_i F_i\geq 0.
\label{feasibility_problem}
\end{align}
There exist $\vec{x}$ satisfying $F_0+\sum_i x_iF_i\geq 0$ iff $p^\star\leq 0$. Moreover, the problem has strictly feasible points, and so it is strongly dual. Thus, if the original problem is infeasible, then the solver will find solutions of the dual problem with value greater than zero. Each such solution can be regarded as a \emph{certificate of infeasibility}.

Frequent convex constraints can be expressed through linear matrix inequalities, and so they can be part of SDPs. For instance, suppose that the problem variables are $\vec{x}$ and consider the constraint
\begin{equation}
K\geq \sum_j |f_j(\vec{x})|,
\end{equation}
where each $f_j$ is a linear affine function. This constraint is equivalent to:
\begin{align}
&\mu_j\geq \pm f_j(\vec{x}),\nonumber\\
&K\geq \sum_j\mu_j,
\end{align}
where $(\mu_j)_j$ are new optimization variables (\emph{slack variables}, in usual parlour).

Another useful observation is that a constraint of the form:
\begin{equation}
K\geq \bra{f(\vec{x})}A\ket{f(\vec{x})},
\label{quad_ineq}
\end{equation}
where $A\geq 0$ and $\ket{f(\vec{x})}$ is a vector that depends affine linearly of $\vec{x}$, can be expressed through linear matrix inequalities. Indeed, let $A=\sum_{k=1}^ma_k\proj{\phi_k}$, with $a_k>0$ for all $k$. Then, condition (\ref{quad_ineq}) is equivalent to:
\begin{align}
&\left(\begin{array}{cc}K&\vec{g}(\vec{x})^T\\\vec{g}(\vec{x})&\mbox{diag}(a_1^{-1},...,a_m^{-1})\end{array}\right)\geq 0,
\label{quad_ineq_SDP}
\end{align}
with
\begin{equation}
g_j(\vec{x})=\braket{\phi_j}{f(\vec{x})}.
\end{equation}
The logic behind eq. (\ref{quad_ineq_SDP}) is that a matrix $\Gamma$ of the form
\begin{equation}
\Gamma=\left(\begin{array}{cc}A&C\\C^T&B\end{array}\right),
\end{equation}
with $B>0$, is positive semidefinite iff
\begin{equation}
A-CB^{-1}C^T\geq 0,    
\end{equation}
see \cite{horn13}. Substituting this relation in eq. (\ref{quad_ineq_SDP}), we find
\begin{align}
K-\bra{g(\vec{x})}\sum_{k=1}^ma_k\proj{k}\ket{g(\vec{x})}\geq 0.
\end{align}
From the definition of $\vec{g}(\vec{x})$, the left-hand side equals
\begin{equation}
K-\bra{f(\vec{x})}A\ket{f(\vec{x})}.
\end{equation}

Following standard literature in convex optimization, we have been assuming so far that the considered SDP is defined on real variables and has real constraints. The SDPs in this work have, however, both complex variables and complex-valued constraints. Nonetheless, all results for real SDPs reviewed here also hold for complex SDPs, once one replaces every occurrence of ``symmetric" by ``self-adjoint"; and $\bullet^T$, by $\bullet^\dagger$. 

Notice further that the matrix $A$ may have a large range of eigenvalues including zero as well as close to zero ones. In such cases, the numerical tolerance has to be imposed carefully in order to avoid a badly conditioned~\ref{quad_ineq_SDP} that may lead to imprecise results.

\section{Convergence of the SDP hierarchies for $A(\bar{E})$ and $U(\Delta E)$}
\label{app:average_and_var_constraints}
The derivation of the convergence bounds for the SDP hierarchies $(A_m(\bar{E}))_m$, $(U_m(\Delta E))_m$ is very similar, and based on this simple observation:
\begin{prop}
\label{prop:cut_off}
Let $\E=(\tilde{E}_0,...,\tilde{E}_m)$ be a finite energy spectrum, and let ${\cal F}:=(F_0,...,F_m)$ be a vector of non-negative numbers, with $F_m=\max_jF_j>0$, and let $\{P(a|x,t_j)\}_j$ be a dataset (with times $t_1,...,t_N$) such that there exist a normalized distribution $(p_j)_{j=0}^m$, a complex, positive semidefinite $N\times N$ matrix $\gamma$ and $(m+N)\times (m+N)$ matrices $\{\tilde{M}_{a|x}\}_{a,x}$ satisfying:
\begin{align}
&\sum_{j=0}^mp_{j}F_j\leq \mathbb{F},\label{aver_spec_app}\\
&\tilde{M}_{a|x}\geq 0, \forall a,x,\\
&\sum_a\tilde{M}_{a|x}=\left(\begin{array}{cc}\mbox{diag}(p_0,...,p_{m-1})&0\\0&\gamma\end{array}\right),\forall x,\\
&\gamma_{kk}=p_m,\forall k,\\
&P(a|x,t_j)=p(a|x,t_j,\tilde{M}, \E),\forall a,x,j,
\end{align}
for some $\mathbb{F}\in\R^+,\mathbb{F}<F_m$. Then, there exists a dataset $\check{P}$ with realization $(\C^m, \ket{\check{\psi}}, \check{H}, \check{M})$, with $\check{H}=\sum_{j=0}^{m-1}\tilde{E}_j\proj{j}$, such that
\begin{enumerate}
    \item $\|P(t_j)-\check{P}(t_j)\|\leq 2\sqrt{\frac{\mathbb{F}}{F_m}}$.
    \item The energy weights $(\check{p}_k)_k$ of $\ket{\check{\psi}}$ satisfy $\sum_j\check{p}_jF_j\leq \mathbb{F}$.
\end{enumerate}
    
\end{prop}

\begin{proof}
Note that, for condition (\ref{aver_spec_app}) to hold, it must be the case that
\begin{equation}
p_m\leq \frac{\mathbb{F}}{F_m}.
\label{bound_occu}
\end{equation}

Define the matrix
\begin{equation}
\Lambda:=\left(\begin{array}{cc}\mbox{diag}(p_1^{-1/2},...,p_{m-1}^{-1/2})&0\\0&\gamma^{-1/2}\end{array}\right),
\end{equation}
where, for simplicity, we have assumed that $p_j>0$, for all $j$ and that $\gamma>0$ (the general case can be derived, e.g., by perturbing both objects to make them positive or positive definite, respectively). Then we have that
\begin{equation}
M_{a|x}:=\Lambda \tilde{M}_{a|x} \Lambda^\dagger
\end{equation}
define POVMs. Next, define the vectors
\begin{equation}
\ket{\psi^+_j}:=\frac{1}{\sqrt{p_m}}\gamma^{1/2}\ket{j}.
\end{equation}
Then we have that
\begin{equation}
\braket{\psi^+_j}{\psi^+_k}=\frac{1}{p_m}\gamma_{jk};
\end{equation}
in particular, $\{\ket{\psi^+_j}\}_j$ are normalized.

It is immediate to verify that
\begin{equation}
P(a|x,t_j)=\bra{\psi_j}M_{a|x}\ket{\psi_j},
\end{equation}
with
\begin{equation}
\ket{\psi_j}:=\sum_{k=0}^{m-1}\sqrt{p_k}e^{-iE_kt_j}\ket{k}\oplus\sqrt{p_{m}}\ket{\psi^+_j}.
\end{equation}
We next define the normalized vector
\begin{equation}
\ket{\check{\psi}}=\sum_{j=0}^{m-1}\sqrt{\frac{p_j}{1-p_m}}\ket{j}.
\end{equation}
Notice that
\begin{align}
&\sum_{k=0}^{m-1}\frac{p_k}{1-p_m}F_k=\sum_{k=0}^{m-1}p_kF_k+p_m\sum_{k=0}^{m-1}\frac{p_k}{1-p_k}F_k\nonumber\\
&\leq \sum_{k=0}^{m-1}p_kF_k+p_mF_m\leq \mathbb{F}.
\end{align}
For $\check{H}:=\sum_{k=0}^{m-1}E_k\proj{k}$, the dataset
\begin{equation}
\check{P}(a|x,t_j):=\bra{\check{\psi}}e^{i\check{H}t_j}M_{a|x}e^{-i \check{H}t_j}\ket{\check{\psi}}
\end{equation}
thus satisfies the second condition of the Proposition. Moreover, it can be verified that
\begin{equation}
\bra{\psi_j}e^{-iHt_j}\ket{\check{\psi}}= \sqrt{1-p_m}.
\label{diffe_vectors}
\end{equation}
Hence,
\begin{align}
&\sum_a|P(a|x,t_j)-\check{P}(a|x,t_j)|\nonumber\\
&\leq \|\proj{\psi_j}-e^{-i \check{H}t_j}\proj{\psi}e^{i \check{H}t_j}\|_1=2\sqrt{p_m}\nonumber\\
&\leq 2\sqrt{\frac{\mathbb{F}}{F_m}},
\end{align}    
\end{proof}

Now, consider $P\in A_m(\bar{E})$. Condition (\ref{aver_spec}) implies that eq. (\ref{aver_spec_app}) holds with $F_k=\tilde{E}_k\geq 0$. Moreover, by eq. (\ref{minimum_m}), it holds that $E_m\geq \bar{E}$. Thus, by Proposition \ref{prop:cut_off} there exists $\check{P}\in A(\bar{E})$ such that
\begin{equation}
\|\check{P}(t_j)-p(t_j,\tilde{M},\E) \|\leq 2\sqrt{\frac{\bar{E}}{\eta}}.
\end{equation}
Combining this with eq. (\ref{approx_aver_E}), we have that
\begin{align}
\|\check{P}(t_j)-P(t_j)\|&\leq 2\sin\left(\frac{\eta \texttt{t}}{m}\right)+2\sqrt{\frac{\bar{E}}{\eta}}\nonumber\\
&\leq 2\left(\frac{\eta \texttt{t}}{m}+\sqrt{\frac{\bar{E}}{\eta}}\right).
\end{align}
The minimum of the right-hand side of the equation above over all possible values of $\eta$ is the right-hand side of (\ref{error_aver_hier}), and the minimizer is given by eq. (\ref{eta_optim}). Note that, for eq. (\ref{approx_aver_E}) to hold, it is necessary that 
\begin{equation}
m\geq \frac{2\eta\max_j|t_j|}{\pi}.
\label{rel_m_and_E_plus}
\end{equation}
Substituting the optimal $\eta$, we find that this condition is equivalent to 
\begin{equation}
m\geq \frac{2\bar{E}\texttt{t}}{\pi^3},
\end{equation}
which is already implied by eq. (\ref{minimum_m}).

Now, for $P\in U_{\tilde{m}}(\Delta E)$, we use the correspondence $F_{k}=\tilde{E}^2_{k-\tilde{m}+1}$. From eqs. (\ref{min_m_var}) and (\ref{def_eta_var}), it follows that $\eta^2=\tilde{F}_m\geq \mathbb{F}=:(\Delta E)^2$. By Proposition \ref{prop:cut_off}, there thus exists $\check{P}$ with 
\begin{equation}
\|\check{P}(t_j)-p(t_j,\tilde{M},\E) \|\leq 2\frac{\Delta E}{\eta}.
\end{equation}
satisfying 
\begin{equation}
\bra{\check{\psi}}\check{H}^2\ket{\check{\psi}}\leq (\Delta E)^2,
\end{equation}
It thus follows that
\begin{equation}
\bra{\check{\psi}}\check{H}^2\ket{\check{\psi}}-\bra{\check{\psi}}\check{H}\ket{\check{\psi}}^2\leq (\Delta E)^2,
\end{equation}
and so $\check{P}\in U(\Delta E)$. Combining this with eq. (\ref{approx_var}), we find that $\check{P}$ satisfies:
\begin{align}
\|\check{P}(t_j)-P(t_j)\|&\leq 2\sin\left(\frac{\eta \texttt{t}}{\tilde{m}}\right)+2\frac{\Delta E}{\eta}\nonumber\\
&2\left(\frac{\eta \texttt{t}}{\tilde{m}}+\frac{\Delta E}{\eta}\right).
\end{align}
Minimizing with respect to $\eta$, we obtain the right-hand side of eq. (\ref{error_var_hier}), with minimizer (\ref{def_eta_var}).

\section{Decay matrices}
\label{app:overlaps}
The goal of this section is to arrive at a practical characterization of the following cone of complex matrices.
\begin{defin}
The set of \emph{decay matrices} $G(\vec{t})$ is the closure of the cone of square matrices $\gamma$ such that
\begin{equation}
\gamma_{jk}=\int_{-\infty}^\infty dE e^{-iE(t_k-t_j)}\rho^+(E),j,k=1,...,N,
\label{def_G}
\end{equation}
for some non-normalized measure $\rho^+(E) dE$.
\end{defin}
To motivate it, consider this family of cones of matrices.
\begin{defin}
Let $t_1,...,t_N, E^+\in\R$. The cone $G(E^+,\vec{t})$ is the closure of the set of matrices $\gamma$ of the form:
\begin{equation}
\gamma_{jk}=\int_{E^+}^\infty dE \rho^+(E)e^{-iE(t_j-t_k)},j,k=1,...,N,
\end{equation}
where $\rho^+(E)dE$ is a (possibly non-normalized) measure.
\end{defin}
By definition, for each $\gamma\in G(E^+,\vec{t})$ and each $\delta>0$, there exists a (non-normalized) state $\ket{\psi^+}$ in some Hilbert space $\H^+$ and a Hamiltonian $H^+\geq E^+$ such that 
\begin{equation}
|\gamma_{jk}-\bra{\tilde{\psi}^+}e^{-iH^+(t_k-t_j)}\ket{\tilde{\psi}^+}| \leq \delta,\forall j,k.
\label{hamil2gamma}
\end{equation}
The set $G(E^+,\vec{t})$ thus plays a role in the characterization of the soft sets $S(\E;\epsilon), S(E^+;\epsilon)$, see Sections \ref{sec:charact_soft}, \ref{sec:charact_soft_continuous}.

We next show that $G(E^+,\vec{t})$ is, in fact, independent of $E^+$. From now on, we will thus consider 
$G(\vec{t})$.
\begin{lemma}
\label{lemma_indent_Gs}
$G(E^+,\vec{t})=G(-\infty,\vec{t})=G(\vec{t})$.
\end{lemma}
\begin{proof}
Clearly $G(E^+,\vec{t})\subset G(-\infty,\vec{t})$, so we next prove the opposite inclusion relation. 

Let $\gamma\in G(-\infty,\vec{t})$. Then, for any $\delta>0$, there exists a measure $\rho(E)dE$ over the reals such that the matrix
\begin{equation}
\gamma'_{jk}:=\int_{-\infty}^\infty dE \rho(E)e^{-iE(t_j-t_k)}
\label{measu_inf}
\end{equation}
satisfies
\begin{equation}
|\gamma_{jk}-\gamma'_{jk}|\leq \delta,\forall j,k.
\end{equation}
Since $(\gamma'_{jk})_{jk}$ is a conic combination of matrices of the form
\begin{equation}
(e^{-iE(t_k-t_j)})_{jk},
\end{equation}
for $E\in \R$, by Caratheodory's theorem we can choose $\rho(E)dE$ to be of the form
\begin{equation}
\rho(E)=\sum_{j=1}^L\lambda_j\delta(E-E_j),
\end{equation}
with $L\leq\frac{N(N-1)}{2}+1$. Now, let $E_1=\min_j E_j$. If $E_1\geq E^+$, then $\gamma'\in G(E^+,\vec{t})$ and we are done. Otherwise, choose $\Delta\geq E^+-E_1$ such that
\begin{equation}
|e^{-i\Delta t_j}-1|<\delta, \forall j.
\end{equation}
The existence of such an energy difference $\Delta$ is guaranteed from the fact that the trajectory
\begin{equation}
(e^{-i\Delta t_1},...,e^{-i\Delta t_N}):\Delta\in\R
\end{equation}
is almost periodic in $[0,2\pi)^N$ \cite{chávez2020}. Next, define the distribution
\begin{equation}
\rho^+(E):=\rho(E-\Delta).
\end{equation}
By definition, $\rho^+(E)$ has support in $[E^+,\infty)$. Moreover,
\begin{align}
\gamma^+_{jk}&:=\sum_l\lambda_l e^{-i(E_l-\Delta)(t_j-t_k)}\nonumber\\
&=\sum_l\lambda_l e^{-iE_l(t_j-t_k)}+r^\delta_{jk}\nonumber\\
&=\gamma'_{jk}+r^\delta_{jk},
\end{align}
with
\begin{equation}
|r^\delta_{jk}|\leq \delta\sum_l\lambda_l\leq (\gamma_{11}+\delta)\delta,
\end{equation}
where we invoked the relations $|\gamma_{11}-\gamma'_{11}|\leq \delta$ and $\gamma'_{11}=\sum_l\lambda_l$.

Thus, for any $\delta>0$, there exists $\gamma^+\in G(E^+,\vec{t})$ such that
\begin{equation}
|\gamma_{jk}-\gamma^+_{jk}|\leq \delta+(\gamma_{11}+\delta)\delta.
\end{equation}
Since $G(E^+,\vec{t})$ is closed, it follows that $\gamma\in G(E^+,\vec{t})$.
    
\end{proof}

What kind of constraints is $G(\vec{t})$ subject to? Clearly, for any $\vec{t}=(t_1,...,t_N)$, we have that $G(\vec{t})\subset C_N$, where $C_N$ denotes the cone of correlation matrices (namely positive semidefinite matrices with identical diagonal terms). However, as we will next see, for general $\vec{t}$ it is not the case that $G(\vec{t})= C_N$.

\subsection{The general case}
Given measurement times $t_1,...,t_N$, let $n$ denote the smallest natural number such that
\begin{equation}
t_k=\sum_{j=1}^n a_{kj}s_j+t_0, \mbox{ for } k=1,...,N,
\label{rational_decomp0}
\end{equation}
where $\{a_{kj}\}_{kj}\subset \Z$, $t_0\in\R$, $\{s_j\}_j\subset \R$. Note that it must hold that the generators $\{s_j\}_j$ are \emph{non-congruent} numbers, i.e., there does not exist a non-zero vector of rational numbers $(q_j)_j$ such that
\begin{equation}
\sum_{j=1}^nq_js_j=0.
\label{congruency}
\end{equation}
Otherwise, one could express one of the numbers $\{s_j\}_j$ as a rational linear combination of the others. Dividing the rest by the minimum common denominator, we would then have another decomposition like (\ref{rational_decomp0}) with one generator less.

\begin{lemma}
\label{lemma_all_Gs}
Let $t_1,...,t_N$ admit a decomposition like (\ref{rational_decomp0}), with non-congruent $\{s_j\}_j$. Then, $G(\vec{t})$ equals the cone $Z(a)$ of matrices generated by
\begin{equation}
\gamma_{jk}(z)=\prod_{l=1}^nz_l^{a_{kl}-a_{jl}},
\label{gamma_convex}
\end{equation}
with $z\in \mathbb{T}^n:=\{z\in\C^n:|z_l|=1, \forall l\}$.
\end{lemma}
\begin{proof}
Let $\gamma\in G(\vec{t})$. Then, for any $\delta>0$, there exists a (non-normalized) measure $\rho(E)dE$ such that the matrix $\gamma'$, defined by
\begin{align}
\gamma'_{jk}:=\int \rho(E)dE e^{-iE(t_k-t_j)},
\end{align}
satisfies $|\gamma_{jk}-\gamma'_{jk}|\leq \delta$, for all $j,k$. Expressing $\{t_j\}_j$ in terms of $t_0$ and the generators, we have that
\begin{align}
\gamma'_{jk}=\int \rho(E)dE e^{-iE(\sum_l(a_{jl}-a_{kl})s_l)},
\end{align}
Define $z_l(E):=e^{-iEs_l}$. Then we have that
\begin{align}
\gamma'_{jk}=\int \rho(E)dE \prod_lz_l(E)^{a_{kl}-a_{jl}}.
\end{align}
Thus, $\gamma'\in Z(a)$. Since $\delta>0$ is arbitrary, this proves that $G(\vec{t})\subset Z(a)$.

Conversely, suppose that $\gamma \in Z(a)$. Then there exists a measure $\rho(z)dz$ with support in $\mathbb{T}^n$ such that 
\begin{equation}
\gamma_{jk}=\int \rho(z)dz\gamma_{jk}(z).
\end{equation}

Since $\{s_j\}_j$ are not congruent, the trajectory 
\begin{equation}
(e^{-iEs_1},...,e^{-iEs_n}):E\in\R
\end{equation}
is dense in $\mathbb{T}^n$. This implies that, for any $\delta$, there exists $E(z)\in \R$ such that 
\begin{equation}
|\gamma_{jk}(z)-e^{-iE(z)(t_j-t_k)}|\leq \delta,
\end{equation}
for all $j,k$. Thus, the matrix $\gamma'\in G(\vec{t})$, defined through
\begin{equation}
\gamma'_{jk}=\int \rho(z)dze^{-iE(z)(t_j-t_k)},
\end{equation}
satisfies
\begin{equation}
|\gamma'_{jk}-\gamma_{jk}|\leq \delta\int \rho(z)dz=\delta \gamma_{11}.
\end{equation}
Since $\delta$ can be made arbitrarily small, we then have that $Z(a)\subset G(\vec{t})$.

\end{proof}

\begin{remark}
\label{remark_func2pol}
There is a one-to-one correspondence between the set $F$ of linear functionals $h:Z(a)\to \R$ and the set of symmetric trigonometric polynomials of the form
\begin{equation}
h(z,z*)=\sum_{\vec{j}\in {\cal Z}(a)}h(\vec{j})z^{\vec{j}},
\label{polyZ}
\end{equation}
where $z^{\vec{j}}=\prod_{k=1}^nz^{j_k}$ and ${\cal Z}(a)\subset \mathbb{Z}^n$ is defined by
\begin{equation}
{\cal Z}(a)=\{(a_{j,l}-a_{k,l})_l:j,k=1,...,n\}.
\end{equation}
Indeed, define $\vec{a}^j:=(a_{j,l})_l$ and let $\gamma\in Z(a)$, with measure $\sigma(z)dz$. For any matrix $H$, it holds that
\begin{align}
&\sum_{jk}H_{jk}\gamma_{jk}\nonumber\\
&=\int dz\sigma(z)\sum_{jk}H_{jk}z^{\vec{a}^j-\vec{a}^k}\nonumber\\
&=\int dz\sigma(z)h(z,z^*),
\end{align}
where $h(z,z^*)$ is a polynomial of the form (\ref{polyZ}). If $\sum_{jk}H_{jk}\gamma_{jk}\in \R$ for all $\gamma\in Z(a)$, then $h$ must be symmetric, i.e., $h(z,z^*)\in \R$, for all $z\in \mathbb{T}^n$. In turn, any symmetric polynomial $h$ defines a linear functional on $Z(a)$ by 
\begin{equation}
\gamma(h):=\sum_{\vec{j}\in {\cal Z}(a)}h_{\vec{j}}\gamma_{k(\vec{j}),l(\vec{j})},
\end{equation}
where $k(\vec{j}),l(\vec{j})$ are any pair of functions such that $\vec{a}^{l(\vec{j})}-\vec{a}^{k(\vec{j})}=\vec{j}$, for $\vec{j}\in {\cal Z}(a)$.

\end{remark}

\begin{remark}
\label{remark_id_G}
For all $a$, $\id\in Z(a)$. In effect, consider the distribution $\sigma(z)dz$, with
\begin{equation}
\sigma(z)=\frac{1}{R^n}\sum_{j_1,....,j_n=0}^{R-1}\prod_{l=1}^n\delta\left(z_l-e^{-i\frac{2\pi j_l}{R}}\right),
\end{equation}
where $R=2\max_{j,k,l}|a_{kl}-a_{jl}|$. It can be verified that 
\begin{equation}
\int dz\sigma(z)\gamma_{jk}(z)=\delta_{jk}.    
\end{equation}

\end{remark}

The cone $Z(a)$ can be characterized through the Lasserre-Parrilo hierarchy of SDP relaxations \cite{lasserre} \cite{parrilo}. In a nutshell, one would consider the set $\R\langle (z_1,...,z_n)\rangle$ of polynomials in the complex variables $(z_1,...,z_n)$ of module $1$ and their complex conjugates $(\bar{z}_1,....,\bar{z}_N)=(z^{-1}_1,....,z^{-1}_N)$. Next, given some distribution $\rho(z)dz$ of weight $p_{n}$, one thinks of the properties of \emph{momenta} $y$ of the form
\begin{equation}
y(k_1,...,k_n)=\langle\prod_{j=1}^nz^{k_j}\rangle_\rho,
\end{equation}
where $k_1,...,k_n\in\Z$. For any $m\in\N$, it turns out that the \emph{moment matrix}
\begin{equation}
(W_m(y))_{\vec{j},\vec{k}}=y(k_1-j_1,...,k_n-j_n),
\end{equation}
labeled by the vectors $\vec{j},\vec{k}\in\{-m,...,m\}^n$, is positive semidefinite, see \cite{lasserre}. We relate the momenta $y$ with the matrix $\gamma$ through
\begin{equation}
y(a_{k,1}-a_{j,1},...,a_{k,n}-a_{j,n})=\gamma_{jk},\forall j,k.
\end{equation}
We are ready to provide an SDP relaxation for $G(\vec{t})$.
\begin{defin}
Let $\vec{t}$ admit a minimal decomposition  (\ref{rational_decomp0}), and let $m\geq \max_{i,j,k,l}|a_{ij}-a_{kl}|$. $G^m(\vec{t})$ is the set of all $N\times N$ matrices $\gamma$ such that there exists $(y(\vec{k}))_{\vec{k}\in\{-m,...,m\}^n}\in \C^{(2m+1)^n}$, with
\begin{align}
&W_m(y)\geq 0,\nonumber\\
&y(a_{k,1}-a_{j,1},...,a_{k,n}-a_{j,n})=\gamma_{jk},\forall j,k.
\end{align}
\end{defin}
From the above, it follows that $G(\vec{t})\subset G^m(\vec{t})$. Moreover, $\lim_{m\to\infty}G^m(\vec{t})=G(\vec{t})$ \cite{lasserre, parrilo}. 

\begin{remark}
\label{remark_func2pol_m}
The argument of remark \ref{remark_func2pol} generalizes to $G^m(\vec{t})$: the set of linear functionals on $G^m(\vec{t})$ is also in one-to-one correspondence with the set of symmetric polynomials $h(z,z^*)$. 
\end{remark}

The work of \cite{bounds_lasserre} provides bounds on the convergence of the Lasserre-Parrilo hierarchy when applied to solve polynomial optimization problems within $\mathbb{T}^n$. Such bounds allow us to prove the following convenient result.
\begin{prop}
\label{prop_G_lasserre}
Let $\gamma\in G^m(\vec{t})$, and define
\begin{equation}
d:=\left\lceil\frac{1}{2}\max_{j,k}|a_{jl}-a_{kl}|\right\rceil.
\label{def_d_K0}
\end{equation}
Then, for $m\geq 3d$ and $\varepsilon\leq 1$, with
\begin{align}
&\varepsilon:=N(N-1)\left[\left(1-\frac{6d^2}{m^2}\right)^{-n}-1\right],\nonumber\\
&\approx \frac{6d^2n N(N-1)}{m^2},
\label{def_varepsilon_d}
\end{align}
it holds that the matrix
\begin{equation}
\tilde{\gamma}=\frac{1}{1+\varepsilon}\gamma +\frac{\varepsilon}{1+\varepsilon}\gamma_{11}\id_N
\label{def_tilde_gamma}
\end{equation}
satisfies $\tilde{\gamma}\in G(\vec{t})$.

\end{prop}
In other words: for large $m$, it is enough to mix $\gamma\in G^m(\vec{t})$ a little bit with the $\gamma_{11}$-scaled identity matrix to arrive at a feasible decay matrix $\tilde{\gamma}$.

To prove the proposition, we first need to prove a minor lemma.
\begin{lemma}
\label{lemma_min2norm}
Let $h(z,z^*)$ be a symmetric, trigonometric polynomial with no constant term, i.e.,
\begin{equation}
h(z,z^*)=\sum_{\vec{j}\in Z}h_{\vec{j}}z^{\vec{j}}+h^*_{\vec{j}}z^{\vec{j}}z^{-\vec{j}},
\end{equation}
where $Z\subset \Z^n\setminus\{\vec{0}\}$ is such that, if $\vec{j}\in Z$, then $-\vec{j}\not\in Z$. Then,
\begin{equation}
-\min_{z\in \mathbb{T}^n}h(z,z^*)\geq \max_{\vec{j}\in Z}|h_{\vec{j}}|.
\label{upper_bound_min}
\end{equation}
Thus,
\begin{equation}
\|h\|_F:=\sum_{\vec{j}}|h_{\vec{j}}|\leq -|Z|\min_{z\in \mathbb{T}^n}h(z,z^*).
\label{min2norm}
\end{equation}
\end{lemma}
\begin{proof}
Let $\vec{s}\in Z$ be such that $|h_{\vec{s}}|\geq |h_{\vec{k}}|$, $\forall \vec{k}\in Z$, let $\theta\in\R$ be the complex phase of $h_{\vec{s}}$ (i.e., $h_{\vec{s}}=|h_{\vec{s}}|e^{i\theta}$) and let $R\in \N$ be such that $R\geq 2|j_k|$, for all $\vec{j}\in Z$. Consider the probability distribution
\begin{align}
&\sigma(z)dz\nonumber\\
&=\frac{1}{(2R+1)^n}\sum_{\vec{k}\in \{-R,...,R\}^n\}}\left(1-\cos\left(\frac{2\pi\vec{k}\cdot \vec{s}}{2R+1}-\theta\right)\right)\nonumber\\
&\cdot\prod_{l=1}^n\delta\left(z_l-e^{-i\frac{2\pi k_l}{R}}\right)dz.
\end{align}
One can verify that $\sigma(z)dz$ is indeed non-negative and normalized. Moreover, 
\begin{equation}
\int_{\mathbb{T}^n} \sigma(z)h(z,z^*)dz=-|h_{\vec{s}}|.
\end{equation}
This implies eq. (\ref{upper_bound_min}).

Eq. (\ref{min2norm}) follows from eq. (\ref{upper_bound_min}) and the obvious inequality 
\begin{align}
\sum_{\vec{j}\in Z}|h_{\vec{j}}|\leq |Z|\max_{\vec{j}}|h_{\vec{j}}|.
\end{align}
    
\end{proof}

\begin{proof}[Proof of Proposition \ref{prop_G_lasserre}]
Let $\gamma\in G^m(\vec{t})$ and let $\tilde{\gamma}$ be as defined by eq. (\ref{def_tilde_gamma}). Assume that $\tilde{\gamma}\not\in G(\vec{t})$. Since $G(\vec{t})$ is convex and closed, by the Hahn-Banach theorem there exists a functional $h$, which by Remark \ref{remark_func2pol} we can identify with a symmetric trigonometric polynomial with support in ${\cal Z}(a)$, such that
\begin{equation}
\hat{\gamma}(h)\geq 0, \forall \hat{\gamma}\in G(\vec{t}),
\label{separation}
\end{equation}
and $\tilde{\gamma}(h)<0$.

Let us decompose $h$ as $h(z,z^*)=\lambda+h'(z,z^*)$, with $\lambda\in\R$ and $h'(z,z^*)$ having no constant term. By eq. (\ref{separation}) and the identity $Z(a)=G(\vec{t})$, the polynomial $h$ must satisfy $h(z,z^*)\geq 0$, for all $z\in \mathbb{T]}_n$. 

By Lemma \ref{lemma_min2norm}, $\min_{z}h'(z,z)<0$ if $h'(z,z)\not=0$, which implies that $\lambda>0$. Re-scaling $h$, we can choose $\lambda=1$. In that case, also by Lemma \ref{lemma_min2norm}, we have that
\begin{equation}
\|h'\|_F\leq -N(N-1)\min_{z\in \mathbb{T}^n} h'(z,z^*) \leq N(N-1).
\label{bound_F_norm}
\end{equation}
Recall that $\gamma\in G^m(\vec{t})$. Adapting \cite[Theorem 1]{bounds_lasserre} to a scenario with unnormalized moment matrices, we have that
\begin{equation}
\gamma(h)\geq -\epsilon\gamma_{11},
\end{equation}
for
\begin{equation}
\epsilon\leq \|h'\|_F\left[\left(1-\frac{6r^2}{m^2}\right)^{-n}-1\right],
\end{equation}
where $r:=\left\lceil\frac{\mbox{deg}(h)}{2}\right\rceil$. Note that $r=d$, as defined in eq. (\ref{def_d_K0}).
By eq. (\ref{bound_F_norm}) we thus have that $\epsilon\leq\varepsilon$. It follows that
\begin{align}
&\tilde{\gamma}(h)=\frac{\varepsilon}{1+\varepsilon}\gamma_{11}+\frac{\gamma(h)}{1+\varepsilon}\nonumber\\
&\geq \frac{1}{1+\varepsilon}\gamma_{11}(\varepsilon-\epsilon)\geq 0.
\end{align}
Thus, $\tilde{\gamma}(h)\geq 0$ and we reach a contradiction. This implies that $\tilde{\gamma}\in G(\vec{t})$.

\end{proof}

\subsection{Exact SDP representability}
In the previous section, we provided a simple recipe to generate a complete hierarchy of SDP relaxations for $G(\vec{t})$. In this section, we examine under which circumstances the set $G(\vec{t})$ admits an exact SDP representation, i.e., no hierarchies are needed.

\subsubsection{Equidistant times}
\label{app:equidistant}
Suppose that the decomposition (\ref{rational_decomp0}) is such that $n=1$. That is, 
\begin{equation}
t_k=a_k s+t_0,
\label{chipped_eq}
\end{equation}
with $a_1,...,a_N\in\N$. We next show that the corresponding set $G(\vec{t})$ admits an SDP representation. 

In fact, for ease of notation, we will provide an SDP representation of $G(\vec{t})$ when $\vec{t}$ is a vector of equidistant times, i.e., there exist $\Delta,t_0'$ such that
\begin{equation}
t_k=k \Delta+t'_0.
\label{full_eq}
\end{equation}
Note that the times (\ref{full_eq}) are a superset of (\ref{chipped_eq}), for $s=\Delta$ and $t'_0=t_0+\Delta(\min_k a_k+1)$. Thus, decay matrices for the measurement times (\ref{chipped_eq}) are submatrices of decay matrices for the measurement times (\ref{full_eq}).

Let then $t_k=k\Delta+t'_0$, for $k=1,...,N$. If $\gamma\in G(\vec{t})$, then there exists some unnormalized measure $\rho(E)dE$ such that 
\begin{equation}
\gamma_{jk}=\int_{-\infty}^{\infty}dE \rho(E)e^{-iE\Delta(k-j)},
\end{equation}
which implies that
\begin{equation}
\gamma_{jk}=g(k-j),
\end{equation}
for some function $g$.

As it turns out, that and positive semidefiniteness are the only features characterizing $G(\vec{t})$ in this scenario.
\begin{lemma}
For $\vec{t}=(\Delta,2\Delta,...,N\Delta)+t'_0$, define the SDP-representable set of matrices:
\begin{equation}
\tilde{G}(\vec{t})=\{\gamma:\gamma\geq 0, \exists y,\gamma_{jk}=y(k-j),\forall j,k\}.
\end{equation}  
Then, $G(\vec{t})=\tilde{G}(\vec{t})$.
\end{lemma}
\begin{proof}
Clearly, $\tilde{G}(\vec{t})\supset G(\vec{t})$; to prove the lemma, it suffices to show that $\tilde{G}(\vec{t})\subset G(\vec{t})$. 

We will proceed by contradiction: suppose that $\gamma\in \tilde{G}(\vec{t})$, $\gamma\not\in \tilde{G}(\vec{t})$. By the Hahn-Banach theorem and Remarks \ref{remark_func2pol}, \ref{remark_func2pol_m}, there exists a symmetric trigonometric polynomial $p(z,z^*)$ of degree smaller than or equal to $N-1$ such that
\begin{enumerate}
    \item $p(z,z^*)\geq0$, for all $z\in\mathbb{T}$, or, equivalently, $\check{\gamma}(p)\geq 0$, for all $\check{\gamma}\in G(\vec{t})$.
    \item $\gamma(p)<0$.    
\end{enumerate}



By \cite[Lemma 2.4]{helton2006positivepolynomialsscalarmatrix}, any such non-negative polynomial $p$ must satisfy
\begin{equation}
p(z,\bar{z})=q(z)^*q(z),
\end{equation}
where $q(z)=\sum_{k=0}^{N-1}q_kz^k$ for some complex numbers $(q_k)_k$. We thus have that
\begin{equation}
\gamma(p)=\sum_{jk}q_j^*q_k\gamma_{k-j}=\sum_{jk}q_j^*q_k\gamma_{jk}=\bra{q}\gamma\ket{q}\geq 0.
\end{equation}
We reach a contradiction, which means that $\gamma\in G(\vec{t})$.

\end{proof}

\subsubsection{Non-congruent time differences}
\label{app:non-congruent}
Suppose that the $N$ measurement times $t_1,..., t_N$ are such that $\{t_k-t_N\}_{k=1}^{N-1}$ are non-congruent numbers. Then we can express the $N$ times as
\begin{align}
t_k=&s_k+t_N, \mbox{ for } k=1,...,N-1,\nonumber\\
&t_N,\mbox{ otherwise}.
\end{align}
By Lemma \ref{lemma_all_Gs}, we thus have that the set of normalized decay matrices $\bar{G}(\vec{t}):=G(\vec{t})\cap\{\gamma:\gamma_{11}=1\}$ equals the set of convex combinations of matrices $\gamma_{jk}(z)$ of the form
\begin{equation}
\gamma_{jk}(z)=\prod_{j=l}^{N-1}z_l^{\delta_{k,l}-\delta_{j,l}},
\end{equation}
with $z\in\mathbb{T}^{N-1}$.

For arbitrary $z\in\mathbb{T}^{N-1}$, $\tilde{z}_N\in \mathbb{T}$, define $\tilde{z}:=\left(\frac{z_1}{\tilde{z}_N},\frac{z_2}{\tilde{z}_N},...,\tilde{z}_N\right)\in\mathbb{T}^N$. Then, $\gamma(z)=\tilde{\gamma}(\tilde{z})$, with
\begin{equation}
\tilde{\gamma}_{jk}(\tilde{z}):=\prod_{j=l}^{N}\tilde{z}_l^{\delta_{k,l}-\delta_{j,l}}.
\end{equation}

Recapping, we have that $\bar{G}(\vec{t})$ equals the following set:
\begin{defin}
$\mbox{CUT}^N_\infty$ is the convex hull of the set of matrices $\{\tilde{\gamma}(z):z\in\mathbb{T}^N\}$.
\end{defin}

This set has appeared earlier in the literature \cite{lifting,lifting2}. Call $\bar{C}_N$ the normalization of the set $C_N$ of correlation matrices, i.e., $\bar{C}_N=\{\gamma\geq 0:\gamma_{jj}=1,\forall j\}$. In \cite{lifting} it is shown that
\begin{equation}
\mbox{CUT}^3_\infty=\bar{C}_3.
\end{equation}
Therefore, for $t_1-t_3,t_2-t_3$ non-congruent, $G(\vec{t})=C_3$, and thus it is is SDP representable. In \cite{lifting} it is observed that $\bar{C}_4$ is strictly greater than $\mbox{CUT}^4_\infty$: in fact, the set  $\mbox{CUT}^4_\infty$ is shown in \cite{lifting,lifting2} to satisfy some extra inequalities. Nonetheless, $\mbox{CUT}^4_\infty$ is conjectured to be SDP representable \cite{lifting2}. More precisely, the conjecture in \cite{lifting2} implies that $\mbox{CUT}^4_\infty=G_2(\vec{t})$.

\subsection{$G(\vec{t})$ is discontinuous}
\label{app:disco}
Consider an scenario with three equidistant times $\vec{t}=(0,\Delta,2\Delta)$. By Appendix \ref{app:equidistant}, we have that
\begin{equation}
\bar{G}(\vec{t})=\{\gamma\geq 0:\gamma_{jk}=g(j-k), g(0)=1\}.
\end{equation}

Now, for any $\delta>0$, there exists $\vec{t}'$, with $\|\vec{t}'-\vec{t}\|_{\infty}<\delta$ such that the real numbers $\vec{t}'_1-\vec{t}'_3$, $\vec{t}'_2-\vec{t}'_3$ are non-congruent. This implies, by Section \ref{app:non-congruent}, that
\begin{equation}
\bar{G}(\vec{t}')=\{\gamma\geq 0:\gamma_{jj}=1,\forall j\}.
\end{equation}
Note that this set is independent of $\delta$. In particular, it includes the matrix
\begin{equation}
\gamma=\left(\begin{array}{ccc}1&1&0\\1&1&0\\0&0&1\end{array}\right).
\end{equation}
$\gamma_{12}\not=\gamma_{23}$, and so $\gamma\not\in \bar{G}(\vec{t})$. Thus, $\lim_{\vec{t}'\to\vec{t}}\not=\bar{G}(\vec{t})$.

\section{Convex characterization of $S(\E;\epsilon)$ and convergence bounds for $S^k(\E;\epsilon)$ and $S^k(E^+;\epsilon)$}
\label{app:soft_constraints}
In Appendix \ref{app:overlaps}, it is proven that any decay matrix $\gamma\in G(\vec{t})$ can be arbitrarily well approximated by a matrix of the form (\ref{hamil2gamma}). This is the basis for the following characterization of the closure of $S(\vec{E};\epsilon)$:
\begin{lemma}
\label{lemma_high_energy}
$P\in \overline{S(\vec{E};\epsilon;\vec{t})}$ iff there exist a distribution $\{p_j\}_j$, matrices $\{\tilde{M}_{a|x}\}_{a,x}$ and $\gamma\in G(\vec{t})$ such that
\begin{align}
&p_m\leq \epsilon,\label{cond_soft1}\\
&\tilde{M}_{a|x}\geq 0,\forall a,x,\label{cond_soft2}\\
&\sum_a\tilde{M}_{a|x}=\left(\begin{array}{cc}\mbox{diag}(p_1,...,p_{m-1})&0\\0&\gamma\end{array}\right),\label{cond_soft3}\\
&P(a|x,t_j)=p(a|x,t_j,\tilde{M}),\forall a,x,j.\label{cond_soft4}
\end{align}
\end{lemma}
\begin{proof}

By Section \ref{sec:charact_soft}, conditions (\ref{cond_soft1}-\ref{cond_soft4}) define a relaxation of $S(\E; \epsilon)$. To prove the opposite implication, namely, that any $P$ satisfying the conditions must belong to $\overline{S(\E; \epsilon)}$, we assume that $p_1,...,p_{n}>0$ and $\gamma>0$; the case where some $p_k$'s are zero or $\gamma$ has zero eigenvalues can be dealt with by perturbing wither $\vec{p}$ or $\gamma$. In that last regard, we refer the reader to Remark \ref{remark_id_G} in Appendix \ref{app:overlaps}, which implies that, for any $\gamma\in G(\vec{t})$ and $\lambda\in [0,1]$, it holds that $\lambda\gamma+(1-\lambda)\gamma_{11}\id\in G(\vec{t})$.

Since $\gamma\in G(\vec{t})$, there exist \footnote{Technically speaking, $G(\vec{t})$ is the \emph{closure} of all matrices $\gamma$ admitting such a decomposition. Again for the sake of simplicity, we assume that such a decomposition exists for $\gamma$: otherwise, one could perturb $\gamma$ infinitesimally to be on the safe side.} a Hilbert space $\H^+$, $H^+\in B(\H^+)$, with $H^+> E^+=E_n$, and an unnormalized state $\ket{\psi^+}\in\H^+$ such that eq. (\ref{hamil2gamma}) holds, see Appendix \ref{app:overlaps}.

Define $\Lambda:=\sum_{j=1}^Ne^{-iH^+t_j}\ket{\psi^+}\bra{j}$, and note that $\gamma=\Lambda^\dagger\Lambda$. Since $\gamma>0$, the rank of $\Lambda$ is $N$. Let thus $\Lambda^{-1}$ denote its left inverse, i.e., 
\begin{equation}
\Lambda^{-1}\Lambda\ket{j}=\ket{j}, \mbox{ for } j=1,...,N.
\end{equation}
The matrix $\Lambda\Lambda^{-1}:={\cal P}$ is a projector onto the space spanned by $\H^+_E:=\{e^{-iH^+t_j}\ket{\psi^+}\}_j$.

Then, the matrices
\begin{align}
&\hat{M}_{a|x}:=\nonumber\\
&\left(\begin{array}{cc}\mbox{diag}(\vec{p}^{-1/2})&0\\0&\Lambda^{-1}\end{array}\right)^\dagger \tilde{M}_{a|x}\left(\begin{array}{cc}\mbox{diag}(\vec{p}^{-1/2})&0\\0&\Lambda^{-1}\end{array}\right)
\end{align}
are positive semidefinite. Moreover, by eq. (\ref{cond_soft3}), they also satisfy 
\begin{equation}
\sum_a \hat{M}_{a|x}=\left(\begin{array}{cc}\id&0\\0&(\Lambda^{-1})^\dagger\gamma\Lambda^{-1}\end{array}\right)=\left(\begin{array}{cc}\id&0\\0&{\cal P}\end{array}\right).
\end{equation}
The matrices 
\begin{equation}
M_{a|x}:=\hat{M}_{a|x},a\not=1,M_{1|x}=\hat{M}_{1|x}+\left(\begin{array}{cc}0&0\\0&\id-{\cal P}\end{array}\right)
\end{equation}
are therefore positive semidefinite and satisfy $\sum_{a}M_{a|x}=\id$ -they are POVMs.

Next, define the Hamiltonian $H:=\sum_{k=0}^{m-1} E_k\proj{k}\oplus H^+$ and the vector
\begin{equation}
\ket{\psi}:=\sum_{k=0}^{m-1}\sqrt{p_k}\ket{k}+\sqrt{p_n}\ket{\psi^+}.
\end{equation}
Then it is immediate that
\begin{equation}
e^{-iHt_j}\ket{\psi}=\sum_{k=0}^{n-1}\sqrt{p_k}e^{-iE_kt_j}\ket{k}\oplus\sqrt{p_n}\Lambda\ket{j}.
\end{equation}
Thus,
\begin{align}
&\bra{\psi}e^{iHt_j}M_{a|x}e^{-iHt_j}\ket{\psi}\nonumber\\
&=\bra{\tilde{\psi}_j}\tilde{M}_{a|x}\ket{\tilde{\psi}_j}\nonumber\\
&=p(a|x,t_j,\tilde{M})=P(a|x,t_j).
\end{align}
\end{proof}

Lemma \ref{lemma_high_energy} represents an SDP characterization of $\overline{S(\E;\epsilon;\vec{t})}$ iff $G(\vec{t})$ turns out to be SDP-representable. As explained in Appendix \ref{app:overlaps}, whether this is the case strongly depends on the exact values of the dataset's measurement times $\vec{t}$. In fact, except for a few cases, $G(\vec{t})$ must be characterized through hierarchies of SDPs $(G^k(\vec{t}))_k$. Replacing $G(\vec{t})$ by the corresponding SDP relaxation $G^k(\vec{t})$, we arrive at the SDP relaxation $S^k(\E;\epsilon;\vec{t})$ of $\overline{S(\E;\epsilon;\vec{t})}$, see Section \ref{sec:charact_soft}. Throughout the rest of this appendix we will take for granted that $G(\vec{t})$ is not SDP-representable. Consequently, we will try to estimate how fast $S^k(\E;\epsilon)$ converges to the closure of $S(\E;\epsilon)$.

Given a vector of measurement times $\vec{t}$, let $n$ denote the smallest natural number such that
\begin{equation}
t_k=\sum_{j=1}^n a_{kj}s_j+t_0, \mbox{ for } k=1,...,N,
\label{rational_decomp}
\end{equation}
for some $\{s_j\}_j\subset \R$, $\{a_{kj}\}_{j,k}\subset \Z$. Define
\begin{equation}
d:=\left\lceil\frac{1}{2}\max_{j,k,l}|a_{jl}-a_{kl}|\right\rceil.
\label{def_d_K}
\end{equation}
and
\begin{equation}
\varepsilon_k:=N(N-1)\left[\left(1-\frac{6d^2}{k^2}\right)^{-n}-1\right].
\label{def_varepsilon_k}
\end{equation}
The next lemma estimates how fast $S^k(\E;\epsilon)$ converges to the closure of $S(\E;\epsilon)$.

\begin{lemma}
\label{lemma_S_sup_m2S}
Let $P\in S^k(\vec{E};\epsilon;\vec{t})$, $n\geq 3d$ and $\varepsilon_k\leq 1$. Then, there exists $\hat{P}\in S(\vec{E};\epsilon;\vec{t})$ such that
\begin{equation}
\|P-\hat{P}\|\leq \frac{2\varepsilon_k}{1+\varepsilon_k},
\end{equation}
with $\varepsilon_k$ as given in eq. (\ref{def_varepsilon_k}).
\end{lemma}
Taking the limit $k\gg 1$, we arrive at eq. (\ref{conv_soft_hier}).
\begin{proof}
Let $\{\tilde{M}_{a|x}\}_{a,x}$ be positive semidefinite matrices such that $P(a|x,t_j)=p(a|x,t_j,\tilde{M})$ and
\begin{equation}
\sum_a \tilde{M}_{a|x}=\left(\begin{array}{cc}\mbox{diag}(p_0,...,p_{n-1})&0\\
0&\gamma\end{array}\right),
\end{equation}
with $\gamma\in G^k(\vec{t})$. By Proposition \ref{prop_G_lasserre}, we have that
\begin{equation}
\hat{\gamma}:=\frac{1}{1+\varepsilon_k}(\gamma +\varepsilon_k\gamma_{11}\id)\in G(\vec{t}).
\end{equation}
Now, define the matrices
\begin{align}
&\hat{M}_{a|x}:=\frac{1}{1+\varepsilon_k}\tilde{M}_{a|x}\nonumber\\
&+\frac{\varepsilon_k}{|A|(1+\varepsilon_k)}\left(\begin{array}{cc}\mbox{diag}(1-\gamma_{11},0,...0)&0\\0&\gamma_{11}\id\end{array}\right).
\end{align}
These matrices are positive semidefinite. In addition, they satisfy
\begin{equation}
\sum_{a}\hat{M}_{a|x}=\left(\begin{array}{cc}D&0\\0&\hat{\gamma}\end{array}\right),
\end{equation}
for some diagonal matrix $D$. Thus, $\hat{P}(a|x,t_j):=p(a|x,t_j,\hat{M})$ satisfies $\hat{P}\in S(\vec{E};\epsilon;\vec{t})$ by Lemma \ref{lemma_high_energy}.

Finally, note that
\begin{equation}
P(a|x,t_j)-\hat{P}(a|x,t_j)=\frac{\varepsilon_k}{1+\varepsilon_k}\left(P(a|x,t_j)-\frac{1}{|A|}\right).
\end{equation}
It follows that
\begin{equation}
\sum_a|P(a|x,t_j)-\hat{P}(a|x,t_j)|\leq \frac{2\varepsilon_k}{1+\varepsilon_k},
\end{equation}
which proves the statement of the lemma.

\end{proof}

We next exploit the previously derived lemma to estimate how close $S_m^k(E^+;\epsilon)$ is to $S(E^+;\epsilon)$.
\begin{prop}
\label{prop_not_asymptotic}
Let $P\in S^k_m(E^+;\epsilon,\vec{t})$, $n\geq 3d$ and $\varepsilon_k\leq 1$. Then, there exists $\hat{P}\in  S(E^+;\epsilon,\vec{t})$ such that
\begin{equation}
\|\check{P}(t_j)-P(t_j)\|\leq \frac{2\varepsilon_k}{1+\varepsilon_k}+2\sin\left(\frac{E^+\texttt{t}}{m}\right),\forall j.
\end{equation}
\end{prop}

\begin{proof}
Let $P\in S^k_m(E^+;\epsilon,\vec{t})$. By definition, for $0\leq E_0\leq...\leq E_{n-1}$, there exists $\tilde{P}\in S^k(\vec{E};E^+\epsilon,\vec{t})$ such that
\begin{equation}
\|\tilde{P}(t_j)-P(t_j)\|\leq 2\sin\left(\frac{E^+\texttt{t}}{m}\right),\forall j.
\end{equation}
By Lemma \ref{lemma_S_sup_m2S}, we have that there exists $\hat{P}\in S(\vec{E};E^+\epsilon,\vec{t})\subset S(E^+;\epsilon,\vec{t})$ such that 
\begin{equation}
\|\check{P}(t_j)-\tilde{P}(t_j)\|\leq \frac{2\varepsilon_k}{1+\varepsilon_k},\forall j.
\end{equation}
By the triangular inequality, it hence follows that
\begin{equation}
\|\check{P}(t_j)-P(t_j)\|\leq \frac{2\varepsilon_k}{1+\varepsilon_k}+2\sin\left(\frac{E^+\texttt{t}}{m}\right).
\end{equation}

\end{proof}

\section{Self-testing datasets: the proofs}
\label{app:ST}
The main goal of this appendix is to prove the following result:
\begin{prop}
\label{prop_self_test}
Let $P\in S(E^+)$, with realization $(\H,\rho,H,M)$, satisfy
\begin{equation} \label{eq:approx_data}
\sum_a|P(a|t_j)-\mathbf{D}_N(a|t_j)|\leq \delta,
\end{equation}
for $j=0,...,N-1$. Then, there exists an isometry $V$ such that
\begin{align}
&VM_ae^{-iHt}\rho e^{iHt}V^\dagger\nonumber\\
&=\bar{M}^N_ae^{-i\bar{H}^Nt}\proj{\bar{\psi}^N} e^{i\bar{H}^Nt}\otimes \gamma_{\text{junk}}\nonumber\\
&+O\left(\delta^{1/64}\right)+O\left(\delta^{1/8}t\right),
\label{eq_prop_self_test}
\end{align}
for $a=1,2$ and some normalized quantum state $\gamma_{\text{junk}}$.
\end{prop}

Its proof is based on the insight that certain (noisy) statistics essentially fix the energy distribution of the state generating them, as captured by the following lemma.
\begin{lemma}
\label{lemma_density_noise}
Let $N\in\N$, let $\ket{\psi}$ be a normalized state with respect to some Hamiltonian $H$. Suppose that the hard support constraint holds, i.e., $\mbox{supp}(\sigma(E))\in[0,E^+]$, where $\sigma(E)dE$ is the energy density of $\ket{\psi}$, and that
\begin{equation} \label{eq:overlap}
|\bra{\psi}e^{iH(t_0-t_k)}\ket{\psi}|\leq \epsilon, k=1,...,N-1,
\end{equation}
with
\begin{equation}
t_k:=\frac{2\pi (N-1)k}{N E^+}.
\end{equation}
Then,  
\begin{equation} \label{eq:energy_density}
\sum_{k=0}^{N-1} \int_{\frac{E^+(k-\epsilon^{1/4})}{N-1}}^{\frac{E^+(k+\epsilon^{1/4})}{N-1}}\sigma(E) dE \gtrapprox 1 -O(\sqrt{\epsilon}).
\end{equation}
\end{lemma}

Lemma~\ref{lemma_density} ion the main text follows as $\epsilon \rightarrow 0$.

\begin{proof}
Condition~\eqref{eq:overlap} implies that
\begin{equation}
\left| \int_0^{E^+} e^{i \frac{2 \pi (N-1) k E}{NE^+}}\sigma(E) dE \right| \leq \epsilon, \ k=1, \ldots, N-1.
\end{equation}
Now let us consider the auxiliary function
\begin{align}
V(E) :&= \sin \left( \frac{\pi (N-1) E}{E^+}\right)  \prod_{j=1}^{N-2} \sin \left( \frac{j \pi}{N} -  \frac{\pi (N-1) E}{N E^+} \right) 
\label{def_V}
\end{align}
Notice that $V(E)=0$ iff $E=\frac{E^+ k}{N-1}$. Moreover, $V(E) \geq 0$ for $E\in [0,E^+]$. This can be seen by noticing that, for very small $E>0$, all factors in the product (\ref{def_V}) are positive. The first factor changes sign at $E\in\{\frac{E^+k}{N-1}:k=1,...,N-2\}$, while factors of the form 
\begin{equation}
V_j(E):=\sin \left( \frac{j \pi}{N} -  \frac{\pi (N-1) E}{N E^+} \right)
\end{equation}
only change sign at $E=\frac{E^+j}{N-1}$. Each sign change of the first factor is therefore canceled by the sign change of $V_j(E)$, and so $V(E)$ remains non-negative in $[0,E^+]$. 

In addition, if one expands the sine functions in eq. (\ref{def_V}) as sums of imaginary exponentials, one sees that $V(E)$ is a linear combination of the functions 
\begin{equation}
\{e^{i\frac{2\pi (N-1) kEj}{NE^+}}:j=1,...,N-1\}    
\end{equation}
and their complex conjugates. More specifically,
\begin{equation}
    V(E)=\frac{1}{2^{N-1}}\sum_{j=1}^{N-1} \mu_{j}^+ e^{i \frac{2 \pi (N-1) j E}{NE^+}} + \mu_{j}^- e^{-i \frac{2 \pi (N-1) j E}{NE^+}},
\end{equation}
where each $\mu_j^{\pm}=\frac{1}{i^{N-1}}\sum_k c_k  e^{i \frac{k \pi}{N}}$, where $c_k$ is an integer that counts the number of terms with the respective exponential, taking the sign into account; notice that $k$ may be negative.

The above implies that
\begin{align}
\int_0^{E^+} &V(E) \sigma(E) dE \nonumber \\ 
&= \frac{1}{2^{N-1}} \sum_{j=1}^{N-1} \mu_j^+ \int_0^{E^+}  e^{i \frac{2 \pi (N-1) j E}{NE^+}} \sigma(E) dE \nonumber \\
& \qquad \qquad \ \ + \mu_{j}^- \int_0^{E^+} e^{-i \frac{2 \pi (N-1) j E}{NE^+}} \sigma(E) dE \nonumber\\ 
&\leq \frac{1}{2^{N-1}} \sum_{j=1}^{N-1} |\mu_j^+ |\epsilon  + |\mu_{j}^-|\epsilon \nonumber\\
&\leq \epsilon, \label{eq:ve}
\end{align}
where in the second to last inequality we take the absolute value and use the triangle inequality; the last inequality follows as there are ${2^{N-1}}$ terms in the decomposition of $V(E)$ in total and each has an absolute value $\frac{1}{{2^{N-1}}}$.

Let us further define the energy range $E_{\delta} = \bigcup_{k=1}^{N-1} [\frac{E^+(k-1)}{N-1}+\delta, \frac{E^+k}{N-1} -\delta]$. Since $V(E)$ is positive, we have that
\begin{equation}
    \epsilon \geq \int_{E_{\delta}} V(E) \sigma(E) dE
\end{equation}
It is now helpful to rewrite 
\begin{equation}
V(E)= \frac{\sin^2 \left( \frac{\pi (N-1) E}{E^+}\right)}{2^{N-1} \sin \left(\frac{\pi (N-1) E}{N E^+} \right) \sin \left( \frac{\pi}{N} + \frac{\pi (N-1) E}{NE^+}\right)}.
\end{equation}
For $E\in [0, E^+]$, one has that 
\begin{equation}
0 \leq \sin \left(\frac{\pi (N-1) E}{N E^+} \right) ,\sin \left( \frac{\pi}{N} + \frac{\pi (N-1) E}{NE^+}\right) \leq 1.
\end{equation}
This implies that 
\begin{align}
    \epsilon &\geq \int_{E_{\delta}} V(E) \sigma(E) dE \\
    &\geq \int_{E_{\delta}} \frac{1}{2^{N-1}} \sin^2 \left( \frac{\pi (N-1) E}{E^+} \right) \sigma(E) dE.
\end{align}
Let $\delta = \frac{\epsilon^{1/4} E^+}{\pi (N-1)}$, then on $E_\delta$ we have 
\begin{align}
    \int_{E_{\delta}} \frac{1}{2^{N-1}} \sin^2 &\left( \frac{\pi (N-1) E}{E^+} \right) \sigma(E) dE \nonumber \\
    &\geq \frac{1}{2^{N-1}} \sin^2 \left( \epsilon^{1/4} \right) \int_{E_{\delta}}  \sigma(E) dE 
\end{align}
and, for small $\epsilon$,
\begin{equation}
 \int_{E_{\delta}}  \sigma(E) dE  \lesssim O(\sqrt{\epsilon} \cdot 2^{N-1}).
\end{equation}

\end{proof}

Towards the proof of the self-testing result in Proposition~\ref{prop_self_test}, let us first prove the following two lemmas.

\begin{lemma}
\label{lemma_POVM_annihilation}
Let $\tr(\rho M)\leq \epsilon$, for $0\leq M\leq \id$. Then,
\begin{equation}
\|M\rho\|_1\leq \sqrt{\epsilon}.
\end{equation}
\end{lemma}
\begin{proof}
Let $\ket{\psi}$ be a purification of $\rho$. Then,
\begin{equation}
\epsilon \geq \bra{\psi}(M\otimes \id)\ket{\psi}\geq \bra{\psi}(M^2\otimes \id)\ket{\psi}.
\end{equation}
Thus, $\|(M\otimes \id)\ket{\psi}\|_2\leq \sqrt{\epsilon}$ and, furthermore, 
\begin{equation}
\|(M\otimes \id)\ket{\psi}\|_2 = \|(M\otimes \id)\ket{\psi}\bra{\psi}\|_1\geq \|M\rho\|_1,
\end{equation}
which proves the lemma.
\end{proof}

\begin{lemma}
\label{lem:overlap}
Let $P \in S(E^+)$ satisfy \eqref{eq:approx_data} for $j=1, \ldots, N-1$ with a pure state realization $(\H,\psi,H,M)$, then the states $\ket{\psi(t_j)}=e^{iHt_j}\ket{\psi}$ evolved to times $t_j$, for $j=1, \ldots, N-1 $ satisfy
\begin{equation}
    |\braket{\psi}{\psi(t_j)}| \leq \sqrt{2\delta-\delta^2}.
\end{equation}
\end{lemma}

\begin{proof}
First, notice that \eqref{eq:approx_data} implies that for any $j=1, \ldots, N-1$, we can lower bound \begin{equation}
    \tr({M}_0 (\ketbra{\psi}{\psi}-\ketbra{\psi(t_j)}{\psi(t_j)})) \geq 1-\delta. \label{eq:1}
\end{equation}
At the same time
\begin{align}
  \tr({M}_0 (\ketbra{\psi}{\psi} &-\ketbra{\psi(t_j)}{\psi(t_j)})) \nonumber \\
  &\leq \frac{1}{2} \| \ketbra{\psi}{\psi}-\ketbra{\psi(t_j)}{\psi(t_j)}))\|_1 \nonumber\\
  &=\sqrt{1- |\braket{\psi}{\psi(t_j)}|^2}.
\end{align}
Combining this with \eqref{eq:1} completes the proof.    
\end{proof}

We are now in a position to prove Proposition~\ref{prop_self_test}.

\begin{proof}[Proof of Proposition~\ref{prop_self_test}]
First, we note that any realisation $(\H,\rho,H,M)$ can be purified to a realisation $(\H,\psi,H,M)$. For the latter, we have by Lemma~\ref{lem:overlap} that 
$|\braket{\psi}{\psi(t_j)}| \leq O(\sqrt{\delta})$.
By Lemma~\ref{lemma_density_noise}, this implies for the energy density $\sigma(E) dE$ of $\ket{\psi}$ that 
\begin{equation} \label{eq:energy_density}
\sum_{k=0}^{N-1}\int_{\frac{E^+(k-\delta^{1/8})}{N-1}}^{\frac{E^+(k+\delta^{1/8})}{N-1}}\sigma(E) dE \gtrapprox 1-O(\delta^{1/4}), \ k=0,\ldots,(N-1).
\end{equation}

For $k=0,\ldots,N-1$, define the orthogonal projectors
\begin{equation}
\Pi_k := \int_{\frac{E^+ (k-\delta^{1/8})}{N-1}}^{\frac{E^+ (k+\delta^{1/8})}{N-1}} \ketbra{E}{E} dE,
\end{equation}
(extending the spectrum to the negative at $k=0$ is just a notational convenience), and call $\Pi$ their sum. We next introduce the state
\begin{equation}
\ket{\psi'}:=\frac{\Pi\ket{\psi}}{\|\Pi\ket{\psi}\|_2}.
\end{equation}
By eq. (\ref{eq:energy_density}), we have that $\ket{\psi}, \ket{\psi'}$ are very similar for small $\delta$. More concretely,
\begin{equation}
\|\proj{\psi}-\proj{\psi'}\|_1\leq O(\delta^{1/8}).
\label{diff_psi_psi_prime}
\end{equation}
Let us further introduce the Hamiltonian 
\begin{align} \label{binned}
    H':= \sum_{k=0}^{N-1} \frac{E^+ k}{N-1} \Pi_k,
\end{align}
and denote by $\ket{\psi'(t)}$ ($\ket{\psi(t)}$) the state that results from propagating $\ket{\psi'}$ ($\ket{\psi}$) with $H'$ ($H$) for a time $t$.

For $\delta t^{1/8}\leq \frac{\pi (N-1)}{E^+}$, through the same discretization arguments used in Section \ref{sec:charact}, we arrive at
\begin{align}
&|\bra{\psi'(t)}e^{-iHt}\ket{\psi'}|\geq \cos\left(\frac{E^+\delta^{1/8}t}{N-1}\right),
\end{align}
which implies that
\begin{equation}
\|\proj{\psi'(t)}-e^{-iHt}\proj{\psi'}e^{iHt}\|_1\leq O(t\delta^{1/8}).
\end{equation}
Thus,
\begin{align}
&\|\proj{\psi(t)}-\proj{\psi'(t)}\|_1\nonumber\\
&\leq \|e^{-iHt}\proj{\psi}e^{iHt}-e^{-iHt}\proj{\psi'}e^{iHt}\|_1\nonumber\\
&+\|e^{-iHt}\proj{\psi'}e^{iHt}-\proj{\psi'(t)}\|_1\nonumber\\
&\leq O(\delta^{1/8})+O(\delta^{1/8}t).
\end{align}
In particular, we have that
\begin{align}
&\|\proj{\psi(t_l)}-\proj{\psi'(t_l)}\|_1\leq O(\delta^{1/8}),
\label{diff_psi_psi_prime_t}
\end{align}
for $l=0,...,N-1$.
Due to the last equation and the fact that $P$ fits $\mathbf{D}_N$, it holds that
\begin{align}
&\tr(M_1\proj{\psi'(0)})\leq O(\delta^{1/8}),\nonumber\\
&\tr(M_0\proj{\psi'(t_l)})\leq O(\delta^{1/8}),l\not=0.
\label{primed_stats}
\end{align}

The state $\ket{\psi'}$ can be written as 
\begin{align}
    \ket{\psi'} = \sum_k \ket{\psi_j (k)} \ket{\phi_j(k)}_A,
\end{align}
where the states $\{\ket{\psi_j (k)}\}_j$ are an orthonormal basis of the energy eigenspace of $H'$ with eigenvalue $\frac{E^+k}{N-1}$, and the states $\ket{\phi_j(k)}$ are non-normalised, non-orthogonal states of $\rho$'s purifying system $\mathcal{H}_A$.
Now, the trace-class operators 
$\Phi_k:=\sum_j \ketbra{\phi_j(k)}{\phi_j(k)}$ satisfy
\begin{align}
    \| \Phi_{k_1}-\Phi_{k_2} \|_1 
    & \!= \! \frac{1}{N} \| \!\!\!\sum_{k,s=0}^{N-1} \!e^{\frac{2\pi i (k_1-k)s}{N}} \Phi_k \! - \!\! \! \sum_{k,s=0}^{N-1} \! e^{\frac{2\pi i (k_2-k)s}{N}} \Phi_k\|_1 \nonumber\\
    &\leq 2 \| \sum_{k=0}^{N-1} e^{\frac{-2\pi i k s}{N}} \Phi_k \|_1, 
\end{align}
where in the first line we used the geometrical series and in the second the triangle inequality. We further find, due to the polar decomposition, that
\begin{align}
    \| \sum_{k=0}^{N-1} e^{\frac{-2\pi i k s}{N}} \Phi_k \|_1 &= \max_{U, \text{ unitary}} |\tr( U \sum_{k=0}^{N-1} e^{\frac{-2\pi i k s}{N}} \Phi_k)| \nonumber \\
    &= \max_{U, \text{ unitary}} |\tr(\id \otimes U e^{-i H't_s} \ketbra{\psi'}{\psi'})| \nonumber \\    
   &\leq O(\delta^{1/16}).
\end{align} 
The explanation for the last inequality is as follows:
\begin{align}
&|\tr((\id \otimes U) e^{-i H't_s} \ketbra{\psi'}{\psi'})|\nonumber\\
&\leq |\tr((M_0 \otimes U) e^{-i H't_s} \ketbra{\psi'}{\psi'})|\nonumber\\
&+|\tr((M_1 \otimes U) e^{-i H't_s} \ketbra{\psi'}{\psi'})|\nonumber\\
&= |\tr((\id \otimes U)(M_0\otimes \id)| e^{-i H't_s} \ketbra{\psi'}{\psi'})\nonumber\\
&+|\tr((\id \otimes U) e^{-i H't_s} \ketbra{\psi'}{\psi'}(M_1\otimes \id)|.
\end{align}
Each of the last two terms is bounded by $O(\delta^{1/16})$. Indeed, by eq. (\ref{primed_stats}) we have that
\begin{align}
&\|(M_0\otimes \id)\ket{\psi'(t_s)}\|_2^2\nonumber\\
&=\bra{\psi'(t_s)}((M_0)^2\otimes \id)\ket{\psi'(t_s)}\nonumber\\
&\leq\bra{\psi'(t_s)}(M_0\otimes\id)\ket{\psi'(t_s)}\leq O(\delta^{1/8}),
\end{align}
and similarly for the second term. We conclude that
\begin{align}
\|\Phi_{k_1}-\Phi_{k_2} \|_1\leq O(\delta^{1/16}). 
\label{eq:phi}
\end{align}

Now, for each $k$, let us consider the vectors 
\begin{equation} 
    \ket{\psi_k'}= \sum_j \ket{\psi_j(k)} \ket{\phi_j(k)}.
\end{equation}
$\ket{\psi_k'}$ is a purification of $\Phi_k$, with purifying system $\mathcal{H}_k$ (the eigenspace of $H'$ to eigenvalue $\frac{E^+ k}{N-1}$).
Since $\tr(\sum_k\Phi_k)=1$, by \eqref{eq:phi} it follows that
\begin{equation}
|\tr(\Phi_k)-\frac{1}{N}|\leq O(\delta^{1/16}),\forall k.
\label{norm_approx}
\end{equation}
Let $k$ be such that $\mbox{dim}(\H_k)=\min_j\mbox{dim}(\H_j)$, and let $P_k$ be the projector onto the support of $\Phi_k$. Define:
\begin{align}
&\Phi''_j:= \frac{P\Phi_jP}{N\tr(\Phi_jP)},\mbox{ for }j\not=k,\nonumber\\
&\Phi''_k:= \frac{\Phi_k}{N\tr(\Phi_kk)},\nonumber\\
&\ket{\psi''_j}:= \frac{1}{\sqrt{N}}\frac{(\id\otimes P)\ket{\psi_j}}{\|(\id\otimes P)\ket{\psi_j}\|_2},\mbox{ for }j\not=k,\nonumber\\
&\ket{\psi''_k}:= \frac{1}{\sqrt{N}}\frac{\ket{\psi_k}}{\|\ket{\psi_k}\|_2}.
\end{align}
It can be verified that $\tr_{\H}(\proj{\psi''_j})=\Phi''_j$, for all $j$. In addition, by eqs. (\ref{eq:phi}) and (\ref{norm_approx}), we have that 
\begin{align}
&\|\Phi_j-\Phi''_k\|_1\leq O(\delta^{1/16}),\nonumber\\
&\|\ket{\psi''_j}-\ket{\psi_j}\|_2\leq O(\delta^{1/32}),
\end{align}
for $j=1,...,N$.
By assumption, the Schmidt number of the vectors $\{\ket{\psi''_j}\}_j$ is smaller than or equal to $d:=\mbox{dim}(\H_k)$. Call $d_k$ the dimension of the support of $\Phi_k$. Next, we modify the vector $\ket{\psi''_j}$ so that its reduced state in system $A$ is exactly equal to $\Phi_k$. 

To do so, we invoke Uhlmann's theorem \cite{nielsen00}: given the two states $\rho_k:=N\Phi''_k$, $\rho_j:=N\Phi''_j$, there exist purifications $\ket{\omega_k},\ket{\omega_j}\in \C^{d}\otimes \C^{d_k}$ such that
\begin{align}
&|\braket{\omega_k}{\omega_j}|^2=F(\rho_k,\rho_j)\nonumber\\
&:=\tr(\sqrt{\rho_1^{1/2}\rho_2\rho_1^{1/2}})^2.
\end{align}
Since all purifications are equivalent modulo local unitaries and $\mbox{dim}(\H_j)\geq d_k$, there exists an isometry $U_j$ such that
\begin{equation}
\sqrt{N}\ket{\psi''_j}=(U_j\otimes \id)\ket{\omega_j}.
\end{equation}
We therefore define the state
\begin{equation}
\ket{\tilde{\psi}_j}:=\frac{1}{\sqrt{N}}(U_j\otimes \id)\ket{\omega_k};
\end{equation}
from the arguments above, we have that
\begin{align}
&|\braket{\tilde{\psi}_j}{\psi''_j}|=\sqrt{F(\rho_k,\rho_j)}\nonumber\\
&\leq 1-\frac{N}{2}\|\Phi''_k-\Phi''_j\|_1=1-O(\delta^{1/16}).
\end{align}
Putting everything together, we have that the state
\begin{equation}
\ket{\tilde{\psi}}:=\frac{1}{\sqrt{N}}\sum_j\ket{\tilde{\psi}_j}
\end{equation}
satisfies
\begin{equation}
\|\ket{\tilde{\psi}}-\ket{\psi'}\|_2\leq O(\delta^{1/32}).
\label{diff_tilde_sin_tilde}
\end{equation}
Moreover, 
\begin{equation}
\tr_{\H}(\proj{\tilde{\psi}_j})=\Phi''_k=:\tilde{\Phi}.
\end{equation}

Let $\tilde{\Phi}=\sum_l \frac{\lambda_l}{N}\ketbra{\varphi_l}{\varphi_l}$ be the spectral decomposition of $\tilde{\Phi}$. Then, for $k\in\{0,...,N-1\}$, there must exist an orthonormal set of vectors $\{\ket{\hat{\psi}_l(k)}\}_l\subset \H_k$ and $\{\lambda_l\}_l\subset \R^+$, with $\sum_l\lambda_l=1$, such that
\begin{equation}
\ket{\tilde{\psi}_k}=\frac{1}{\sqrt{N}}\sum_{l}\sqrt{\lambda_l}\ket{\hat{\psi}_l(k)}\ket{\varphi_l}.
\end{equation}
Hence, we can write 
\begin{equation}
\tilde{\rho} =\tr_A(\proj{\tilde{\psi}})= \sum_l \lambda_l \ketbra{\hat{\psi}_l}{\hat{\psi}_l},
\label{decomp_rho_tilde}
\end{equation}
with
\begin{equation}
\ket{\hat{\psi}_l}:=\frac{1}{\sqrt{N}}\sum_{k=0}^{N-1}\ket{\hat{\psi}_l(k)}.
\end{equation}
Also, for $\rho':=\tr_A(\proj{\psi'})$, it holds that
\begin{align}
&\|\rho'-\tilde{\rho}\|_1\leq \|\proj{\psi'}-\proj{\tilde{\psi}}\|_1\nonumber\\
&\leq O(\delta^{1/32}),
\label{diff_rho_prime_rho_tilde}
\end{align}
where the last equality follows from eq. (\ref{diff_tilde_sin_tilde}). 

Combined with eq. (\ref{primed_stats}), the last relation implies that
\begin{align}
&\tr(M_1\tilde{\rho})\leq O(\delta^{1/32}),\nonumber\\
&\tr(M_0e^{-iHt_l}\tilde{\rho}e^{iHt_l})\leq O(\delta^{1/32}),l\not=0.
\end{align}
By Lemma \ref{lemma_POVM_annihilation}, this implies that
\begin{align}
&\|M_1\tilde{\rho}\|_1\leq O(\delta^{1/64}),\nonumber\\
&\|M_0e^{-iHt_l}\tilde{\rho}e^{iHt_l}\|_1\leq O(\delta^{1/64}),l\not=0.
\label{stats_tilde_rho}
\end{align}

We therefore introduce the modified POVM:
\begin{align}
&\tilde{M}_0:=\sum_j \proj{\hat{\psi}_j},\nonumber\\
&\tilde{M}_1:=\id-\tilde{M}_0,
\end{align}
which by decomposition (\ref{decomp_rho_tilde}) and eq. (\ref{stats_tilde_rho}) satisfies
\begin{align}
&\|(\tilde{M}_a-M_a)e^{-iHt_l}\tilde{\rho}e^{iHt_l}\|_1\leq O(\delta^{1/64}),
\label{N_times}
\end{align}
for $l=0,...,N-1$, $a=0,1$. In fact, the inequality above holds for general times $t$. Indeed, note that
\begin{align}
&e^{-iH't}\tilde{\rho}e^{i{H'}t}= \nonumber \\
&\sum_{l,m=0}^{N-1}\braket{\bar{\psi}^N(t_l)}{\bar{\psi}^N(t)}\braket{\bar{\psi}^N(t)}{\bar{\psi}^N(t_m)}e^{-i{H'}t_l}\tilde{\rho}e^{i{H'}t_m}.
\label{eq:anytime}
\end{align}
From this and eq. (\ref{N_times}) we thus infer that
\begin{align}
&\|(\tilde{M}_a-M_a)e^{-iHt}\tilde{\rho}e^{iHt}\|_1\leq O(\delta^{1/64}),
\label{all_times}
\end{align}
for all $t\in\R$, $a\in\{0,1\}$.

Call $H^\perp$ the orthogonal complement of the subspace spanned by $\{\hat{\psi}_l(k)\}_{l,k}\subset \H$, and let $\ket{\alpha_j}_j$ be an orthonormal basis thereof. Define the isometry
$V: \H \rightarrow \mathbb{C}^N \otimes \mathbb{C}^2 \otimes l_2(\C)$
\begin{align}
    V \ket{\hat{\psi}_l(k)} &:= \ket{k} \ket{0} \ket{l} \ \forall k,l \nonumber \\
    V \ket{\alpha_j} &:= \ket{0} \ket{1} \ket{j} \ \forall j,
\end{align}
 Applying this isometry we get 
 \begin{align}
 V \tilde{M}_a e^{-iH't} \tilde{\rho} e^{-iH't} V^\dagger = \bar{M}^N_a e^{-i\bar{H}^Nt} \bar{\rho} e^{-i\bar{H}^Nt} \otimes \gamma_{\text{junk}},
\end{align}
where $\gamma_{\text{junk}}=\ketbra{0}{0} \otimes \sum_j \lambda_j \ketbra{j}{j}$. 

To prove the statement of the lemma, it remains to bound the quantity
\begin{equation}
\|M_ae^{-iHt}\rho e^{iHt}-\tilde{M}_ae^{-iH't}\tilde{\rho} e^{iH't}\|_1.
\end{equation}
From eq. (\ref{diff_psi_psi_prime_t}) and the contractiveness of the trace, we have that
\begin{equation}
\|e^{-iHt}\rho e^{iHt}-e^{-iH't}\rho'e^{-iH't}\|_1\leq O(\delta^{1/8})+O(\delta^{1/8}t).
\end{equation}
Invoking eqs. (\ref{diff_rho_prime_rho_tilde}), (\ref{all_times}) and the triangular inequality, we conclude:
\begin{equation}
\|M_ae^{-iHt}\rho e^{iHt}-\tilde{M}_ae^{-iH't}\tilde{\rho} e^{iH't}\|_1\leq O(\delta^{1/8}t)+O(\delta^{1/64}).
\end{equation}

\end{proof}

\section{Proof of Lemma \ref{lemma_T_m_approx}}
\label{app:T_m_approx}
The proof relies on the following, general lemma.
\begin{lemma}
Let $Y\subset \R^n$ be a closed, convex cone with non-empty interior and let $\hat{n}\in\R^n$. Consider the following convex optimization problem:
\begin{align}
f^\star:=&\min \vec{f}\cdot \vec{y}\nonumber\\
\mbox{such that }&\vec{y}\in Y, \vec{y}\cdot \hat{n}=1,\nonumber\\
&\mathbb{A}\cdot \vec{y}\leq \vec{b}.
\label{problem_abstract}
\end{align}
Suppose that there exist $r>0$, and feasible $\vec{y}^0$ such that 
\begin{equation}
\vec{b}-\mathbb{A} \cdot \vec{y}^0>r \left(\begin{array}{c}1\\1\\\vdots\\1\end{array}\right).
\label{special_interior}
\end{equation}


Let $\|\|$ be any vector norm, and let $\vec{y}^+\in\R^n$ be such that 
\begin{enumerate}
    \item $\mathbb{A}\cdot\vec{y}^+\leq \vec{b}$.
    \item $\|\vec{y}^+-\vec{y}\|\leq \epsilon$, for some $\vec{y}\in Y$ with $\hat{n}\cdot \vec{y}=1$.
\end{enumerate}
Then, $\vec{f}\cdot{y}^+$ is lower bounded by 
\begin{equation}
f^\star-\epsilon\left(\max_{\|\vec{u}\|\leq 1}\|A\vec{u}\|_\infty\frac{\vec{f}\cdot\vec{y}^0-f^\star}{r}+\max_{\|\vec{u}\|\leq 1}|\vec{f}\cdot\vec{u}|\right).
\label{bound_accuracy}
\end{equation}

\end{lemma}
\begin{proof}
First we show that problem (\ref{problem_abstract}) has no duality gap. Take any point $\vec{y}$ belonging to the interior of $Y\cap\{\vec{y}:\vec{y}\cdot \hat{n}=1\}$. Then, there exists $\lambda\in (0,1)$ such that the interior point
\begin{equation}
\vec{y}^\lambda:=\lambda \vec{y}^0+(1-\lambda)\vec{y}
\end{equation}
satisfies the condition $\vec{b}>\mathbb{A}\cdot\vec{y}^\lambda$. Problem (\ref{problem_abstract}) therefore admits a strictly feasible point, and thus, by Slater's criterion \cite{Nocedal2006}, has no duality gap.

The dual of Problem (\ref{problem_abstract}) can be shown to be:
\begin{align}
f^\star:=&\max -\lambda-\vec{\mu}\cdot\vec{b}\nonumber\\
\mbox{such that }&\vec{f}+\lambda\hat{n}+\mathbb{A}^T\vec{\mu}\in Y^*,\nonumber\\
&\vec{\mu}\geq 0.
\label{abstract_dual}
\end{align}
For any $\delta>0$, there exists a feasible point of (\ref{abstract_dual}), with objective value greater than $f^\star-\delta$. Let $(\lambda,\vec{\mu})$ be any such point. Then we have that
\begin{align}
&0\leq (\vec{f}+\lambda\hat{n}+\mathbb{A}^T\vec{\mu})\cdot \vec{y}^0\nonumber\\
&=\vec{f}\cdot \vec{y}^0+\lambda +\vec{\mu}\cdot \mathbb{A}\cdot \vec{y}^0\nonumber\\
&\leq \vec{f}\cdot \vec{y}^0+\lambda+\vec{\mu}\cdot \vec{b}-r\sum_j\mu_j\nonumber\\
&\leq \vec{f}\cdot \vec{y}^0-f^\star +\delta-r\sum_j\mu_j.
\end{align}
It follows that
\begin{equation}
\sum_j\mu_j\leq \frac{1}{r}\left(\vec{f}\cdot \vec{y}^0-\vec{f}^\star +\delta\right).
\end{equation}
Since $\{\mu_j\}_j$ is bounded, then so is $\lambda$; indeed, by assumption,
\begin{equation}
-f^\star -\vec{\mu}\cdot\vec{b}\leq \lambda\leq -f^\star+\delta -\vec{\mu}\cdot\vec{b}.
\end{equation}
Since in the vicinity of the optimal objective value all dual variables are bounded, it follows that the dual problem (\ref{abstract_dual}) has an optimal solution $(\lambda^\star,\vec{\mu}^\star)$. Moreover, it satisfies
\begin{equation}
\|\vec{\mu}^\star\|_1\leq \frac{1}{r}\left(\vec{f}\cdot \vec{y}^0-\vec{f}^\star\right).
\label{bound_mu}
\end{equation}

Now, let $\vec{y}^+$ satisfy the conditions of the lemma. Then, there exists $\vec{y}\in Y$ satisfying $\vec{y}\cdot \hat{n}=1$, such that $\vec{y}^+=\vec{y}+\vec{\varepsilon}$, with $\|\vec{\varepsilon}\|\leq \epsilon$. Thus, we have:
\begin{align}
&0\leq \vec{f}\cdot\vec{y}+\lambda^\star+\vec{\mu}^\star\cdot \mathbb{A}\vec{y}\nonumber\\
&=\vec{f}\cdot\vec{y}^++\lambda^\star+\vec{\mu}^\star\cdot\mathbb{A}\cdot\vec{y}^+-\vec{\mu}^\star\cdot \mathbb{A}\vec{\varepsilon}-\vec{f}\cdot \vec{\varepsilon}\nonumber\\
&= \vec{f}\cdot\vec{y}^+-f^\star-\vec{\mu}^\star\cdot(\vec{b}-\mathbb{A}\vec{y}^+)-\vec{\mu}^\star\cdot \mathbb{A}\vec{\varepsilon}-\vec{f}\cdot \vec{\varepsilon}\nonumber\\
&\leq \vec{f}\cdot \vec{y}^+-\vec{f}^\star-\vec{\mu}^\star\cdot \mathbb{A}\vec{\varepsilon}-\vec{f}\cdot \vec{\varepsilon},
\end{align}
where we used the identity $\lambda^\star=-f^\star-\vec{\mu}^\star\cdot\vec{b}$ in the second line; and the inequalities $\mu\geq 0$, $\vec{b}-\mathbb{A}\vec{y}^+\geq0$ in the third one. From this point, one just needs to invoke the bound (\ref{bound_mu}) on $\|\vec{\mu}^\star\|_1$ to arrive at eq. (\ref{bound_accuracy}).
\end{proof}

To apply this result to our extrapolation scenario, consider an extrapolation problem $(\tilde{P},\delta,f)$ and choose a target set of datasets $T\subset \D^{N+1}$. Note that $\D^{N+1}$ has no interior, due to equality constraints of the form: 
\begin{equation}
\sum_aP(a|x,t_j)=\sum_aP(a|x',t_k).
\end{equation}
Let us then choose a lower dimensional representation $\D^{N+1}$ that is free from such linear constraints, e.g.: 
\begin{align}
&\left(P(a|x,t_j):a=1,...,A-1,\right.\nonumber\\
&\left.x=1,...,X, j=1,...,N+1\right)
\end{align}
We add an extra dimension to express normalization. We obtain the set of vectors of the form:
\begin{align}
&\mathbb{D}_{N+1}:=\left\{(1,\left(P(a|x,t_j):a=1,...,A-1,\right.\right.\nonumber\\
&\left.\left. x=1,...,X, j=1,...,N+1\right):P\in \D^{N+1}\right\}.
\end{align}
There exists an invertible linear transformation $C$ mapping $\D^{N+1}$ to $\mathbb{D}_{N+1}$. Let $Y$ denote the cone generated by the transformed target datasets $CT$ and define the vector $\hat{n}:=(1,0,...,0)$. Note that $C^{-1}\vec{y}\in T$ for any $\vec{y}\in Y$ with $\vec{y}\cdot\hat{n}=1$. We next define the vector $\vec{f}$ through the identity $\vec{f}\cdot \vec{y}=f(C^{-1}\vec{y})$. With these definitions, we account for the first two lines of problem (\ref{problem_abstract}). 

We still need to encode the constraints
\begin{equation}
\sum_a|P(a|x,t_j)-\tilde{P}(a|x,t_j)|\leq \delta(x,j),\forall x, j=1,...,N.
\end{equation}
These constraints are equivalent to
\begin{align}
&\sum_a g(a,x,j)P(a|x,t_j)\leq \delta(x,j)+\sum_a\tilde{P}(a|x,t_j)g(a,x,j),\nonumber\\
&\forall x, j=1,...,N,
\label{linearized_cons}
\end{align}
for all functions 
\begin{equation}
g:\{1,...,|A|\}\times\{1,...,|X|\}\times \{1,...,N\}\to\{-1,1\}.    
\end{equation}

In turn, those can be represented by the constraint
\begin{equation}
\mathbb{A}'\cdot{P}\leq \vec{b},
\end{equation}
where 
\begin{align}
&\mathbb{A}'_{(a,x,j,g),(a',x',k)}:=\delta_{a,a'}\delta_{x,x'}\delta_{jk}(1-\delta_{j,N+1})g(a,x,j),\nonumber\\
&b_{(x,j,g)}:=\delta(x,j)+\sum_a\tilde{P}(a|x,t_j)g(a,x,j).
\end{align}
Define
\begin{equation}
\mathbb{A}:=\mathbb{A}'C^{-1}.
\end{equation}
Then we have one-to-one mapping between all $\vec{P}\in T$ fitting the noisy dataset $(\tilde{P},\delta)$ and all vectors $\vec{y}\in Y$ satisfying $\vec{y}\cdot \hat{n}=1$ and $\mathbb{A} \vec{y} \leq \vec{b}$. We have just reformulated the extrapolation problem as an instance of problem (\ref{problem_abstract}).

To invoke the lemma, we just need to choose a norm $\|\|$. We take the norm defined by
\begin{equation}
\|\vec{s}\|:=\|C^{-1}\vec{s}\|',
\end{equation}
with $\|\|':=l_\infty(l_1(\R^{A})^{X(N+1)})$. That is:
\begin{equation}
\|q\|'=\max_{x,j}\sum_a|q(a|x,t_j)|.
\end{equation}
With this norm, 
\begin{equation}
\max_{\|\vec{v}\|\leq 1}\|\mathbb{A}\vec{v}\|_{\infty}=1.    
\end{equation}
In addition,
\begin{equation}
\max_{\|\vec{v}\|\leq 1}|\vec{f}\cdot \vec{v}|\leq \sum_{x}\max_{a}|f(a|x)|.
\end{equation}

We can reformulate our relaxation bounds as: for any $P^+\in T_m$, there exists $P\in T$ such that
\begin{equation}
\|P^+-P\|'\leq \epsilon_m.
\end{equation}

Consequently, if there exists $P^0\in T$ such that 
\begin{equation}
\delta(x,j)-\sum_{a}|P^0(a|x,t_j)-\tilde{P}(a|x,t_j)|\geq r>0,
\end{equation}
for all $x,j$, then we can invoke the lemma and conclude that
\begin{equation}
f(P^+)\geq f^\star-\epsilon_m\left(\frac{f(P^0(\tau))-f^\star}{r}+\sum_{x}\max_{a}|f(a|x)|\right).
\end{equation}
Replacing $f^\star$ by $\mu^-$, we arrive at the statement of Lemma \ref{lemma_T_m_approx}. The proof for $\mu^+$ is analogous.

\section{The limits of extrapolation}
\label{app:limits}
In this appendix, we set to prove the following two results:
\begin{prop}
\label{prop_nogo2}
Let $\tau>T$, and, under the hard energy constraint $\mbox{spec}(H)\subset \left[0,\frac{2}{T}\right]$, let the solution of the extrapolation problem $\mathbb{E}(N, T,\delta,\tau)$ be $[\frac{1}{2}-\delta',\frac{1}{2}+\delta']$. Then, for $\delta\ll \frac{T}{\tau}$, it holds that
\begin{equation}
\delta'\gtrsim \frac{1}{2}\delta\left(\frac{\tau}{T}\right)^{\frac{\log\left(\frac{1}{\delta}\right)}{\log\log\left(\frac{1}{\delta}\right)}}.
\end{equation}
\end{prop}

\begin{prop}
\label{prop_nogo1}
Assume the hard energy support constraint $\mbox{spec}(H)\subset [0,\frac{1}{T}]$. Then, for any $N\in\N$, the dataset $\mathbf{O}(N, T,\delta)$ has a Knightian uncertainty at $\tau>T$, unless $\delta\leq \frac{1}{g(\tau)}$, for some super-exponential function $g$, independent of $N$.
\end{prop}
Both proofs will rely on the following construction.
\begin{lemma}
\label{lemma_construction}
Given a set of non-zero real numbers $\{c_k\}_{k=0}^n$ and $\{E_k\}_{k=0}^n\subset \R$, there exists a timeline $P(t)$ with realization $(\C^{2n},\rho, H,M)$ such that $\mbox{spec}(H)=\{E_k\}_{k=1}^n\cup\{0\}$ and
\begin{equation}
\alpha(t):=P(1|t)-P(0|t)= \frac{\sum_kc_k\frac{e^{-iE_kt}+e^{iE_kt}}{2}}{\sum_l|c_l|}.
\end{equation}
\end{lemma}
\begin{proof}
Define $\ket{\phi(c)}:=\frac{1}{\sqrt{2}}(\ket{0}+\sign(c)\ket{1})$, and consider the quantum state $\rho\in B(\C^{2}\otimes \C^{n})$, given by
\begin{align}
&\rho=\frac{1}{\sum_l|c_l|}\sum_{k=0}^n|c_k|\proj{\phi(c_k)}\otimes \proj{k}.
\label{state_c}
\end{align}
and the Hamiltonian
\begin{equation}
H=\proj{1}\otimes \sum_kE_k\proj{k}.
\label{hamil_E}
\end{equation}
We also introduce the dichotomic operator
\begin{equation}
A:=\left(\ket{0}\bra{1}+\ket{1}\bra{0}\right)\otimes\id_n,
\label{dicho_gen}
\end{equation}
which defines the two-outcome POVM $M$ through $M_a:=\frac{1}{2}(\id+(-1)^aA)$. It can be verified that the timeline with realization $(\C^{2n},\rho, H,M)$ satisfies
\begin{equation}\label{eqn:trace1}
P(1|t)-p(0|t)=\tr(e^{-iHt}\rho e^{iHt} A)=\frac{\sum_kc_k\frac{e^{-iE_kt}+e^{iE_kt}}{2}}{\sum_l|c_l|}.
\end{equation}
    
\end{proof}

Now we are well equipped to prove Propositions \ref{prop_nogo2}, \ref{prop_nogo1}.

\begin{proof}[Proof of Proposition \ref{prop_nogo2}]

For $n=2m, m\in \N$, set 
\begin{align}
&E_k=\frac{n-2k}{nT},\nonumber\\
&c_k=\frac{1}{2^n}\left(\begin{array}{c}n\\k\end{array}\right)(-1)^k,
\end{align}
for $k=0,...,n$. Note that $\sum_k|c_k|=1$. By Lemma \ref{lemma_construction}, there exists a timeline $P_n(t)$ complying with the hard energy support constraint $\mbox{spec}(H)\subset \left[-\frac{1}{T},\frac{1}{T}\right]$ such that
\begin{align}
&P_n(1|t)-P_n(0|t)\nonumber\\
&=\frac{1}{2}\sum_{k=0}^n \frac{1}{2^n}\left(\begin{array}{c}n\\k\end{array}\right)\left((-e^{-\frac{it}{nT}})^k(e^{\frac{it}{nT}})^{n-k}\right.\nonumber\\
&\left. +(-e^{\frac{it}{nT}})^k(e^{-\frac{it}{nT}})^{n-k}\right)\nonumber\\
&=\frac{1}{2}\left(\sin\left(\frac{t}{nT}\right)^n+\sin\left(-\frac{t}{nT}\right)^n\right),\nonumber\\
&=\sin\left(\frac{t}{nT}\right)^n,
\label{asin}
\end{align}
for all $t$. In the last line, we invoked that $\sin(\bullet)$ is an odd function and that $n$ is even.

Note that the spectrum of $H$ is actually contained in $\left[-\frac{1}{T},\frac{1}{T}\right]$. This Hamiltonian generates exactly the same timeline as $H^+:=H+\frac{1}{T}$, whose spectrum is indeed contained in $\left[0,\frac{2}{T}\right]$. For mathematical convenience, though, we will continue the proof with $H$ instead of $H^+$.

\begin{equation}\label{asin2}
\overline{P}_n(1|t)-\overline{P}_n(0|t)=-\sin\left(\frac{t}{nT}\right)^n,
\end{equation}
for all $t$.


For $t\in [0,T]$, eq. (\ref{asin}) tell us that
\begin{align}
&\sum_a\left|P_n(a|t)-\frac{1}{2}\right|= \sin\left(\frac{1}{n}\right)^n,
\end{align}
Hence, we set 
\begin{equation}
\delta= \sin\left(\frac{1}{n}\right)^n.
\end{equation}

On the other hand, for a fixed extrapolation time $\tau>T$, by eq. (\ref{asin}) it holds that
\begin{equation}
\delta'\geq \left|P_n(0|t)-\frac{1}{2}\right|=\frac{1}{2}\sin\left(\frac{\tau}{nT}\right)^n.
\end{equation}
Recall that $n=2m$, so the right-hand side is non-negative.

For large $n$, we have 
\begin{equation}
\delta\approx \frac{1}{n^n}, \delta'\gtrsim \frac{1}{2}\left(\frac{\tau}{Tn}\right)^n.
\end{equation}
The first relation implies that $n\approx e^{W\left(\log\left(\frac{1}{\delta}\right)\right)}$, where $W$ denotes Lambert's $W$ function, i.e., $W(s)e^{W(s)}=s$. For large $s$, this function behaves as \cite{LambertW}:
\begin{equation}
W(s)\approx \log(s)-\log(\log(s))+o(1),
\end{equation}
which implies $n\approx \frac{\log\left(\frac{1}{\delta}\right)}{\log\log\left(\frac{1}{\delta}\right)}$. Hence,
\begin{equation}
\frac{\delta'}{\delta}\gtrsim \frac{1}{2}\left(\frac{\tau}{T}\right)^n\approx \frac{1}{2}\left(\frac{\tau}{T}\right)^{\frac{\log\left(\frac{1}{\delta}\right)}{\log\log\left(\frac{1}{\delta}\right)}}.
\end{equation}

\end{proof}

\begin{proof}[Proof of Proposition \ref{prop_nogo1}]

For arbitrary $\lambda>1$, $n\in\N$, define 
\begin{align}
&E_k=\frac{k}{nT},\nonumber\\
&c_k=\left(\begin{array}{c}n\\k\end{array}\right)(1-\lambda)^{n-k}\lambda^k,
\end{align}
for $k=0,...,n$. Invoking Lemma \ref{lemma_construction}, we have that there exists a timeline $P_{\lambda,n}(t)$, compatible with the hard energy constraint $\mbox{spec}(H)\subset \left[0,\frac{1}{T}\right]$, such that
\begin{align}
&P_{\lambda,n}(1|t)-P_{\lambda,n}(0|t)\nonumber\\
&=\frac{1}{2}\frac{\sum_{k=0}^n\left(\begin{array}{c}n\\k\end{array}\right)(1-\lambda)^{n-k}\left(\lambda^ke^{-it\frac{k}{TN}}+\lambda^ke^{it\frac{k}{TN}}\right)}{\sum_{j=0}^n\left(\begin{array}{c}n\\j\end{array}\right)\lambda^j(\lambda-1)^{n-j}}\nonumber\\
&=\frac{1}{2}\frac{(1-\lambda+\lambda e^{-i\frac{t}{nT}})^n+(1-\lambda+\lambda e^{i\frac{t}{nT}})^n}{(2\lambda-1)^n}.
\label{done_sum}
\end{align}

This can be rewritten as:
\begin{align}
&P_{\lambda,n}(1|t)-P_{\lambda,n}(0|t)=\frac{\cos\left(\frac{\lambda t}{T}\right)+\epsilon_n\left(\lambda,\frac{t}{T}\right)}{(2\lambda-1)^n},
\label{super_osc}
\end{align}
with
\begin{align}
&\epsilon_n\left(\lambda,\frac{t}{T}\right):=\frac{1}{2}\left[(1-\lambda+\lambda e^{-i\frac{t}{nT}})^n\right.\nonumber\\
&+\left.(1-\lambda+\lambda e^{i\frac{t}{nT}})^n\right]-\cos(\frac{\lambda t}{T}).
\label{eqn:deltanu}
\end{align}
As a consequence of the identity $e^x=\lim_{n\to\infty}(1+x/n)^n$, we have that $\lim_{n\to \infty}\epsilon_n\left(\lambda,\tilde{t}\right)=0$, for all $\tilde{t}\in\R$. 

Therefore, there exists a number $n_0(\lambda)\in\N$ such that, for $n\geq n_0(\lambda)$ and $t\in [0, T]$, it holds that
\begin{equation}\label{epsilon1}
|\epsilon_n\left(\lambda,\tilde{t}\right)|\leq 1,
\end{equation}
for $\tilde{t}\in[0,1]$. For $n\geq n_0(\lambda)$ we hence have that
\begin{align}
&\sum_a\left|P_{\lambda,n}(a|t)-\frac{1}{2}\right|=|P_{\lambda,n}(1|t)-P_{\lambda,n}(0|t)|\nonumber\\
&\leq \frac{2}{(2\lambda+1)^n}=:\delta_n(\lambda).
\end{align}
Then, for any $N$, the timeline $P_{\lambda,n}$ fits the noisy dataset $\mathbf{O}(N,T,\delta_n(\lambda))$, see eq. (\ref{def_dataset_O}).

Now, set $\tau=\pi n T$. Then, by eq. (\ref{done_sum}), we have that
\begin{equation}
P_{\lambda,n}(1|\tau)-P_{\lambda,n}(0|\tau)=\frac{(1-2\lambda)^n}{(2\lambda-1)^n}=(-1)^n.
\end{equation}
Namely, $P_{\lambda,n}(a|\tau)=\frac{(-1)^{a+n+1}+1}{2}$.
If we swap the POVM elements of the realization of $P_{\lambda,n}$, we obtain another timeline $\overline{P}_{\lambda,n}$, which also fits $\mathbf{O}(N,T,\delta_n(\lambda))$ and predicts that 
\begin{equation}
\overline{P}_{\lambda,n}(a|\tau)=\frac{(-1)^{a+n+1}+1}{2}.
\end{equation}

In sum, for $\delta=\delta_n(\lambda)$, $n\geq n(\lambda)$, there is Knightian uncertainty at $\tau=\pi n$ independently of the number $N$ of measurement times. Call $\delta^\star_n$ the minimum noise that predicts a Knightian uncertainty at $\tau=\pi n T$. We have proven that
\begin{align}
&\lim_{n\to\infty}\frac{\log(\delta^\star_n)}{n}\nonumber\\
&\leq \lim_{n\to\infty}\frac{\log(\delta_n(\lambda))}{n}\nonumber\\
&=-(2\lambda+1).
\end{align}
Since $\lambda>1$ is arbitrary, it follows that $\delta^\star_n$ is the inverse of a function $g(n)$, superexponential in $n$ (and hence, superexponential in $\tau$).

\end{proof}

\section{An aha! dataset generated by two-level systems}
\label{app:aha_2}
Let $A=2$ and consider the single-measurement dataset $\tilde{P}_1$:
\begin{equation}
\tilde{P}_1(a|0)=\tilde{P}_1\left(a|\frac{\pi}{E^+}\right)=\frac{1}{2},a=1,2,
\end{equation}
depicted in Figure \ref{fig:weird_effects} with two blue points.

\begin{figure}
    \centering
    \includegraphics[width=\linewidth]{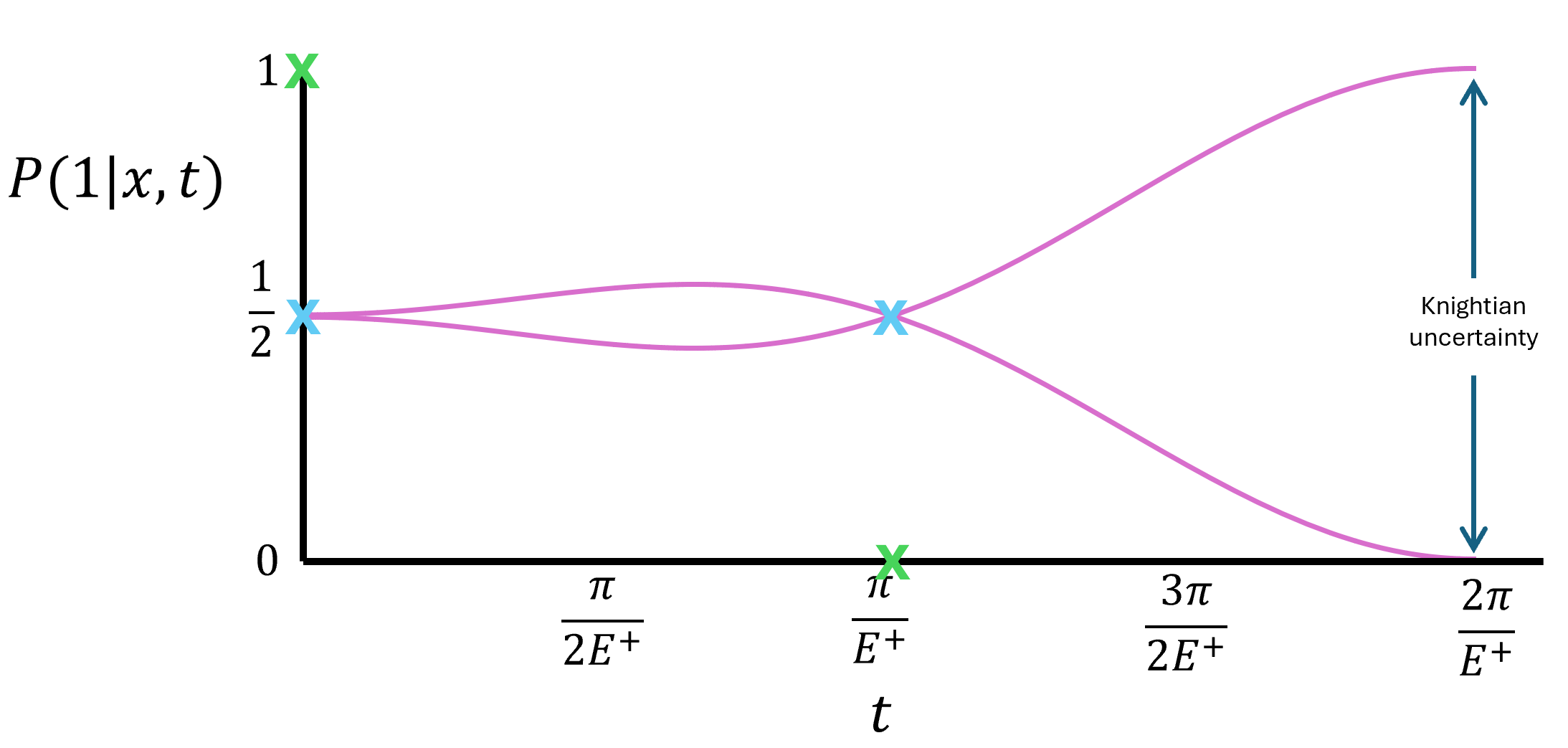}
    \caption{\textbf{An aha! dataset in two-level quantum systems.} The datasets $\tilde{P}_1,\tilde{P}_2$ are respectively depicted by blue and green crosses. Purple curves denote the timeline with realization $(\C^3,\ket{\bar{\psi}^3},\bar{H}^3,\bar{M})$, and its complement. $\tilde{P}_2$ is shown to be an aha! dataset for $\tilde{P}_1$ at $\tau=\frac{2\pi}{E^+}$.}
    \label{fig:weird_effects}
\end{figure}

Let $\tau:=\frac{2\pi}{E^+}$, and consider the realization $(\C^3,\ket{\bar{\psi}^3},\bar{H}^3,\bar{M})$, with
\begin{align}
\bar{M}_{1|1}=e^{-i\bar{H}^3\tau_2}\frac{7\id+9\proj{\bar{\psi}_3}}{16}e^{i\bar{H}^3\tau_2}.
\end{align}
The timeline with realization $(\C^3,\ket{\bar{\psi}^3},\bar{H}^3,\bar{M})$ fits $\tilde{P}_1$, predicting $P(a|1,\tau)=\delta_{a,1}$. The complementary timeline, obtained through the exchange $a=1\leftrightarrow a=2$, predicts the complementary distribution $\delta_{a,2}$. By taking convex combinations of these two timelines, we can obtain any datapoint at $\tau$, and so $\tilde{P}_1$ has Knightian uncertainty at time $\tau$. This is denoted in Figure \ref{fig:weird_effects} by a vertical purple dotted line.

Now, add the dataset $\tilde{P}_2(a|t_j):=\mathbf{D}_2(a|t_j)$, depicted in Figure \ref{fig:weird_effects} with green dots, and consider the full dataset $\tilde{P}$, with $\tilde{P}(a|x,t_j)=\tilde{P}_x(a|t_j)$, for $x=1,2$.

By Proposition \ref{prop_self_test}, any realization $(\H,\rho,H,M)$ of $\tilde{P}_2$ has the property that 
\begin{equation}
\rho(t)=V^\dagger e^{-i\bar{H}^2t}(\proj{\bar{\psi}^2}\otimes\gamma) e^{i\bar{H}^2t} V,    
\end{equation}
for some isometry $V$ and state $\gamma$. Since $e^{-i\bar{H}^2\tau}=1$, this means that $\rho(0)=\rho(\tau)$, and consequently all timelines $P$ fitting $\tilde{P}$ satisfy $P(a|1,\tau_2)=\tilde{P}(a|1,0)=\frac{1}{2}$. Thanks to the addition of the dataset $\tilde{P}_2$, we have hence gone from a scenario of total uncertainty to a scenario of total certainty. It follows that $\tilde{P}_2$ is an aha! dataset for $\tilde{P}_1$.

\section{Timelines fitting the datasets $\mathbf{A}_1, \mathbf{A}_2$}
\label{app:aha_3}
Consider the timeline $P^+$, with realization $(\C^3,\ket{\psi^+},H^+,M^+)$, where 
\begin{align}
&\ket{\psi^+}=\frac{1}{\sqrt{3}}(\ket{1}+\ket{2}+\ket{3}),\nonumber\\
&H^+=\frac{E^+}{2}\proj{2}+E^+\proj{3},\nonumber\\
&M^+_{1|1}=\proj{\psi^+}+\frac{1}{3}\proj{\psi^+(t_1)},\nonumber\\
&M^+_{1|2}=\id-\proj{\psi^+(t_1)}.
\end{align}
It can be verified that $P^+$ fits the dataset $\mathbf{A}$, defined by eqs. (\ref{aha_partial}), (\ref{aha_full}). In addition, $P^{+}(1|x,\tau)=1$, for $x=1,2$.

We next define, for $x=1,2$, the timeline $P^{-}_x$, with realization $(\C^4,\ket{\psi^{-}_x},H^{-}_x,M^{(-,x)})$, where
\begin{align}
&\ket{\psi^{-}_1}=\frac{1}{\sqrt{6}}\ket{1}+\frac{1}{\sqrt{3}}\ket{2}+\frac{1}{\sqrt{3}}\ket{3}+\frac{1}{\sqrt{6}}\ket{4},\nonumber\\
&\ket{\psi^{-}_2}=\frac{1}{2}(\ket{1}+\ket{2}+\ket{3}+\ket{4}),\nonumber\\
&H^{-}_1 =\frac{E^+}{4}\proj{2}+\frac{E^+}{2}\proj{3}+\frac{3E^+}{4}\proj{4},\nonumber\\
&H^{-}_2 =\frac{E^+}{2}\proj{2}+\frac{3E^+}{4}\proj{3}+E^+\proj{4},\nonumber\\
&M^{(-,1)}_1=\ket{\psi^{(-,1)}}\bra{\psi^{(-,1)}},\nonumber\\
&M^{(-,2)}_1=\proj{\psi^{(-,2)}}+\proj{\alpha},
\end{align}
with
\begin{equation}
    \ket{\alpha}:= \frac{1}{2}\ket{1}-\frac{1}{2}\ket{2}+\frac{1}{2}\ket{3}-\frac{1}{2}\ket{4}.
\end{equation}
It is immediate that, for $x=1,2$, the timeline $P^{-}_x$ fits $\mathbf{A}_x$. Moreover, $P^{-}_x(1|\tau)=0$, for $x=1,2$.

\section{Timelines fitting the dataset $\mathbf{F}$}
\label{app:fog}
Given a probability distribution $q=(q_1, q_2, q_3)$, consider the timeline $P_{q}$, generated by
\begin{align}
&\rho=\proj{\bar{\psi}^4},H=\bar{H}^4, \nonumber\\
&M_1=\proj{\bar{\psi}^4}+q_1\proj{\bar{\psi}^4(t_1)},\nonumber\\
&M_2=\proj{\bar{\psi}^4(t_2)}+q_2\proj{\bar{\psi}^4(t_1)}, \nonumber\\
&M_3=\proj{\bar{\psi}^4(t_3)}+q_3\proj{\bar{\psi}^4(t_1)},
\end{align}
where $\ket{\bar{\psi}}^4$ and $\bar{H}^4$ are those from \eqref{reference_realization}. 
Each such timeline fits the dataset $\mathbf{F}$ defined in eq. (\ref{fog_dataset}). Moreover, at $\tau_1:=t_1+\frac{6\pi}{E^+}$, $P_q$ predicts the datapoint $P(a|\tau_1)=q_a$, $a=1,2,3$.

The timelines $P^0, P^1, P^2$ alluded to in Section \ref{sec:fog} are defined as:
\begin{align}
&P^0=P_{(1,0,0)},\nonumber\\
&P^1=P_{(0,1,0)},\nonumber\\ 
&P^2=P_{(0,0,1)}.
\end{align}

\end{appendix}

\end{document}